\documentclass[11pt,a4paper,reqno]{amsart}%
\usepackage{amsthm,amsmath,amsfonts,amssymb,amsxtra,appendix,bookmark,dsfont,bbm,bm,graphicx}

\theoremstyle{plain}
\newtheorem{theorem}{Theorem}
\newtheorem{lemma}[theorem]{Lemma}
\newtheorem{corollary}[theorem]{Corollary}
\newtheorem{proposition}[theorem]{Proposition}

\theoremstyle{definition}

\usepackage[english]{babel}

\theoremstyle{remark}
\newtheorem{remark}[theorem]{Remark}

\numberwithin{equation}{section}

\def\bq{\begin{eqnarray}}
\def\eq{\end{eqnarray}}
\def\bqq{\begin{align*}}
\def\eqq{\end{align*}}
\def\ba{\begin{aligned}}
\def\ea{\end{aligned}}

\def\nn{\nonumber}

\renewcommand{\epsilon}{\varepsilon}

\def\cF {\mathcal{F}}

\def\R {\mathbb{R}}

\def\cD {\mathcal{D}}

\def\R {\mathbb{R}}

\title[The Bogoliubov free energy functional II.]{The Bogoliubov free energy functional II. The dilute limit}

\author[M. Napi\'orkowski]{Marcin Napi\'orkowski}
\address{Institute of Science and Technology Austria, Am Campus 1, 3400 Klosterneuburg, Austria \& \newline
Department of Mathematical Methods in Physics, Faculty of Physics, University of Warsaw, Pasteura 5, 02-093 Warsaw, Poland}
\email{marcin.napiorkowski@fuw.edu.pl}

\author[R. Reuvers]{Robin Reuvers}
\address{QMATH, Department of Mathematical Sciences, University of Copenhagen, Universitetsparken 5, DK-2100 Copenhagen \O, Denmark\newline
\tiny Present: DAMTP, Centre for Mathematical Sciences, University of Cambridge, Wilberforce Road, Cambridge CB3 0WA, United Kingdom} 
\email{r.reuvers@damtp.cam.ac.uk}

\author[J.~P. Solovej]{Jan Philip Solovej}
\address{QMATH, Department of Mathematical Sciences, University of Copenhagen, Universitetsparken 5, DK-2100 Copenhagen \O, Denmark} 
\email{solovej@math.ku.dk}

\begin{document}

\begin{abstract} 
We analyse the canonical Bogoliubov free energy functional in three dimensions at low temperatures in the dilute limit. We prove existence of a first-order phase transition and, in the limit $\int V\to 8\pi a$, we determine the critical temperature to be $T_{\rm{c}}=T_{\rm{fc}}(1+1.49\rho^{1/3}a)$ to leading order. Here, $T_{\rm{fc}}$ is the critical temperature of the free Bose gas, $\rho$ is the density of the gas and $a$ is the scattering length of the pair-interaction potential $V$. We also prove asymptotic expansions for the free energy. In particular, we recover the Lee--Huang--Yang formula in the limit $\int V\to 8\pi a$.

\end{abstract}
\maketitle
\tableofcontents

\section{Introduction}
\label{intro}
For a \textit{non-interacting}, or \textit{free}, Bose gas with density $\rho$, the textbook argument by Einstein shows that the phase transition to BEC happens at a critical temperature (in units $\hbar=2m=k_B=1$)
\begin{equation}
\label{Tfc}
T_{\rm fc}=4\pi\zeta(3/2)^{-2/3}\rho^{2/3}.
\end{equation}
How do interactions between the bosons affect this free critical temperature?
A system of particular interest is liquid helium, in which the nuclei interact rather strongly, and one can ask how Einstein's argument and the free critical temperature \eqref{Tfc} are altered by this interaction.
Feynman studied this problem with path integrals \cite{Feynman-53,Feynman-53.2}. Arguing that the potential resulted in an increased effective mass, he predicted that the critical temperature would decrease compared to the free case, which had indeed been observed for liquid helium. He did not make any quantitative predictions.

To make such quantitative predictions, various simplifications were considered. The first one is to replace the interaction potential for liquid helium by a hard-core potential with radius $a>0$
\begin{equation}
\label{hcp}
V(x)=
  \begin{cases}
    \begin{aligned}
       & \infty
    \end{aligned}           & |x|\leq a \\
    \ 0 & |x|>a
  \end{cases}.
\end{equation}
To simplify things further, it is common to study a weakly-interacting or dilute gas. For a hard-core potential, the natural length scale is given by the radius $a$. We could compare this length scale to the one defined by the density: $\rho^{-1/3}$, the average distance between the particles. Diluteness now means that the particles meet only rarely, that is, the average distance between the particles is much bigger than the length scale of the potential, or 
\begin{equation}
\label{dillim12}
\rho^{1/3}a\ll1.
\end{equation}

This assumption is not valid for liquid helium, but it is for experiments with trapped dilute cold gases such as \cite{Cornell-95, Ketterle-95}. In any case, one can repeat Feynman's question: how is the free critical temperature \eqref{Tfc} altered by the hard-core interaction? 

Lee and Yang were the first to study this \cite{LeeYan-58} in the translation-invariant case. They used pseudopotential methods developed in \cite{HuaYan-57,LeeHuaYan-57} to conclude that the shift in critical temperature should be proportional to $\rho^{1/3}a$. In the appendix of \cite{LeeYan-58}, they solve a simplified system, which gives
\begin{equation}
\label{exprTc}
T_{\rm{c}}=T_{\rm{fc}}(1+1.79(\rho^{1/3}a)+o(\rho^{1/3}a)).
\end{equation}

It is such an approximate expression that we will be looking for in this paper, but for a general class of potentials. To properly define the dilute limit \eqref{dillim12} without reference to a hard-core potential, we consider a characteristic length scale of the potential that is known as the \textit{scattering length} $a$ (see \cite{LieSeiSolYng-05} for a definition). It coincides with the core radius for the hard-core potential.

For general potentials, there has been a lot of debate about whether the linear dependence on $\rho^{1/3}a$ in \eqref{exprTc} is correct (\cite{GKW,Hua1,Hua2,Toyoda-82} predict exponents of 1/2, 3/2, 1/2 and 1/2, respectively, where the latter is the only one predicting a decrease in $T_{\rm c}$ compared to $T_{\rm fc}$). Nonetheless, \eqref{exprTc} is still expected to hold true, at least up to the value of the constant 1.79, which we discuss shortly.

It is good to remember that the search for \eqref{exprTc} for general potentials started from a desire to understand BEC in superfluid helium, but that particular problem remains intractable to this day. In its stead, the dilute setting has become a well-known and challenging object of study of its own. Indeed, the predicted critical temperature for a dilute gas \eqref{exprTc} is higher than $T_{\rm fc}$, whereas the critical temperature of liquid helium is lower, which shows that the systems are quite different. Nonetheless, we have little hope of understanding the strongly-interacting case if we cannot even treat this weakly-interacting set-up, justifying the attention this problem has received (see \cite{Andersen-04} for an overview).\\

We start from a Hamiltonian for a gas of $N$ bosons that interact via a (periodized) repulsive pair potential $V^l$ in a three-dimensional box $\left[-l/2,l/2\right]^3$ with periodic boundary conditions:
\[
H_N=\sum_{1\leq i\leq N}-\Delta^l_i+\sum_{1\leq i<j\leq N}V^l_{ij}.
\]
The particle density is $\rho=N/l^3$. Assuming the interaction only depends on the distance between the particles, $H_N$ is translation invariant, and we therefore write its second-quantized form in momentum space
\begin{equation}
\label{HN}
H=\sum_p p^2 a^\dagger_p a_p+\frac{1}{2l^3}\sum_{p,q,k} \widehat{V}^l(k) a_{p+k}^\dagger a_{q-k}^\dagger a_q a_p.
\end{equation}
Here, only particular $p$ are included in the sum, as determined by the size of the box $l$, but we will consider the thermodynamic limit $l\to\infty$.\\

To the best of our knowledge, the only rigorous fact known about the critical temperature for the Hamiltonian \eqref{HN} is the upper bound established by Seiringer and Ueltschi  using the Feynman--Kac formula \cite{SeiUel-09}. It is not surprising that such results are thin on the ground: it remains impossible to prove BEC in the dilute limit at positive temperature, let alone determine the critical point exactly. 

As for approximate models, we already mentioned Lee and Yang's expression \eqref{exprTc} for the hard-core gas \cite{LeeYan-58}. This expression can only be found in the appendix of their paper, perhaps because Lee and Yang considered their calculation to be physically inaccurate since it predicts a first---rather than the expected second---order phase transition. The fact that \eqref{exprTc} was hidden in the appendix has presumably led to the widespread misconception that Lee and Yang only predicted a shift linear in $\rho^{1/3}a$, without saying anything about the sign or size of the constant  \cite{Andersen-04,Baymetal-01,SeiUel-09,Smith}. Even if Lee and Yang themselves did not really trust their result, it fits reasonably well with numerics: Monte Carlo methods \cite{Arnold,Kash,NhoLan-04} suggest that the form \eqref{exprTc} is correct, but that the numerical value 1.79 should be closer to 1.3.

So how do Lee and Yang approach this problem? They replace the boundary conditions imposed by the hard-core potential by a pseudopotential that should give the right wave function in the physically relevant region where all the particles are at least distance $2a$ from one another \cite{HuaYan-57,LeeHuaYan-57}. They then assume that only s-wave scattering is important (i.e.\ the momentum of the particles is low), and show that replacing the potential by 
\[
8\pi a \delta(\boldsymbol{r})\partial_rr,
\]
should yield the correct wave function. For smooth functions, this is simply a multiplication by a delta function, but the derivative does play a role for physical wave functions. All this leads to an excitation spectrum of Bogoliubov form, which can now be used to calculate the shift in the critical temperature \eqref{exprTc}. 

Before we explain how this is done, let us point out that this claim
in itself has led to some confusion. In a number of articles in which
the dilute Bose gas is treated with field-theoretic
methods---e.g.\ Bijlsma and Stoof \cite{BijSto-96} and Baym et
al.\ \cite{Baymetal-01}, who find \eqref{exprTc} with constants of 4.7
and 2.9, respectively---it is claimed that mean-field theories such as
Bogoliubov's will simply give $T_{\rm c}=T_{\rm fc}$, or, in other
words, no shift. One argument \cite{Andersen-04} goes as follows: a
particle with momentum $p$ effectively has the energy
\begin{equation}
\label{wrongapprox}
\epsilon(p)\sim\sqrt{p^2(p^2+2\widehat{V}(p)\rho)}\approx p^2\sqrt{1+2\widehat{V}(0)\rho/p^2}\approx p^2+\widehat{V}(0)\rho,
\end{equation}
in which the reader can recognize an approximation to the Bogoliubov dispersion relation \cite{Bogoliubov-47b}.
Inserting this `mean-field' shift of the energy levels into the particle density of the free Bose gas gives
\begin{equation}
\label{freedens}
\frac{1}{e^{(p^2+\widehat{V}(0)\rho-\mu)/T}-1}
\end{equation}
so that the `critical'  $\mu$ is $\widehat{V}(0)\rho$. At this $\mu$, the relation between $T$ and $\rho$ is the same as for the free gas, and so the critical temperature does not change. However, one should be more careful in the comparison with the free gas, and the exact form of the dispersion relation one uses.

In Bogoliubov's analysis, the number of particles $N_0$ in the $p=0$ state enters via a c-number substitution and plays a crucial role. Dividing by the volume, we obtain a \textit{condensate density} $\rho_0=N_0/l^3$ that can now be regarded as a parameter. The dispersion relation Lee and Yang derive for the hard-core potential with radius $a$ is 
\begin{equation}
\label{Yangsp}
\epsilon(p)\sim\sqrt{p^2(p^2+16\pi a\rho_0)},
\end{equation}
so, unlike \eqref{wrongapprox}, this gives a $\rho_0$-dependence. Furthermore, we should not define $\mu$ using the free particle density \eqref{freedens}, which just happened to be the minimizer of the free energy in that case. Instead, for fixed $\rho$ and $\rho_0$, we should treat the remaining particles with density $\rho-\rho_0$ grand canonically, resulting in a grand canonical partition function that depends on $T$, $\rho$, $\rho_0$ and a chemical potential $\mu$. Recalling that there are only two independent parameters, one should now eliminate $\rho$ by calculating the value it takes at the minimum of the free energy for fixed $T$, $\rho_0$ and $\mu$, and then minimize over all $\rho_0$. The critical $\mu_{\rm c}$ for fixed temperature is the one where the minimizing $\rho_0$ changes from $\rho_0=0$ (no BEC) to $\rho_0>0$ (BEC). Note that this definition is far more complicated than the naive conclusion $\mu_{\rm c}=\widehat{V}(0)\rho$ above, but it is more correct. That was apparently clear to Lee and Yang, but it seems to have gone out of fashion, resulting in the false belief that the Bogoliubov spectrum does not give a change in the critical temperature.\\

In this work we consider a variational model 
introduced by Critchley and Solomon \cite{CriSol-76}. They evaluate
the expectation value of $H-TS-\mu\mathcal{N}$ in a quasi-free state,
where $S$ is the von Neumann entropy and $\mathcal{N}$ is the particle
number operator, and minimize over all quasi-free states, resulting in
an upper bound to the free energy at temperature $T$ and chemical
potential $\mu$.

This upper bound is well-motivated. The first supporting argument is that the usual treatment of the Hamiltonian \eqref{HN} with the Bogoliubov approximation \cite{Bogoliubov-47b} reduces it to an operator that is quadratic in the creation and annihilation operators, and that ground and Gibbs states of such operators are quasi-free states. A second is that quasi-free states have successfully served as trial states to establish correct bounds on the ground state energy of Bose gases \cite{ErdSchYau-08,GiuSei-09,Solovej-06}, which is of course the $T=0$ free energy.

Expressing the expectation value of $H-TS-\mu\mathcal{N}$ for a general quasi-free state does lead to a complicated non-linear functional. Simplifying it somewhat by throwing out certain terms, Critchley and Solomon conclude that the model will reproduce Bogoliubov's conclusions.

In this paper, we analyse their functional without the simplifications, and determine whether the minimizers display BEC ($\rho_0>0$) or not ($\rho_0=0$). This is a variational reformulation of Bogoliubov's and Lee and Yang's approach that is conceptually clear and more accurate, although it has in common with Lee and Yang's approach that the phase transition is of (presumably unphysical) first order because the density jumps at the critical temperature.

The analysis of the functional without simplifications leads to an approximation for the critical temperature given by
\begin{equation}
\label{predicted}
T_{\rm{c}}=T_{\rm{fc}}(1+1.49\rho^{1/3}a+o(\rho^{1/3}a)),
\end{equation}
in the limit $\int V=\widehat{V}(0)\to8\pi a$ (that is, for a sequence of potentials with $\widehat{V}(0)\to8\pi a$), and the constant 1.49 is indeed closer to the predicted 1.3  \cite{Arnold,Kash,NhoLan-04} than Lee and Yang's 1.79. The same analysis can also be carried out in 2 dimensions and we discuss this in \cite{NapReuSol-17}.\\

By its construction, this model also gives an upper bound to the free energy at positive temperature, which, for the full Hamiltonian \eqref{HN}, was so far only considered by Seiringer \cite{Sei-08} and Yin \cite{Yin-10}. At $T=0$, the free energy is simply the ground state energy, which we can compare with the prediction
\[
4\pi a \rho^2 +\frac{512}{15}\sqrt{\pi}(\rho a)^{5/2}+o((\rho a)^{5/2})
\]
by Lee, Huang and Yang \cite{LeeHuaYan-57}. Our model does reproduce the leading behaviour, but the second order only comes out correctly in the limit $\widehat{V}(0)\to8\pi a$. A similar result was obtained earlier by Erd\"os, Schlein and Yau \cite{ErdSchYau-08}, but the exact upper bound has in fact been proved by Yau and Yin \cite{YauYin-09}.\\

One could ask whether the predicted critical temperature shift \eqref{predicted} can actually be measured. For harmonic traps, a linear shift has indeed been measured \cite{Ensh,Gerb,Sm2}, but it cannot be compared with \eqref{predicted} since there is no translation invariance and the effect of the trap, expected to lower rather than raise the critical temperature, is simply too big. Recently, a BEC was also created in a uniform potential \cite{Gaunt}. The measurements are not precise enough, however, to measure the shift directly, but even if they were, in this set-up the finite size effects due to the boundedness of the trap are expected to be six times larger than the shift caused by the interaction. In the words of \cite{Smith}, `we are thus still lacking a direct measurement of the historically most debated [$T_{\rm c}$] shift'.\\

\textbf{Acknowledgements.} We thank Robert Seiringer and Daniel Ueltschi for bringing the issue of the change in critical temperature to our attention. We also thank the Erwin Schr\"odinger Institute (all authors) and the Department of Mathematics, University of Copenhagen (MN) for the hospitality during the period this work was carried out. We gratefully acknowledge the financial support by the European Union’s Seventh Framework Programme under the ERC Grant Agreement Nos. 321029 (JPS and RR) and 337603 (RR) as well as support by the VILLUM FONDEN via the QMATH Centre of Excellence (Grant No. 10059) (JPS and RR), by the National Science Center (NCN) under grant No. 2012/07/N/ST1/03185 and the Austrian Science Fund (FWF) through project Nr. P 27533-N27 (MN).

\section{The Bogoliubov free energy functional}
This article is the continuation of the previous work \cite{NapReuSol1-15}, in which we derive and analyse the \textit{Bogoliubov free energy functional} that was first introduced by Critchley and Solomon \cite{CriSol-76}. Let us briefly recall the set-up.

As motivated in the introduction, the functional is obtained from \eqref{HN} by substituting a c-number $\rho_0$ through $a_0\rightarrow a_0+\sqrt{l^3\rho_0}$ (justified in \cite{LieSeiYng-05}) and evaluating the expectation value of $H-TS-\mu\mathcal{N}$ of a quasi-free state. Assuming translation invariance and $\langle a_p a_{-p}\rangle=\langle a^\dagger_{-p} a^\dagger_p\rangle$, the two (real-valued) functions \mbox{$\gamma(p):=\langle a^\dagger_p a_p\rangle\geq0$} and \mbox{$\alpha(p):=\langle a_p a_{-p}\rangle$} fully determine this expectation value. Here, $\gamma(p)$ is the \textit{density of particles with momentum $p$}, and $\alpha$ describes the \textit{pairing in the system}. The c-number $\rho_0\geq0$ should be thought of as the \textit{density of the condensate}, so that there is a Bose--Einstein condensate (BEC) if $\rho_0>0$. The total particle density is
\bq
\rho=\rho_0+(2\pi)^{-3}\int_{\mathbb{R}^3}\gamma(p)dp=:\rho_0+\rho_\gamma. \nn
\eq

In the thermodynamic limit, this gives the (grand canonical) Bogoliubov free energy functional
\bq
\label{freeenergyfunctional}
\begin{aligned}
\mathcal{F}(\gamma,\alpha,\rho_{0})&= (2\pi)^{-3}\int_{\mathbb{R}^3} p^{2}\gamma(p)dp-\mu\rho-TS(\gamma,\alpha)+\frac{\widehat{V}(0)}{2}\rho^{2} \\
&\quad+\rho_{0}(2\pi)^{-3}\int_{\mathbb{R}^3}\widehat{V}(p)\left(\gamma(p)+\alpha(p)\right)dp.\\
&\quad+\frac{1}{2}(2\pi)^{-6}\iint_{\mathbb{R}^3\times \mathbb{R}^3}\widehat{V}(p-q)\left(\alpha(p)\alpha(q)+\gamma(p)\gamma(q)\right)dpdq,
\end{aligned}
\eq
with entropy
\[
\begin{aligned}
&S(\gamma,\alpha)= \quad(2\pi)^{-3}\int_{\mathbb{R}^3}s(\gamma(p),\alpha(p))dp\quad=\quad(2\pi)^{-3}\int_{\mathbb{R}^3}s(\beta(p))dp \\
&=(2\pi)^{-3}\int_{\mathbb{R}^3}\left[\left(\beta(p)+\frac{1}{2}\right)\ln\left(\beta(p)+\frac{1}{2}\right)-\left(\beta(p)-\frac{1}{2}\right)\ln\left(\beta(p)-\frac{1}{2}\right)\right]dp,
\end{aligned}
\]
where
\bq
\beta(p):=\sqrt{\left(\frac{1}{2}+\gamma(p)\right)^{2}-\alpha(p)^{2}}.\label{def:beta}
\eq
The functional is defined on the domain $\mathcal{D}$ given by
\bq
\mathcal{D}=\{(\gamma,\alpha,\rho_{0})\ |\ \gamma \in L^{1}((1+p^{2})dp),\ \gamma\geq0,\ \alpha(p)^{2}\leq\gamma(1+\gamma),\ \rho_{0}\geq 0\}.
\nn
\eq
To reiterate, this functional describes the grand canonical free energy of a homogeneous Bose gas at temperature $T\geq0$ and chemical potential $\mu\in \mathbb{R}$ in the thermodynamic limit.\\

The goal of the first paper \cite{NapReuSol1-15} is twofold: to establish the existence of minimizers for the minimization problem
\begin{align}
\label{gcmin}
F(T,\mu)=\inf_{(\gamma,\alpha,\rho_{0})\in \mathcal{D}}\mathcal{F}(\gamma,\alpha,\rho_{0}),
\end{align}
and to analyse their structure (in whether $\rho_0>0$ or not) for different temperatures and chemical potentials. Keeping in mind that the dilute limit $\rho^{1/3}a\ll1$ is defined in terms of the density, the canonical counterparts to \eqref{freeenergyfunctional} and \eqref{gcmin} are considered as well: the functional $\mathcal{F}^{\rm can}=\mathcal{F}+\mu\rho$ at density $\rho\geq0$ and temperature $T\geq0$ is given by 
\bq
\label{freeenergyfunctional2}
\begin{aligned}
\mathcal{F}^{\rm{can}}(\gamma,\alpha,\rho_{0})&= (2\pi)^{-3}\int_{\mathbb{R}^3} p^{2}\gamma(p)dp-TS(\gamma,\alpha)+\frac12\widehat{V}(0)\rho^{2}\\ &+(2\pi)^{-3}\rho_{0}\int_{\mathbb{R}^3}\widehat{V}(p)\left(\gamma(p)+\alpha(p)\right)dp \\
&+(2\pi)^{-6}\frac12\iint_{\mathbb{R}^3\times \mathbb{R}^3}\widehat{V}(p-q)\left(\alpha(p)\alpha(q)+\gamma(p)\gamma(q)\right)dpdq,
\end{aligned}
\eq
with $\rho_0=\rho-\rho_\gamma$. The canonical minimization problem is
\bq
 \label{cmin}
\begin{aligned}
F^{\rm{can}}(T,\rho)&=\inf_{\substack{({\gamma},{\alpha},\rho_{0}=\rho-\rho_\gamma)\in\cD\\
    }}\cF^{\rm{can}}(\gamma,\alpha,\rho_0)=\inf_{0\leq\rho_0\leq\rho}f(\rho-\rho_0,\rho_0),
\end{aligned}
\eq
where
\[
f(\lambda,\rho_0)=\inf_{\substack{({\gamma},{\alpha})\in\cD'\\
      \int\gamma=\lambda
    }}\cF^{\rm{can}}(\gamma,\alpha,\rho_0)
\]
and 
\[
\mathcal{D}'=\{(\gamma,\alpha)\ |\ \gamma\in L^1((1+p^2)dp),\ \gamma(p)\geq0,\ \alpha(p)^2\leq \gamma(p)(\gamma(p)+1)\}.
\]
Strictly
speaking, this is not really a canonical formulation: it is only the
expectation value of the number of particles that we fix. We will
nevertheless describe this energy as `canonical'. The function
$F(T,\mu)$ as a function of $\mu$ is the Legendre transform of the
function $F^{\rm{can}}(T,\rho)$ as a function of $\rho$. 

The main results of \cite{NapReuSol1-15}, which we recall in the next section, state that there exist minimizers for both \eqref{gcmin} and \eqref{cmin} and that both models exhibit a BEC phase transition.

\section{Existence of minimizers and phase transition}
\label{sec:previousresults}
The following results, proven in the accompanying paper \cite{NapReuSol1-15}, provide the basis for any further analysis of the Bogoliubov free energy functional.

Throughout this
article, we assume that the two-body interaction potential and
its Fourier transform 
\[
\widehat{V}(p)=\int_{\R^3}V(x)e^{-ipx}dx, \hspace{1cm} V(x)=(2\pi)^{-3}\int_{\R^3}\widehat{V}(p)e^{ipx}dp
\]
are radial functions that satisfy
\begin{eqnarray}\label{positiveinteraction}
  V\geq 0,\quad \widehat{V}\geq0,\quad V\not\equiv0.
\end{eqnarray}
Moreover, we assume that 
\bq\label{interactionassumptions}
\widehat{V}\in C^{1}(\mathbb{R}^3),\ \widehat{V}\in L^1(\mathbb{R}^3),\ \|\widehat{V}\|_{\infty} <\infty,\ \|\nabla\widehat{V}\|_2<\infty,\ \|\nabla\widehat{V}\|_{\infty}<\infty.
\eq

\begin{theorem}[Existence of grand canonical minimizers for $T>0$]\label{thm:existencepositiveT}
  Let $T>0$. Assume the interaction potential is a radial function that satisfies \eqref{positiveinteraction} and
  \eqref{interactionassumptions}. Then there exists a minimizer for
  the Bogoliubov free energy functional
  \eqref{freeenergyfunctional} defined on $\mathcal{D}$.
\end{theorem} 

It turns out that we need to assume some additional
regularity on the interaction potential to prove a similar statement for $T=0$.

\begin{theorem}[Existence of grand canonical minimizers for $T=0$]\label{thm:existencezeroT}
  Assume the interaction potential fulfils the assumptions of Theorem
  \ref{thm:existencepositiveT}. If we assume in addition that
  $\widehat{V}\in C^{3}(\mathbb{R}^3)$ and that all derivatives of
  $\widehat{V}$ up to third order are bounded, then there exists a
  minimizer for the Bogoliubov free energy functional
  \eqref{freeenergyfunctional} defined on $\mathcal{D}$
  for $T=0$.
\end{theorem}

We would like to stress that the minimizers need not be unique. In fact, we will see (cf.\ Remark \ref{nonuniqueminimizer}) that there exist combinations of $\mu$ and $T$ for which the problem \eqref{gcmin} has two minimizers with two different densities. 

We have analogous results in the canonical setting.

\begin{theorem}[Existence of canonical minimizers for
  $T>0$]\label{thm:existencecanonicalepositiveT}
  Let $T>0$. Assume the interaction potential is a radial function that satisfies \eqref{positiveinteraction} and
  \eqref{interactionassumptions}. Then the variational problem
  \eqref{cmin} admits a minimizer.
\end{theorem}

\begin{theorem}[Existence of canonical minimizers for $T=0$]\label{thm:existencecanonicalzeroT}
  Assume the interaction potential fulfils the assumptions of Theorem
  \ref{thm:existencecanonicalepositiveT}. If we assume in addition
  that $\widehat{V}\in C^{3}(\mathbb{R}^3)$ and that all derivatives of
  $\widehat{V}$ up to third order are bounded, then there exists a
  minimizer for the canonical minimization problem
  \eqref{cmin} at $T=0$.
\end{theorem}

Let us now recall the results concerning the existence of phase transitions in our model.  Our first result
shows that Bose--Einstein Condensation and pairing are connected in these models.

\begin{theorem}
\label{thm:BECvsSF}
  Let
  $(\gamma,\alpha,\rho_0)$ be a minimizing triple for
 either \eqref{freeenergyfunctional} or \eqref{freeenergyfunctional2}. Then \bq \rho_0 =0
  \Longleftrightarrow \alpha \equiv 0. \nn \eq
\end{theorem}

Thus, there can only be one kind of phase transition, and the next results show that it indeed exists.

\begin{theorem}[Existence of grand canonical phase transition]\label{thm:phasetrangrandcan}
  Given $\mu>0$.  Then there exist temperatures $0<T_1<T_2$ such that
  a minimizing triple $(\gamma,\alpha,\rho_0)$ of
  \eqref{gcmin} satisfies
\begin{enumerate}
\item $\rho_0=0$ for $T\geq T_2$;
\item $\rho_0>0$ for $0\leq T \leq T_1$.
\end{enumerate}
\end{theorem} 

\begin{theorem}[Existence of canonical phase transition]\label{thm:phasetrancan}
For fixed  $\rho>0$ there exist temperatures $0<T_3<T_4$ such that a minimizing triple
$(\gamma,\alpha,\rho_0)$ of \eqref{cmin} satisfies 
\begin{enumerate}
\item $\rho_0=0$ for $T\geq T_4$;
\item $\rho_0>0$ for $0\leq T \leq T_3$.
\end{enumerate}
\end{theorem}

\section{Main results and sketch of proof}
\label{mr}
We assume that
\begin{eqnarray}\label{dillim}
\rho^{1/3}a\ll  1,
\end{eqnarray}
where $a$, the scattering length of the potential, is defined by
\[
4\pi a:=\int \Delta w=\frac12 \int Vw, 
\]
and $w$ satisfies
\bq
\label{def:scatteq}
-\Delta w+\frac12 Vw=0
\eq
in the sense of distributions with $w(x)\to1$ as $|x|\to\infty$. The quantity $8\pi a$ is often replaced by $\int V=\widehat{V}(0)$, which is its first-order Born approximation. In fact, $\widehat{V}(0)> 8\pi a$ (see \cite[Appendix C]{LieSeiSolYng-05} for more details). We quantify this discrepancy with the parameter $\nu=\widehat{V}(0)/a$, so that $\nu>8\pi$. The limit $\nu\to8\pi$, that is, a sequence of potentials such that $\widehat{V}(0)$ tends to $8\pi a$, is of special interest.\\ 

For the proofs, it will sometimes be useful to consider the region $T\leq D\rho^{2/3}$ with $D>1$ fixed separately,  in which case we can rewrite the second condition in \eqref{dillim} as 
\bq
\label{eq:thermalwavelengthcondition}
\sqrt{T}a\leq \sqrt{D} \rho^{1/3}a\ll  1.
\eq
In particular, since the thermal wavelength $\Lambda\sim \sqrt{T}^{-1}$, the condition \eqref{eq:thermalwavelengthcondition} implies that $a/\Lambda \ll  1$. Furthermore, we define a constant $C$ by
\begin{equation}
\label{assumptionsV}
\int\widehat{V}\leq Ca^{-2}\hspace{1cm}\text{and}\hspace{1cm}\|\partial^n\widehat{V}\|_{\infty}\leq C a^{n+1}\text{\ for\ } 0\leq n\leq3,
\end{equation}
where $\partial^n$ is shorthand for all $n$-th order partial derivatives. With this definition, our estimates depend only on $C$ and not on $a$. Throughout the paper, we will also use $C$ to denote any unspecified positive constant.

\subsection{The critical temperature} The following theorems contain information about the critical temperature in the dilute limit.

We derive expressions in terms of both $\widehat{V}(0)\geq 8\pi a$ and the scattering length $a$, but these simplify if we consider a sequence of potentials with $\widehat{V}(0)\to8\pi a$, or $\nu\to8\pi$, so that we can determine the value of the constants in that limit numerically.

Note that $T_{\rm{fc}}=c_0\rho^{2/3}$  is the critical temperature of the free Bose gas, and $\rho_{\rm fc}=(T/c_0)^{3/2}$ its corresponding critical density.

\begin{theorem}[Canonical critical temperature]
\label{thm:cancrittemp}
There is a monotone increasing function $h_1:(8\pi,\infty)\to\mathbb{R}$ such that for any minimizing triple $(\gamma,\alpha,\rho_0)$ of \eqref{cmin} at temperature $T$ and density $\rho$
\begin{enumerate}
\item $\rho_0\neq 0$ if $T<T_{\rm fc}\left(1+h_1(\nu)\rho^{1/3}a+o(\rho^{1/3}a)\right)$
\item $\rho_0= 0$ if $T>T_{\rm fc}\left(1+h_1(\nu)\rho^{1/3}a+o(\rho^{1/3}a)\right)$.
\end{enumerate}
The numerical value of $\lim_{\nu\to8\pi}h_1(\nu)$ is $1.49$.
\end{theorem}

\begin{theorem}[Grand-canonical critical temperature]
\label{thm:grandcancrittemp}
There is a function $h_2:(8\pi,\infty)\to\mathbb{R}$ such that for any minimizing triple $(\gamma,\alpha,\rho_0)$ of \eqref{gcmin} at temperature $T$ and chemical potential $\mu>0$
\begin{enumerate}
\item $\rho_0\neq 0$ if $T<\left(\frac{\sqrt{\pi}}{2\zeta(3/2)}\frac{8\pi}{\nu}\right)^{2/3}\left(\frac{\mu}{a}\right)^{2/3}+h_2(\nu)\mu+o(\mu)$
\item $\rho_0= 0$ if $T>\left(\frac{\sqrt{\pi}}{2\zeta(3/2)}\frac{8\pi}{\nu}\right)^{2/3}\left(\frac{\mu}{a}\right)^{2/3}+h_2(\nu)\mu+o(\mu)$.
\end{enumerate}
The numerical value of $\lim_{\nu\to8\pi}h_2(\nu)$ is $0.44$.
\end{theorem}

\subsection{Free energy expansion}
The second main result of this paper provides an expansion of the free energy \eqref{cmin} in the dilute limit.
We first define the integrals that play a central role in our analysis:
\begin{equation}
\label{eq:integrals}
\begin{aligned}
I_1(d,\sigma,\theta)&=(2\pi)^{-3}\int\Big[\sqrt{(p^2+d)^2+2(p^2+d)(1+\theta)\sigma} \\
&\hspace{3cm} -(p^2+d+(1+\theta)\sigma)+\frac{((1+\theta)\sigma)^2}{2p^2}\Big]dp\\
I_2(d,\sigma,\theta,s)&=(2\pi)^{-3}\int\ln\left(1-e^{-\sqrt{(p^2+ds^2)^2+2(p^2+ds^2)(1+\theta)\sigma s^2}}\right)dp\\
I_3(d,\sigma,\theta)&=(2\pi)^{-3}\int\left(\frac{p^2+d+(1+\theta)\sigma}{\sqrt{(p^2+d)^2+2(p^2+d)(1+\theta)\sigma}}-1\right)dp\\
I_4(d,\sigma,\theta,s)&=(2\pi)^{-3}\int\left(e^{\sqrt{(p^2+ds^2)^2+2(p^2+ds^2)(1+\theta)\sigma s^2}}-1\right)^{-1} \\
 &\qquad \qquad \qquad \times \frac{p^2+ds^2+(1+\theta)\sigma s^2}{\sqrt{(p^2+ds^2)^2+2(p^2+ds^2)(1+\theta)\sigma s^2}}dp.
\end{aligned}
\end{equation}
We will consider $d,\sigma,s\geq0$, and $-1\leq\theta\leq0$. For the following theorems, it suffices to set $\theta=0$ and $\sigma=8\pi$. 
The general form will, however, be needed to study the critical temperature.

\begin{theorem}[Canonical free energy expansion]
\label{thm:canfreeenexp}
Assume that $T$ and $\rho$ satisfy the conditions $\eqref{dillim}$ and \eqref{eq:thermalwavelengthcondition}. We then have the following expressions for the canonical free energy  \eqref{cmin}.
\begin{enumerate}
\item For $T>T_{\rm fc}\left(1+h_1(\nu)\rho^{1/3}a+o(\rho^{1/3}a)\right)$, the free energy is 
$$F^{\rm{can}}(T,\rho)=F_0(T,\rho)+\widehat{V}(0)\rho^2+O((\rho a)^{5/2}),$$
and we have $\rho_\gamma=\rho$, $\rho_0=0$ for the minimizer. Here $F_0(T,\rho)$ is the free energy of the non-interacting gas (cf. \eqref{eq:freefreeenergy}).
\item For $T<T_{\rm fc}\left(1+h_1(\nu)\rho^{1/3}a+o(\rho^{1/3}a)\right)$, there exists a universal constant $d_0>0$ such that the free energy is 
\begin{equation*}
\begin{aligned}
F^{\rm{can}}(T,\rho)=\inf_{0\leq d\leq d_0}&[\frac12(\rho a)^{5/2}I_1(d,8\pi,0)+T^{5/2}I_2(d,8\pi,0,\sqrt{\rho_0(d) a/T})\\
&-d\rho_0(d)a(\rho-\rho_0(d))\\
&+\widehat{V}(0)\rho^2-8\pi a\rho_0(d)\rho+\rho_0(d)^2(12\pi a-\widehat{V}(0))]\\
&+o\left(T(\rho a)^{3/2}+(\rho a)^{5/2}\right),
\end{aligned}
\end{equation*}
where
\[
\rho_0(d):=\rho-\frac12(\rho a)^{3/2}I_3(d,8\pi,0)-T^{3/2}I_4(d,8\pi,0,\sqrt{(\rho-\rho_{\rm fc}) a/T}).
\]
\end{enumerate}
\end{theorem}

In fact, we will obtain a more precise energy expansion in the region around the critical temperature. 

The expression for the free energy above involves integrals and a minimization problem in the parameter $d$. If we also assume that $\rho a/T\ll1$, we can  simplify the result, as the following theorem shows.

\begin{theorem}[The canonical free energy for $\rho a/T\ll1$]
\label{thm:expint} 
Let $\Delta\rho=\rho-\rho_{\rm fc}$. For $\rho a/T\ll 1$ and $T<T_{\rm fc}\left(1+h_1(\nu)\rho^{1/3}a+o(\rho^{1/3}a)\right)$, the canonical free energy is given by
\[
\begin{aligned}
F^{\rm{can}}(T,\rho)&=T^{5/2}f_{\rm{min}}+4\pi a \rho^2+(\nu-4\pi)a\rho_{\rm fc}(2\rho-\rho_{\rm fc})\\
&+\left(\frac{\Delta\rho a}{T}\right)^{3/2}\left(-\frac1{3\sqrt2\pi}\right)\left(\nu^{3/2}+(\nu-8\pi)^{3/2}\right)T^{5/2}\\
&+o(T(\rho a)^{3/2}).
\end{aligned}
\]
\end{theorem}

In the case $\rho a/T\gg1$, we can also simplify the expression in the second point of Theorem \ref{thm:canfreeenexp}: the contribution from the integrals $I_2$ and $I_4$ can be neglected in the minimization problem.

\begin{corollary}[The canonical energy for $\rho a/T\gg 1$ ]
\label{corollary112}
For $\rho a/T\gg 1$ and $T<T_{\rm fc}\left(1+h_1(\nu)\rho^{1/3}a+o(\rho^{1/3}a)\right)$, the canonical free energy can be described in terms of a function $g:(8\pi,\infty)\to\mathbb{R}$ as
\[
F^{\rm{can}}(T,\rho)=4\pi a \rho^2+g(\nu)(\rho a)^{5/2}+o\left((\rho a\right)^{5/2}),
\]
with $g(\nu)\to\frac{512}{15}\sqrt{\pi}$ as $\nu\to8\pi$. The latter result is known as the Lee--Huang--Yang formula.
\end{corollary}
Before we proceed to the proof of these theorems, let us sketch the main ideas used in the paper. 

\subsection{Set-up of the paper}
Since the Euler--Lagrange equations of the free energy functional involve the convolutions $\widehat{V}\ast\gamma$ and $\widehat{V}\ast\alpha$, it is very hard to analyse them quantitatively. Even with a Fourier transform, they cannot be solved. The main idea is to replace the non-local terms in the functional by local ones, such that we end up with a simplified functional that can be minimized explicitly, that is, 
\[
\inf_{\substack{\text{$(\gamma,\alpha,\rho_0)$}\\\text{$\rho_0+\rho_\gamma=\rho$}}}
\mathcal{F}^{\rm can}\approx\inf_{\substack{\text{$(\gamma,\alpha,\rho_0)$}\\\text{$\rho_0+\rho_\gamma=\rho$}}}
\mathcal{F}^{\rm sim}=\inf_{0\leq\rho_0\leq\rho}\Big[\inf_{\substack{\text{$(\gamma,\alpha)$}\\\text{$\rho_\gamma=\rho-\rho_0$}}}\mathcal{F}^{\rm sim}\Big]
\]
where the final minimizations can be done explicitly.
 
The approximation involves several steps. First, we replace the convolution term involving $\gamma$ with $\widehat{V}(0)\rho_\gamma^2$. We expect that the particles interact weakly in the dilute limit and it seems reasonable to assume that the system will behave like a free Bose gas to leading order. We therefore expect that the minimizing $\gamma$ is concentrated on a ball of radius $\sqrt{T}$. By our assumptions \eqref{assumptionsV}, $\widehat{V}(p)$ is approximately $\widehat{V}(0)$ on a ball of radius $a^{-1}\gg\sqrt{T}$ (in the region around the critical temperature), justifying the replacement. 

Second, by introducing a trial function $\alpha_0$, we rewrite the convolution terms involving $\alpha$. This trial function will be expressed in terms of $\widehat{Vw}$, where $w$ is the solution to the scattering equation. Finally, we will also substitute $\widehat{V}$ by $\widehat{Vw}$ in the terms that are linear in $\gamma$ at the cost of a small error. All this will be done in Subsection \ref{dersimpfunc}, with Lemma \ref{prop:simplandfullfreeenergycomp} specifying the error terms exactly.

We then minimize the simplified functional. We split the minimization in two steps: first one over $\gamma$ and $\alpha$ with the constraint that $\rho_0+\rho_\gamma=\rho$, followed by a minimization over $0\leq\rho_0\leq\rho$. The first step will be carried out in Subsection \ref{constrminga}, and it will lead to a useful class of minimizers $(\gamma^{\rho_0,\delta},\alpha^{\rho_0,\delta})$. To prepare for the final minimization over $\rho_0$, we will establish further properties of these functions in Subsection \ref{sec:prelapprox}.

In order to prove that this provides a good approximation, we will need to know that the error terms are small for both the minimizer of the full functional and the minimizer of the simplified functional. For the full functional, this is shown in Subsections \ref{apriorifree} and \ref{aprioridilute} along with several other useful a priori estimates. 
 
In Subsection \ref{sec:crittemp}, we will analyse the energy in the region $|\rho-\rho_{\rm{fc}}|\leq C\rho(\rho^{1/3}a)$, since the a priori result of Subsection \ref{furtherapriori} shows that this is where the phase transition occurs. This leads to the calculation of the critical temperature and the proof of Theorems \ref{thm:cancrittemp} and \ref{thm:grandcancrittemp}. 
 
Subsection \ref{subsectionenexp} contains the proof of Theorems \ref{thm:canfreeenexp} and \ref{thm:expint}.

\section{Proof of the main results}
\subsection{Derivation of the simplified functional}
\label{dersimpfunc}
The following simplified functional will serve as an approximation to the canonical free energy functional \eqref{freeenergyfunctional2}:
\begin{eqnarray}
\begin{aligned}
\label{eq:simplifiedfunctional}
  \mathcal{F}^{\rm{sim}}(\gamma,\alpha,\rho_0)&=(2\pi)^{-3}\int \left(p^2+(\rho_0+t_0) \widehat{Vw}(p)\right)\gamma(p)dp
  \\ \qquad &+(2\pi)^{-3}\int(\rho_0+t_0)\widehat{Vw}(p)\alpha(p)dp-TS(\gamma,\alpha)
  \\ \qquad &+\frac14 (2\pi)^{-3}(\rho_0+t_0)^2 \int \frac{\widehat{Vw}(p)^2}{p^2}dp
  \\ \qquad &+\widehat{V}(0)\rho^2+(12\pi a-\widehat{V}(0))\rho_0^2-8\pi a\rho\rho_0
  \\ \qquad &-4\pi a t_0^2-8\pi a t_0(\rho-\rho_0).
\end{aligned}
\end{eqnarray}
Here, $w$ satisfies the scattering equation \eqref{def:scatteq}, and $t_0$ is a parameter that could in principle be chosen to depend on $\rho$ and $\rho_0$. This will turn out to be necessary for the proof of Theorems \ref{thm:cancrittemp} and \ref{thm:grandcancrittemp} in Subsection \ref{sec:crittemp}, and we will state a specific choice for $t_0$ at the start of this subsection. For the proof of Theorems \ref{thm:canfreeenexp} and \ref{thm:expint} in Subsection \ref{subsectionenexp} it will, however, suffice to set $t_0=0$. Before we make a choice for $t_0$, we will work with the general assumption
\begin{equation}
\label{assumpt0}
-\rho_0\leq t_0\leq 0.
\end{equation}

Note that $\cF^{\rm{sim}}$ consist of terms that are both linear and local in $\gamma$ and $\alpha$ (aside from the entropy), and it will therefore be much easier to handle than the full $\mathcal{F}^{\rm{can}}$. 

As shown in Lemma \ref{prop:simplandfullfreeenergycomp}, the difference between $\mathcal{F}^{\rm{sim}}$ and $\mathcal{F}^{\rm can}$ can be expressed in terms of
\bq
\begin{aligned}
E_1(\gamma,\alpha,\rho_0)&:=(2\pi)^{-6}\frac12\iint(\alpha-\alpha_0)(p)\widehat{V}(p-q)(\alpha-\alpha_0)(q)dpdq\\
E_2(\gamma,\alpha,\rho_0)&:=\left|(2\pi)^{-3}\rho_0\int\gamma(p)\widehat{V}(p)dp-\widehat{V}(0)\rho_0\rho_\gamma\right|\\
E_3(\gamma,\alpha,\rho_0)&:=\left|(2\pi)^{-3}\rho_0\int\gamma(p)\widehat{Vw}(p)dp-\widehat{Vw}(0)\rho_0\rho_\gamma\right|\\
E_4(\gamma,\alpha,\rho_0)&:=\left|(2\pi)^{-6}\frac12\iint\gamma(p)\widehat{V}(p-q)\gamma(q)dpdq- \frac12\widehat{V}(0)\rho_\gamma^2\right|.
\end{aligned} \label{def:errorterms}
\eq
Here, the function $\alpha_0$ is chosen to be 
\bq
\alpha_0:=(\rho_0+t_0)\widehat{w}-(2\pi)^3 \rho_0\delta_0=(2\pi)^3 t_0\delta_0-\frac{\rho_0+t_0}{2}\frac{\widehat{Vw}(p)}{p^2}, \label{def:alpha_0}
\eq 
where we have used the Fourier transform of the scattering equation \eqref{def:scatteq}, taking into account the boundary condition:
\bq
\widehat{w}=(2\pi)^3 \delta_0-\frac12\frac{\widehat{Vw}(p)}{p^2}. \label{eq:distrscattsol}
\eq
When a more precise error is required, we will consider
\begin{equation}
\begin{aligned}
E_5(\gamma,\alpha,\rho_0)&:=\Big|(2\pi)^{-6}\frac12\iint\gamma(p)\widehat{V}(p-q)\gamma(q)dpdq \\
&\hspace{2cm}-\frac12\widehat{V}(0)\rho_\gamma^2-\frac{\zeta(3/2)\zeta(5/2)}{256\pi^3}\Delta\widehat{V}(0)T^4\Big|.
\end{aligned} \label{def:E_5}
\end{equation}
Note that the additional term in $E_5$ compared to $E_4$ is independent of $(\gamma,\alpha,\rho_0)$ and including it in the simplified functional will therefore not affect the minimizer (see Corollary \ref{lem:differencecriticalregion}).

The function $\widehat{Vw}$ appears in our definition of $\alpha_0$. It will turn out to be convenient to gather some of its properties before we prove the main result of this section. First of all, $w\geq0$, which implies that $Vw\geq0$, and so
\[
|\widehat{Vw}(p)|\leq\widehat{Vw}(0)=8\pi a.
\]
From \eqref{eq:distrscattsol}, we obtain
\begin{equation}
\label{someeq924}
  \int Vw^2=8\pi a -\frac12(2\pi)^{-3}\int \widehat{Vw}(p)^2|p|^{-2}dp,
\end{equation}
and hence the integral on the left-hand side is bounded by $Ca$. This implies
\[
\int|\widehat{Vw}|^2=(2\pi)^3\int|Vw|^2\leq C\|V\|_{\infty}\int Vw^2\leq \frac{C}{a},
\]
where we have used our assumptions \eqref{assumptionsV}. 
Using the above conclusions, we now estimate 
\[
\begin{aligned}
\left\|\frac{\widehat{Vw}}{p^2}\right\|_1&\leq\int_{|p|\leq a^{-1}}\frac{|\widehat{Vw}|}{p^2}dp+\int_{|p|> a^{-1}}\frac{|\widehat{Vw}|}{p^2}dp\\
&\leq C+\left(\int_{|p|> a^{-1}}|\widehat{Vw}|^2dp\right)^{1/2}\left(\int_{|p|> a^{-1}}\frac{1}{p^4}dp\right)^{1/2}\leq C,
\end{aligned}
\]
where it is important that the estimate is independent of $a$.
Applying \eqref{eq:distrscattsol} again, we have
\[
\widehat{Vw}=\widehat{V}-(2\pi)^{-3}\frac{\widehat{Vw}}{2p^2}\ast\widehat{V}.
\] 
By our assumptions \eqref{assumptionsV} we have for $0\leq n\leq3$:
\begin{equation}
\label{sameest}
\|\partial^n\widehat{Vw}\|_\infty\leq\|\partial^n\widehat{V}\|_\infty\left(1+C\left\|\frac{\widehat{Vw}}{2p^2}\right\|_1\right)\leq Ca^{n+1}.
\end{equation}
We can therefore estimate derivatives of $\widehat{Vw}$ in the same way as those of $\widehat{V}$, and we will use this in the subsections below.\\

The main result of this subsection is the following lemma, which compares the simplified and canonical free energy functionals.
Its message is that, given that the error terms are small for the minimizers of both the simplified and the full functional, it suffices to analyse the simplified functional.

\begin{lemma}
\label{prop:simplandfullfreeenergycomp}
For any triple $(\gamma,\alpha,\rho_0)$ we have
\begin{equation*}
\begin{aligned}
-\Big(E_2+E_3+ E_4\Big)(\gamma,\alpha,\rho_0)\leq\mathcal{F}^{\rm{can}}&(\gamma,\alpha,\rho_{0})-\mathcal{F}^{\rm{sim}}(\gamma,\alpha,\rho_0)\\
&\leq \Big( E_1+E_2+E_3+E_4\Big)(\gamma,\alpha,\rho_0).
\end{aligned}
\end{equation*}
\end{lemma}
\begin{proof}
We have
\bq
\ba \label{eq:differnceFsimFcan1}
\mathcal{F}^{\rm{can}}(\gamma,\alpha,\rho_0)&-\mathcal{F}^{\rm{sim}}(\gamma,\alpha,\rho_0)=(2\pi)^{-3}\rho_0\int\left(\widehat{V}(p)-\widehat{Vw}(p)\right)(\gamma(p)+\alpha(p))dp \\&-(2\pi)^{-3}t_0\int\widehat{Vw}(p)(\gamma(p)+\alpha(p))dp +\frac12\widehat{V}(0)\rho_\gamma^2-\frac12\widehat{V}(0)\rho^2\\
&+\frac12(2\pi)^{-6}\iint\gamma(p)\widehat{V}(p-q)\gamma(q)dpdq-\frac12\widehat{V}(0)\rho_\gamma^2 
\\&-\frac14 (2\pi)^{-3}(\rho_0+t_0)^2 \int \frac{\widehat{Vw}(p)^2}{p^2}dp -(12\pi a-\widehat{V}(0))\rho_0^2
  \\ &+8\pi a\rho\rho_0 +4\pi a t_0^2+8\pi a t_0(\rho-\rho_0)+E_1(\gamma,\alpha,\rho_0)
  \\&+(2\pi)^{-6} \int\alpha(p)(\widehat{V}*\alpha_0)(p)dp-\frac12(2\pi)^{-6}\int\alpha_0(p)(\widehat{V}*\alpha_0)(p)dp.
\ea
\eq
We start by dealing with the last two terms in \eqref{eq:differnceFsimFcan1}. First we have 
\[
(2\pi)^{-3}\widehat{V}*\alpha_0(p)=(\rho_0+t_0)\widehat{Vw}(p)-\rho_0\widehat{V}(p),
\]
which follows immediately from the definition \eqref{def:alpha_0}. This means that the first term in the last line of \eqref{eq:differnceFsimFcan1} cancels the $\alpha$-terms in the first two lines of  \eqref{eq:differnceFsimFcan1}. We thus have
\bq
\ba \label{eq:differnceFsimFcan2}
\mathcal{F}^{\rm{can}}(\gamma,\alpha,\rho_0)&-\mathcal{F}^{\rm{sim}}(\gamma,\alpha,\rho_0)=(2\pi)^{-3}\rho_0\int\left(\widehat{V}(p)-\widehat{Vw}(p)\right)\gamma(p)dp \\&-(2\pi)^{-3}t_0\int\widehat{Vw}(p)\gamma(p)dp +\frac12\widehat{V}(0)\rho_\gamma^2-\frac12\widehat{V}(0)\rho^2\\
&+\frac12(2\pi)^{-6}\iint\gamma(p)\widehat{V}(p-q)\gamma(q)dpdq-\frac12\widehat{V}(0)\rho_\gamma^2 
\\&-\frac14 (2\pi)^{-3}(\rho_0+t_0)^2 \int \frac{\widehat{Vw}(p)^2}{p^2}dp -(12\pi a-\widehat{V}(0))\rho_0^2
\\&+8\pi a\rho\rho_0 +4\pi a t_0^2+8\pi a t_0(\rho-\rho_0)+ E_1(\gamma,\alpha,\rho_0)
  \\&-\frac12(2\pi)^{-6}\int\alpha_0(p)(\widehat{V}*\alpha_0)(p)dp.
\ea
\eq
We now deal with the last term in the above equation.  Using \eqref{def:alpha_0}, we have
\bq
\begin{aligned}
\int\alpha_0(\widehat{V}*\alpha_0)&=\iint\frac{(\rho_0+t_0)\widehat{Vw}(p)}{2p^2}\widehat{V}(p-q)\frac{(\rho_0+t_0)\widehat{Vw}(q)}{2q^2}dpdq \\  \label{eq:alpha0Vaplhpa0}
&\quad +(2\pi)^6 t_0^2\widehat{V}(0)-2t_0(2\pi)^3\int\frac{(\rho_0+t_0)\widehat{Vw}(p)\widehat{V}(p)}{2p^2}dp.
\end{aligned}
\eq
Note that
\[
\frac12\int Vw^2=\int Vw-\frac12\int V+\frac12\int V(1-w)^2=8\pi a-\frac12\int V+\frac12\int V(1-w)^2,
\]
so that  
\bq
\begin{aligned}
&(2\pi)^{-6}\frac12\iint \widehat{(1-w)}(p)\widehat{V}(p-q)\widehat{(1-w)}(q)dpdq\\
&\hspace{3cm}=\frac12\int V-4\pi a-\frac14(2\pi)^{-3}\int \widehat{Vw}(p)^2|p|^{-2}dp. \label{eq:V_1_w_1_w}
\end{aligned}
\eq
These identities together with \eqref{someeq924} allow us to compute the terms in \eqref{eq:alpha0Vaplhpa0}. By \eqref{eq:distrscattsol}, we have
$$\frac{\widehat{Vw}(p)}{2p^2}=\widehat{(1-w)}(p),$$
so that it follows from \eqref{eq:V_1_w_1_w} that
\bq
\begin{aligned}
  &\frac12(2\pi)^{-6}\iint\frac{(\rho_0+t_0)\widehat{Vw}(p)}{2p^2}\widehat{V}(p-q)\frac{(\rho_0+t_0)\widehat{Vw}(q)}{2q^2}dpdq= \\
 &\,=\frac12(\rho_0+t_0)^2\widehat{V}(0)-4\pi a(\rho_0+t_0)^2 -\frac14(2\pi)^{-3}(\rho_0+t_0)^2\int \widehat{Vw}(p)^2|p|^{-2}dp.\nn
\end{aligned}
\eq
Furthermore,
\bq
\int\frac{\widehat{Vw}(p)\widehat{V}(p)}{2p^2}=\int(\widehat{1-w})\widehat{V}=(2\pi)^3\int V(1-w)=(2\pi)^3(\widehat{V}(0)-8\pi a). \nn
\eq
Collecting all terms we obtain
\bq
\begin{aligned}
\frac12(2\pi)^{-6}&\int\alpha_0(\widehat{V}*\alpha_0)=\frac12(\rho_0+t_0)^2\widehat{V}(0)-\frac{(\rho_0+t_0)^2}{4(2\pi)^{3}}\int \frac{\widehat{Vw}(p)^2}{p^2}dp \\&-4\pi a(\rho_0+t_0)^2 - t_0(t_0+\rho_0)(\widehat{V}(0)-8\pi a)+\frac12 t_0^2\widehat{V}(0) \\
&= \frac12 (\widehat{V}(0)-8\pi a)\rho_0^2+4\pi a t_0^2-\frac{(\rho_0+t_0)^2}{4(2\pi)^{3}}\int \frac{\widehat{Vw}(p)^2}{p^2}dp,
\end{aligned}\label{eq:alpha0alpha0result}
\eq
and inserting \eqref{eq:alpha0alpha0result} into \eqref{eq:differnceFsimFcan2} gives
\bq
\label{someequ}
\ba 
\mathcal{F}^{\rm{can}}(\gamma,\alpha,\rho_0)&-\mathcal{F}^{\rm{sim}}(\gamma,\alpha,\rho_0)=(2\pi)^{-3}\rho_0\int\widehat{V}(p)\gamma(p)dp -\widehat{V}(0)\rho_0\rho_\gamma
\\&-(2\pi)^{-3}(\rho_0+t_0)\int\widehat{Vw}(p)\gamma(p)dp+(\rho_0+t_0)8\pi a \rho_\gamma
\\&+\frac12(2\pi)^{-6}\iint\gamma(p)\widehat{V}(p-q)\gamma(q)dpdq-\frac12\widehat{V}(0)\rho_\gamma^2 
\\&+E_1(\gamma,\alpha,\rho_0).
\ea
\eq
Here, we added and subtracted $\widehat{V}(0)\rho_0\rho_\gamma$ and $8\pi a \rho_\gamma(\rho_0+t_0)$ and used that $\widehat{Vw}(0)=8\pi a$. Using the definitions \eqref{def:errorterms}, our assumption \eqref{assumpt0}, and the fact that $E_1\geq 0$ we arrive at the desired result.
\end{proof}

\begin{corollary}
\label{lem:differencecriticalregion}
For any triple $(\gamma,\alpha,\rho_0)$ we have
\begin{equation*}
\begin{aligned}
-\Big(E_2+E_3+E_5\Big)&(\gamma,\alpha,\rho_0)\\
&\leq\left(\mathcal{F}^{\rm{can}}-\mathcal{F}^{\rm{sim}}\right)(\gamma,\alpha,\rho_0)-\frac{\zeta(3/2)\zeta(5/2)}{256\pi^3}\Delta\widehat{V}(0)T^4\\
&\hspace{3cm}\leq \Big(E_1+E_2+E_3+E_5\Big)(\gamma,\alpha,\rho_0).
\end{aligned}
\end{equation*}
\end{corollary}

\subsection{Minimization of the simplified functional in $\gamma$ and $\alpha$}
\label{constrminga}
We will now find the minimizers of the simplified functional \eqref{eq:simplifiedfunctional}. We note that the minimization problem can be rewritten as 
\[
\begin{aligned}
&\inf_{(\gamma,\alpha,\rho_0),\ \rho_\gamma+\rho_0=\rho}\mathcal{F}^{\rm sim}(\gamma,\alpha,\rho_0)
=\inf_{0\leq\rho_0\leq\rho}\Big[\inf_{(\gamma,\alpha),\ \rho_\gamma=\rho-\rho_0}\mathcal{F}^{\rm s}(\gamma,\alpha,\rho_0)\\
&+\widehat{V}(0)\rho^2+(12\pi a-\widehat{V}(0))\rho_0^2-8\pi a\rho\rho_0-4\pi a t_0^2-8\pi a t_0(\rho-\rho_0)\Big],
\end{aligned}
\]
with
\begin{equation}
\label{simplfunc}
\begin{aligned}
\cF^{\rm s}(\gamma,\alpha,\rho_0)&=
(2\pi)^{-3}\int (p^2+(\rho_0+t_0) \widehat{Vw}(p))\gamma(p)dp \\
& \quad+(2\pi)^{-3}(\rho_0+t_0)\int \widehat{Vw}(p)\alpha(p)dp -TS(\gamma,\alpha)
\\&\quad+\frac14 (2\pi)^{-3}(\rho_0+t_0)^2 \int \frac{\widehat{Vw}(p)^2}{p^2}dp.
\end{aligned}
\end{equation}
This suggests that we first focus on the minimization problem 
\[
\inf_{(\gamma,\alpha),\ \rho_\gamma=\rho-\rho_0}\mathcal{F}^{\rm s}(\gamma,\alpha,\rho_0).
\]

Since $\cF^{\rm s}$ is convex in $\gamma$ and $\alpha$, we can enforce the constraint $\rho_\gamma=\rho-\rho_0$ using a Lagrange multiplier $\delta$. Recall that
\bq
\beta(p)=\sqrt{\left(\frac{1}{2}+\gamma(p)\right)^{2}-\alpha(p)^{2}}, \nn
\eq
and define
\begin{eqnarray*}
G(p)&=&T^{-1}\sqrt{(p^2+\delta+(\rho_0+t_0)
  \widehat{Vw}(p))^2-((\rho_0+t_0) \widehat{Vw}(p))^2}\\&=&
T^{-1}\sqrt{(p^2+\delta)^2+
  2(p^2+\delta)(\rho_0+t_0)\widehat{Vw}(p)}.
\end{eqnarray*}
The following result states the minimizers of the minimization problem for $\delta\geq0$.

\begin{lemma}[Simplified functional solution]
\label{prop:simplfunctsol}
Let $\delta\geq0$, $\rho_0\geq0$ and $-\rho_0\leq t_0 \leq 0$. The minimizer of 
\[
\inf_{(\gamma,\alpha)}\left[\mathcal{F}^{\rm s}(\gamma,\alpha,\rho_0)+\delta\int\gamma\right]
\]
is given by
\bq
\ba
\gamma^{\rho_0,\delta}&=\frac{\beta}{TG}(p^2+\delta+(\rho_0+t_0) \widehat{Vw}(p))-\frac12 \\
\alpha^{\rho_0,\delta}&=-\frac{\beta}{TG}(\rho_0+t_0)\widehat{Vw}(p),
\ea \nn
\eq
with $\beta$ and $G$ as above, and the minimum is
\[
\ba
&\cF^{\rm s}(\gamma^{\rho_0,\delta},\alpha^{\rho_0,\delta},\rho_0)+\delta\int\gamma^{\rho_0,\delta}\\
&=(2\pi)^{-3} T\int\ln(1-e^{-G(p)})dp\\
&\quad+(2\pi)^{-3} \frac12\int\Big[\sqrt{(p^2+\delta)^2+2(p^2+\delta)(\rho_0+t_0) \widehat{Vw}(p)}\\
&\hspace{2cm}-(p^2+\delta+(\rho_0+t_0) \widehat{Vw}(p))+\frac12(\rho_0+t_0)^2  \frac{\widehat{Vw}(p)^2}{p^2}\Big]dp.
\ea 
\]
\end{lemma}

\begin{proof}
Since
$$
s'(\beta)=\ln\left(\frac{\beta+\frac12}{\beta-\frac12}\right),
$$
we find the 
the Euler--Lagrange equations to be 
\begin{eqnarray}
\begin{aligned}\label{eq:simplifiedEL}
  p^2+\delta+(\rho_0+t_0) \widehat{Vw}(p)&=
  T\ln\left(\frac{\beta+\frac12}{\beta-\frac12}\right)\frac{\gamma(p)+\frac12}{\beta(p)}\\
    (\rho_0+t_0) \widehat{Vw}(p)&=-T\ln\left(\frac{\beta+\frac12}{\beta-\frac12}\right)\frac{\alpha(p)}{\beta(p)}.
    \end{aligned}
\end{eqnarray}
Squaring and subtracting both equations and using \eqref{def:beta} we obtain
\bq
\ln\left(\frac{\beta+\frac12}{\beta-\frac12}\right)=G(p),\quad 
\beta(p)=\left(e^{G(p)}-1\right)^{-1}+\frac12. \label{eq:betaGrelations}
\eq
One may be concerned about the square root in the definition of $G$. However, using $\widehat{Vw}(0)=8\pi a$, $\widehat{Vw}'(0)=0$ and $\|\widehat{Vw}''\|_\infty\leq Ca^3$, we note that 
\[
\widehat{Vw}(p)\geq 8\pi a- Ca^3p^2. 
\]
We find that $p^2+2(\rho_0+t_0)\widehat{Vw}(p)\geq Cp^2$ for all $p$. Together with $\delta\geq0$, this implies
\[
\ba
&(p^2+\delta)^2+2(p^2+\delta)(\rho_0+t_0)\widehat{Vw}(p)\\
&\hspace{3cm}= (p^2+\delta)\left(p^2+\delta+2(\rho_0+t_0)\widehat{Vw}(p)\right)\geq Cp^4.
\ea
\]
In particular, this means there are no problems with the square root.

Using \eqref{eq:betaGrelations} in \eqref{eq:simplifiedEL} we find for the minimizers
\begin{eqnarray}
\begin{aligned}
  \gamma(p)&=\frac{\beta}{TG}(p^2+\delta+(\rho_0+t_0) \widehat{Vw}(p))-\frac12\\
&=(e^{G(p)}-1)^{-1}\frac{p^2+\delta+(\rho_0+t_0) \widehat{Vw}(p)}{\sqrt{(p^2+\delta)^2+
      2(p^2+\delta)(\rho_0+t_0)\widehat{Vw}(p)}}\\
&\hspace{1cm}+
  \frac12\left(\frac{p^2+\delta+(\rho_0+t_0) \widehat{Vw}(p)}{\sqrt{(p^2+\delta)^2+
        2(p^2+\delta)(\rho_0+t_0)\widehat{Vw}(p)}}-1\right)\\
  \alpha(p)&=-\frac{\beta}{TG}(\rho_0+t_0)\widehat{Vw}(p)\\
&=-
  \left((e^{G(p)}-1)^{-1}+\frac12\right)\frac{(\rho_0+t_0)\widehat{Vw}(p)}{\sqrt{(p^2+\delta)^2+
      2(p^2+\delta)(\rho_0+t_0)\widehat{Vw}(p)}}.
\nn
\end{aligned}
\end{eqnarray}
These indeed satisfy $\alpha^2\leq\gamma(\gamma+1)$. Inserting them into the functional we obtain 
\[
\ba
&(p^2+\delta+(\rho_0+t_0) \widehat{Vw}(p))\gamma(p)+(\rho_0+t_0)\widehat{Vw}(p)\alpha(p)-Ts(\beta(p))\\
&=\frac{\beta(p)}{TG(p)}(TG(p))^2- \frac12(p^2+\delta+(\rho_0+t_0) \widehat{Vw}(p))\\
&\hspace{4cm}+T\beta(p)\ln\left(\frac{\beta(p)-\frac12}{\beta(p)+\frac12}\right)-\frac12T\ln\left(\beta(p)^2-\frac14\right)\\
&=-\frac12(p^2+\delta+(\rho_0+t_0) \widehat{Vw}(p))+\frac12TG(p)+T\ln(1-e^{-G(p)})\\
&=T\ln(1-e^{-G(p)})\\
&+\frac12\left(\sqrt{(p^2+\delta)^2+2(p^2+\delta)(\rho_0+t_0) \widehat{Vw}(p)}-(p^2+\delta+(\rho_0+t_0) \widehat{Vw}(p))\right),
\ea
\]
which gives the right expression. 
\end{proof}

We summarize and rewrite the relevant quantities in the following corollary. The expressions may seem a bit involved, but it will turn out to be useful to write them in this way.

\begin{corollary}
\label{energyexpressions}
Let $-1\leq\theta\leq0$, $d\geq0$, $\sigma\geq0$, $\phi>0$ and $0\leq\rho_0\leq\rho$ be fixed. Assume $\rho_0a/\phi^2=\sigma/8\pi$ and let $\delta=d\phi^2$, and $t_0=\theta\rho_0$. We then have
\[
\ba
&\mathcal{F}^{\rm s}(\gamma^{\rho_0,\delta},\alpha^{\rho_0,\delta},\rho_0)\\
&\quad\quad=(2\pi)^{-3}\phi^5\frac12 \int\Big[\sqrt{(p^2+d)^2+2(p^2+d)(1+\theta)\sigma \frac{\widehat{Vw}(\phi p)}{8\pi a}} \\
&\hspace{3cm} -(p^2+d+(1+\theta)\sigma \frac{\widehat{Vw}(\phi p)}{8\pi a})+\frac{((1+\theta)\sigma \frac{\widehat{Vw}(\phi p)}{8\pi a})^2}{2p^2}\Big]dp\\
&\quad\quad\quad+(2\pi)^{-3}T\phi^3\int\ln\left(1-e^{-\frac{\phi^2}{T}\sqrt{(p^2+d)^2+2(p^2+d)(1+\theta)\sigma \frac{\widehat{Vw}(\phi p)}{8\pi a}}}\right)dp-d\phi^2\rho_{\gamma^{\rho_0,\delta}}\\
&\quad\quad=:F^{(1)}+F^{(2)}-d\phi^2\rho_{\gamma^{\rho_0,\delta}},
\ea
\]
where
\begin{equation}
\begin{aligned}
\label{eq:rhogammasimplified}
\rho_{\gamma^{\rho_0,\delta}}&=(2\pi)^{-3}\phi^3\frac12\int\left(\frac{p^2+d+(1+\theta)\sigma \frac{\widehat{Vw}(\phi p)}{8\pi a}}{\sqrt{(p^2+d)^2+2(p^2+d)(1+\theta)\sigma \frac{\widehat{Vw}(\phi p)}{8\pi a}}}-1\right)dp\\
&\quad+(2\pi)^{-3} \phi^{3}\int\left(e^{\frac{\phi^2}{T}\sqrt{(p^2+d)^2+2(p^2+d)(1+\theta)\sigma \frac{\widehat{Vw}(\phi p)}{8\pi a}}}-1\right)^{-1} \\
&\hspace{4cm}\times \frac{p^2+d+(1+\theta)\sigma \frac{\widehat{Vw}(\phi p)}{8\pi a}}{\sqrt{(p^2+d)^2+2(p^2+d)(1+\theta)\sigma\frac{\widehat{Vw}(\phi p)}{8\pi a}}}dp\\
&=: \rho_{\gamma}^{(1)}+\rho_{\gamma}^{(2)}.
\end{aligned}
\end{equation}
\end{corollary}

In the above, $\phi$ may seem superfluous, but we will use it later to allow for different scalings: we either choose $\phi=Ta$ or $\phi=\sqrt{\rho_0 a}$. This allows us to choose the parameters $\sigma,d$ and $\theta$ to be of order 1 in the different regimes.

\subsection{A priori estimates on the free Bose gas}
\label{apriorifree}
To establish that the error terms in Lemma \ref{prop:simplandfullfreeenergycomp} are small for the minimizer of the full functional, we need a priori estimates, which we will prove in the next subsection. To prepare for this, we prove some facts about the free Bose gas first.

Let $\gamma_{\mu(\rho)}$ denote the minimizer with density $\rho$ for the 
free gas functional
$$
\cF_{0}(\gamma)=(2\pi)^{-3}\int p^2\gamma(p)-Ts(\gamma(p),0) dp.
$$ 
More precisely, $\mu(\rho)\leq 0$ represents the chemical potential
such that $\gamma_{\mu(\rho)}$ actually minimizes
$\cF_{0}(\gamma)-\mu(\rho)(2\pi)^{-3} \int\gamma$. If $\rho>\rho_{\rm fc}$ there is no minimizer with $(2\pi)^{-3} \int\gamma=\rho$ and $\mu(\rho)=0$,
i.e.\ we have the global free minimizer $\gamma_0$ with
$(2\pi)^{-3} \int\gamma_0=\rho_{\rm fc}$.  We denote the minimizing energy
$F_0(T,\rho)=\cF_0(\gamma_{\mu(\rho)})$. 
The minimizer $\gamma_{\mu}$ is given by 
\begin{equation}
\label{gammazero}
\gamma_{\mu}(p)=\frac1{e^{(p^2-\mu)/T}-1},
\end{equation}
hence 
\begin{equation*}
\rho=(2\pi)^{-3}T^{3/2}\int\left[e^{(p^2-T^{-1}\mu(\rho))}-1\right]^{-1}dp,
\end{equation*}
and the energy is 
\bq
\begin{aligned}
F_0(T,\rho)=&\ (2\pi)^{-3}T\int\ln\left(1-e^{-(p^2-\mu(\rho))/T}\right)dp 
+\mu(\rho)\rho \\
=&\ (2\pi)^{-3}T^{5/2}\int\ln\left(1-e^{-(p^2-T^{-1}\mu(\rho))}\right)dp 
+\mu(\rho)\rho.\label{eq:freefreeenergy}
\end{aligned}
\eq
We see that we have the following scalings for $F_0$ and $\mu$:
\begin{equation*}
F_0(T,\rho)=T^{5/2}f_0\left(\rho/T^{3/2}\right),\qquad
\mu(\rho)=Tm\left(\rho/T^{3/2}\right),
\end{equation*}
where $f_0$ and $m$ are the functions independent of $T$ given by
\begin{equation*}
\begin{aligned}
f_0(n)&=(2\pi)^{-3}\int\ln\left(1-e^{-(p^2-m(n))}\right)dp 
+m(n)n,\\
n&=(2\pi)^{-3}\int\left[e^{p^2-m(n)}-1\right]^{-1}dp.
\end{aligned}
\end{equation*}
The critical density is $\rho_{\rm fc}=T^{3/2}n_{\rm fc}$, where 
\begin{equation}
\label{constantfcdensity}
n_{\rm fc}=(2\pi)^{-3}\int\left[e^{p^2}-1\right]^{-1}dp=\left(8\pi^{3/2}\right)^{-1}\zeta( 3/2 ).
\end{equation}
The minimal free energy is $\min_\rho F_0(T,\rho)=T^{5/2}f_{\rm min}$, where 
\begin{align*}
f_{\rm min}&=(2\pi)^{-3}\int\ln\left(1-e^{-p^2}\right)dp\\
&=-\frac23(2\pi)^{-3}\int p^2\left[e^{p^2}-1\right]^{-1}dp=-\left(8\pi^{3/2}\right)^{-1}\zeta( 5/2 ).
\end{align*}
The second identity can for example be seen by putting back in the $T$ dependence, differentiating
$\int\ln(1-e^{-p^2/T})dp$ with respect to $T$ directly under the integral sign, and also noticing that 
it is $\frac32T^{-1}$ times the integral. \\

We now prove two estimates that we will use in the next section. 
\begin{lemma}
There exist constants $c_1,C_1>0$ such that for all $n$
we have 
\begin{equation}\label{eq:e0n}
f_0(n)\leq f_{\rm min}+C_1[n_{\rm fc}-n]_+^3,
\end{equation}
and for all $n_1\leq n_2\leq n_{\rm fc}$ 
\begin{equation}\label{eq:e0nn}
f_0(n_1)\geq f_0(n_2)+c_1(n_2-n_1)^3.
\end{equation}
Also, given $n_0<n_{\rm fc}$, there exists $c_0>0$ such that for all $n_0\leq n\leq n_{\rm fc}$
\begin{equation}\label{eq:e0nn2}
f_0(n)\leq f_0(n_0)-c_0(n-n_0)(n_{\rm fc}-n_0)^2.
\end{equation}
\end{lemma}
\begin{proof}
Let us analyse how the energy $f_0(n)$ goes up if $n=n_{\rm fc}-\delta n$ for $\delta n>0$.
For simplicity we set $\lambda=-m(n)\geq0$. We then have
\begin{equation}\label{eq:densitychange}
\begin{aligned}
  \delta n=&\ (2\pi)^{-3}\left(\int[e^{p^2}-1]^{-1}
    -[e^{p^2+\lambda}-1]^{-1}dp\right)\\
  =&\ (2\pi)^{-3}\lambda^{3/2}\left(\int [e^{\lambda p^2}-1]^{-1}-[e^{\lambda(p^2+1)}-1]^{-1}dp\right)\\
  =&\ (2\pi)^{-3}\lambda^{1/2}\left(\int(|p|^{-2}-(|p|^2+1)^{-1} )dp+o(1)\right)=
  (4\pi)^{-1}\lambda^{1/2}+o(\lambda^{1/2})
\end{aligned}
\end{equation}
as $\lambda\to0$. 
We then find for the energy 
\begin{eqnarray*}
  \lefteqn{(2\pi)^{-3}\int\ln(1-e^{-(p^2+\lambda)})dp 
    -\lambda (2\pi)^{-3}\int[e^{p^2+\lambda}-1]^{-1}dp}&&\\&=&
  (2\pi)^{-3}\int\ln(1-e^{-p^2})dp\\&&+
  (2\pi)^{-3}\lambda^{3/2}\int\ln(1-e^{-\lambda(p^2+1)})-\ln(1-e^{-\lambda p^2})
  -\lambda [e^{\lambda(p^2+1)}-1]^{-1}dp  \\
  &=&f_{\rm min}+
  (2\pi)^{-3}\lambda^{3/2}\left(\int \ln(1+|p|^{-2})-(p^2+1)^{-1}dp+o(1)\right)
  \\&=&f_{\rm min}+(12\pi)^{-1}\lambda^{3/2}+o(\lambda^{3/2})
\end{eqnarray*}
as $\lambda\to0$.
We thus conclude that 
\begin{equation}
\label{eq:e0nnnn}
f_0(n)=f_{\rm min}+\frac{16\pi^2}3[n_{\rm fc}-n]_+^3+o([n_{\rm fc}-n]_+^3)
\end{equation}
as $[n_{\rm fc}-n]_+\to0 $. This proves the statement. We also see that the free Bose gas has a third-order phase transition between the condensed and non-condensed phase.

The final statement is found by combining \eqref{eq:e0nnnn} with the fact that $f_0(n)$ is convex and strictly decreasing in $0\leq n\leq n_{\rm fc}$. 
\end{proof}

\subsection{A priori estimates}
\label{aprioridilute}
In this section, we always assume that $T\leq D\rho^{2/3}$ for some fixed constant $D$. The estimates below will depend on $D$.

Our goal will be to acquire some tools to approximate the free energy functional \eqref{freeenergyfunctional} in the dilute limit $\rho^{1/3}a\ll  1$. Propositions \ref{lm:gvg1}, \ref{lm:gvg2} and \ref{prop:impconvgamma} provide a priori bounds for the terms involving $\gamma$ and $\widehat{V}$. The first estimate holds in general for $T\leq D\rho^{2/3}$. The two other estimates are sharper and provide bounds at densities very close to the free critical density where, according to Subsection \ref{furtherapriori}, the phase transition has to occur. This means that we can zoom in on this region and analyse the nature of the minimizers there. This will be done in Subsection \ref{sec:crittemp}.

Let $\left(\gamma,\alpha,\rho_0=\rho-\rho_\gamma\right)$ be a minimizing triple for \eqref{cmin} at a temperature $T$.

Using the bound $\widehat{V}(p)\leq \widehat{V}(0)$ we find the following upper bound in terms of the free gas energy $\cF_0$
\begin{align}\label{eq:upperbound}
\cF^{\rm{can}}(\gamma,\alpha,\rho_0)\leq\ &\cF^{\rm{can}}(\gamma_{\mu(\rho)},0,[\rho-\rho_{\rm fc}]_+)\nonumber\\\leq\ &
\cF_0(\gamma_{\mu(\rho)})+\rho^2\widehat{V}(0)-\frac12[\rho-\rho_{\rm fc}]_+^2\widehat{V}(0).
\end{align}
We also have
\begin{equation}\label{eq:lowerbound}
\begin{aligned}
  \cF^{\rm{can}}(\gamma,\alpha,\rho_0)\geq\ &
  \cF_0(\gamma)+\frac12\widehat{V}(0)\rho^2+
  \rho_0(2\pi)^{-3}\int\widehat{V}(p)\gamma(p)dp \\&+
  \frac12(2\pi)^{-6}\iint\gamma(p)\widehat{V}(p-q)\gamma(q)dpdq
  -\frac12 \rho_0^2\widehat{V}(0),
\end{aligned}
\end{equation}
where we have first used that the entropy decreases if we replace $\alpha$
by 0 and then minimized over $\alpha$, finding the minimizer
$\alpha=-(2\pi)^{3}\rho_0\delta_0$. We conclude  
\begin{equation}\label{eq:1stenergy}
\cF_0(\gamma_{\mu(\rho_\gamma)})\leq \cF_0(\gamma)\leq \cF_0(\gamma_{\mu(\rho)})+\rho^2\widehat{V}(0).
\end{equation}
We will use this to give an estimate on the integral of $\gamma$ in a
region $|p|>b$, where $b$ is to be chosen below.  We shall use the following result.

\begin{lemma}[A priori kinetic energy bound]\label{lm:apriorikinetic}
If for some $Y>0$ the function $\gamma$ satisfies
$\cF_0(\gamma)\leq \cF_0(\gamma_{\mu(\rho_\gamma)})+Y$, then for all $b$ with $b^2>8T$ we have 
$$
\frac12  (2\pi)^{-3}\int_{|p|>b} p^2\gamma(p)dp \leq Y+CT^{5/2}e^{-b^2/4T}.
$$
\end{lemma}

\begin{proof} 
  Using the fact that $\mu(\rho_\gamma)\leq0$ and
  $\mu(\rho_\gamma)(2\pi)^{-3}\int\gamma_{\mu(\rho_\gamma)}=\mu(\rho_\gamma)\rho_\gamma$,
  the result follows from
\begin{align*}
  \cF_0(\gamma)-\mu&(\rho_\gamma)\rho_\gamma\geq \ 
  (2\pi)^{-3}\int_{|p|<b} \left(p^2\gamma(p)-Ts(\gamma(p),0)-\mu(\rho_\gamma)\gamma(p)\right) dp
  \\+& \ 
  \frac12 (2\pi)^{-3}\int\limits_{|p|>b} p^2\gamma(p)dp +\frac12 (2\pi)^{-3}\int\limits_{|p|>b} \left(p^2\gamma(p)-2Ts(\gamma(p),0)\right) dp \\
  \geq & \ (2\pi)^{-3}\int \left(p^2\gamma_{\mu(\rho_\gamma)}(p)-Ts(\gamma_{\mu(\rho_\gamma)}(p),0)-\mu(\rho_\gamma)\gamma_{\mu(\rho_\gamma)}(p) \right)dp\\
  &\ +\frac12 (2\pi)^{-3}\int_{|p|>b} p^2\gamma(p)dp +(2\pi)^{-3}T\int_{|p|>b} \ln(1-e^{-p^2/2T})dp\\
  \geq &\
  \cF_0(\gamma_{\mu(\rho_\gamma)})-\mu(\rho_\gamma)\rho_\gamma+\frac12
  (2\pi)^{-3}\int_{|p|>b} p^2\gamma(p) -CT^{5/2}e^{-b^2/4T},
\end{align*}
which holds for $b^2>8T$, since then 
\bq
\begin{aligned}
\int_{|p|>b} \ln(1-e^{-p^2/2T})dp&\geq -C\int_{|p|>b}e^{-p^2/2T}dp\geq -Ce^{-b^2/4T}\int e^{-p^2/4T}dp \\
&=-CT^{3/2}e^{-b^2/4T}.\nn
\end{aligned}
\eq
\end{proof}

Since $\cF_0(\gamma_{\mu(\rho_\gamma)})\geq \cF_0(\gamma_{\mu(\rho)})$, we can use this lemma with 
$Y=\rho^2\widehat{V}(0)$ to conclude from (\ref{eq:1stenergy}) that 
\bq
\begin{aligned}
\iint_{|p-q|>2b} \gamma(p)\widehat{V}(p-q)\gamma(q)dpdq
\leq\ & C \widehat{V}(0)\rho\int_{|p|>b}\gamma(p)dp \\
\leq \ &C \widehat{V}(0)\rho(\rho^2\widehat{V}(0)+T^{5/2}e^{-b^2/4T})b^{-2}.\label{eq:gvg1}
\end{aligned}
\eq
We choose $b=a^{-1}(\rho^{1/3} a)^{3/4}$. Then $b^2/T\geq D^{-1}(\rho^{1/3} a)^{-1/2}\gg1$ and we find 
\begin{equation}
\label{estm1}
\iint_{|p-q|>2b}
\gamma(p)\widehat{V}(p-q)\gamma(q)dpdq\leq
C\rho^3\widehat{V}(0)^2b^{-2}\leq C\rho^2 a(\rho^{1/3} a)^{3/2}.
\end{equation}
Of course, the same bound holds if $\widehat{V}(p-q)$ is replaced by $\widehat{V}(0)$.
On the other hand we also have
\begin{equation}
\label{estm2}
\iint_{|p-q|<2b} \gamma(p)|\widehat{V}(p-q)-\widehat{V}(0)|\gamma(q)dpdq
\leq 
Cb^2\|\partial^2\widehat{V}\|_{\infty}\rho^2\leq C\rho^2 a(\rho^{1/3} a)^{3/2}.
\end{equation}
For the same choice of $b$:
\begin{equation}
\begin{aligned}
\left|\int\gamma(p)\widehat{V}(p)dp \right.&\left. 
-\widehat{V}(0)\int\gamma(p)dp\right| \leq \left|\left(\int\limits_{|p|\leq b}+\int\limits_{|p|> b}\right)\gamma(p)\left(\widehat{V}(p)
- \widehat{V}(0)\right)dp\right| \\ &\leq Cb^2 \|\partial^2\widehat{V}\|_{\infty} \int_{|p|\leq b}\gamma(p)dp+C\widehat{V}(0)b^{-2}\int_{|p|> b}p^2 \gamma(p)dp\\
&\leq C\rho a^3 b^2+Ca b^{-2}(\rho^2 a+T^{5/2}e^{-b^2/4T})\leq C\rho a(\rho^{1/3} a)^{3/2}, \label{estimongammV}
\end{aligned}  
\end{equation}
The same bounds hold for $\widehat{Vw}$ by \eqref{sameest}. We have thus shown the following result.

\begin{proposition}[A priori estimates on $E_2$ and $E_4$]
\label{lm:gvg1}
Any minimizing triple $(\gamma,\alpha,\rho_0)$ with density $\rho=\rho_\gamma+\rho_0$ and temperature $T$
satisfying $T<D\rho^{2/3}$ obeys the estimates
$$
\left|(2\pi)^{-6}\iint\gamma(p)\widehat{V}(p-q)\gamma(q)dpdq
- \widehat{V}(0)\rho_\gamma^2\right|\leq C\rho^2 a(\rho^{1/3} a)^{3/2},
$$
$$
\left|(2\pi)^{-3}\int\gamma(p)\widehat{V}(p)dp
- \widehat{V}(0)\rho_\gamma\right|\leq C\rho a(\rho^{1/3} a)^{3/2},
$$
where the constant $C$ depends on $D$ and the potential $V$. This also holds with $\widehat{V}$ replaced by $\widehat{Vw}$.
\end{proposition}

From (\ref{eq:upperbound}), (\ref{eq:lowerbound}), and Proposition \ref{lm:gvg1} we find that 
\begin{equation}\label{eq:2ndenergy}
\cF_0(\gamma_{\mu{(\rho)}})\geq \cF_0(\gamma)+\frac12[\rho-\rho_{\rm fc}]_+^2\widehat{V}(0)
-\rho_0^2\widehat{V}(0)
-C\rho^2a(\rho^{1/3} a)^{3/2},
\end{equation}
which implies 
\begin{equation}
\label{apriorestrho0}
\rho_0^2 \widehat{V}(0)\geq \frac12[\rho-\rho_{\rm fc}]_+^2\widehat{V}(0)
-C\rho^2a(\rho^{1/3} a)^{3/2}.
\end{equation}
We thus get the following result.

\begin{lemma}
\label{lemma189}
If $(\gamma,\alpha,\rho_0)$ is a minimizing triple with $\rho=\rho_\gamma+\rho_0$ satisfying 
\[
\rho>\rho_{\rm{fc}}+C\rho(\rho^{1/3} a)^{3/4},
\]
then $\rho_0\neq 0$.
\end{lemma}

It follows that phase transition can only take place for
\[
\rho\leq\rho_{\rm fc}+C\rho(\rho^{1/3} a)^{3/4}\leq \rho_{\rm fc}+C'\rho_{\rm fc}(\rho_{\rm fc}^{1/3} a)^{3/4}.
\]
Hence from now on we consider only  
\begin{equation}\label{eq:firstboundrho}
\rho\leq \rho_{\rm fc}+C'\rho_{\rm fc}(\rho_{\rm fc}^{1/3}a)^{3/4}.
\end{equation}
Under this condition we shall give an upper bound on $\rho_0$.

If $\rho_0>2C'\rho_{\rm fc}(\rho_{\rm fc}^{1/3}a)^{3/4}$, then
$\rho_\gamma=\rho-\rho_0\leq \rho_{\rm fc}-\frac12\rho_0$ and thus 
\begin{equation}\label{eq:rho0firststep}
\begin{aligned}
\cF_0(\gamma_{\mu(\rho)})
\geq\ & \cF_0(\gamma)-\widehat{V}(0)\rho_0^2
-C\rho_{\rm fc}^2a(\rho_{\rm fc}^{1/3}a)^{3/2}\\
\geq\ &\cF_0(\gamma_{\mu(\rho)})+cT^{-2}\rho_0^3-\widehat{V}(0)\rho_0^2
-C\rho_{\rm fc}^2a(\rho_{\rm fc}^{1/3}a)^{3/2},
\end{aligned}
\end{equation}
where we have used the lower bound in \eqref{eq:e0nn} with $n_1=T^{-3/2}\rho_\gamma$ and \mbox{$n_2=T^{-3/2}\min\{\rho,\rho_{\rm{fc}}\}$}.
We conclude that $\rho_0<C\rho_{\rm fc}(\rho_{\rm fc}^{1/3}a)^{5/6}$, which, in the dilute limit, contradicts the assumption \mbox{$\rho_0>2C'\rho_{\rm fc}(\rho_{\rm fc}^{1/3}a)^{3/4}$}. We conclude that \eqref{eq:firstboundrho} implies
\[
\rho_0\leq 2C'\rho_{\rm fc}(\rho_{\rm fc}^{1/3}a)^{3/4}.
\]
If we insert this bound into \eqref{eq:2ndenergy}, we obtain
\begin{equation}
\label{eq:3energy}
\cF_0(\gamma_{\mu(\rho)})\geq\cF_0(\gamma)-C\rho_{\rm{fc}}^2 a(\rho_{\rm{fc}}^{1/3}a)^{3/2}.
\end{equation}
Since $\cF_0(\gamma_{\mu(\rho)})\leq \cF_0(\gamma_{\mu(\rho_\gamma)})$, we use Lemma \ref{lm:apriorikinetic} with \mbox{$Y=Ca\rho_{\rm fc}^2(\rho_{\rm fc}^{1/3}a)^{3/2}$}, and, as in (\ref{eq:gvg1}), arrive at
\begin{equation}
\label{eq248}
\ba
\iint_{|p-q|>2b} \gamma(p)\widehat{V}(p-q)&\gamma(q)dpdq\\
&\leq C \widehat{V}(0)\rho(a\rho_{\rm fc}^2(\rho_{\rm fc}^{1/3}a)^{3/2}+T^{5/2}e^{-b^2/4T})b^{-2}.
\ea
\end{equation}
We choose $b=a^{-1}(\rho_{\rm fc}^{1/3}a)^{3/4}$, such that 
$b^2/T\geq c(\rho_{\rm fc}^{1/3}a)^{-1/2}\gg 1$. The error above is 
then $C\rho_{\rm fc}^2 a (\rho_{\rm fc}^{1/3}a)^{3}$.
This time we can expand $\widehat{V}$ to second order
\bq
\begin{aligned}
  \iint\limits_{|p-q|<2b}
  \gamma(p)|\widehat{V}(p-q)-\widehat{V}(0)-&\frac16\Delta\widehat{V}(0)
(p-q)^2|\gamma(q)dpdq\leq C b^3\sup|\partial^3\widehat{V}|\rho^2 \\
&=Cb^3a^4\rho^2 \leq C\rho_{\rm fc}^2 a (\rho_{\rm fc}^{1/3}a)^{2+1/4}.\nn
\end{aligned}
\eq
Note that the integrals of the terms involving $\widehat{V}(0)$ and $\Delta\widehat{V}(0)$ over $\{|p-q|>2b\}$ can be estimated with Lemma \ref{lm:apriorikinetic} like \eqref{eq248}, that all these bounds can also be derived for $\int\widehat{V}(p)\gamma(p)dp$, and that we can derive similar bounds for $\widehat{Vw}$ using \eqref{sameest}, so that we arrive at the following improvement of Proposition \ref{lm:gvg1}.

\begin{proposition}
\label{lm:gvg2}
Any minimizing triple $(\gamma,\alpha,\rho_0)$ with density $\rho=\rho_\gamma+\rho_0$
 and temperature $T$ satisfying $D^{-3/2}T^{3/2}<\rho<\rho_{\rm fc}+C'\rho_{\rm fc}(\rho_{\rm fc}^{1/3}a)^{3/4}$
obeys the estimates
\begin{equation}\label{eq:2ndorderpq}
\begin{aligned}
\left|(2\pi)^{-6}\iint\gamma(p)\widehat{V}(p-q)\gamma(q)dpdq -
  \widehat{V}(0)\rho_\gamma^2-\right.&\left.\frac{1}{3(2\pi)^{3}}\Delta\widehat{V}(0)\rho_\gamma\int
  p^2\gamma(p)dp\right|\\
  & \leq C\rho_{\rm fc}^2 a (\rho_{\rm fc}^{1/3}a)^{2+1/4}
\end{aligned}
\end{equation}
and 
\[
\begin{aligned}
\left|(2\pi)^{-3}\int\widehat{V}(p)\gamma(p)dp -
  \widehat{V}(0)\rho_\gamma-\frac{1}{6(2\pi)^{3}}\Delta\widehat{V}(0)\right.&\left.\int p^2\gamma(p)dp\right| \\
  &\leq C\rho_{\rm fc} a (\rho_{\rm fc}^{1/3}a)^{2+1/4},
  \end{aligned}
\]
where the constants $C$ depend on $D$ and the potential $V$ This also holds with $\widehat{V}$ replaced by $\widehat{Vw}$.
\end{proposition}
We are now ready to prove two more results. First, we provide an upper bound on $\rho_0$ and one on densities where a phase transition can occur (`critical densities'), which will be matched with a lower bound in the next section to show that there is no phase transition outside the region $|\rho-\rho_{\rm{fc}}|<C\rho(\rho^{1/3} a)$. The second is an a priori estimate on the error $E_5$.
\begin{lemma}[Upper bound on critical densities and $\rho_0$]
\label{lemub}
Assume that the density $\rho=\rho_0+\rho_\gamma$ and temperature $T$ satisfy $D^{-3/2}T^{3/2}<\rho<\rho_{\rm fc}+C'\rho_{\rm fc}(\rho_{\rm fc}^{1/3}a)^{3/4}$. Then,
\begin{itemize}
\item $\rho_0< C \rho (\rho^{1/3}a)$.
\item there exists a constant $C$ such that any minimizing triple with $\rho>\rho_{\rm{fc}}+C\rho(\rho^{1/3} a)$ has $\rho_0\neq 0$.
\end{itemize}
\end{lemma}
\begin{proof}
 For $|\delta|<1$ (both positive and negative) we find using the scaling of the free gas energy that 
 \begin{equation}
\label{eq:scalingfree}
\cF_0(\gamma)\geq (2\pi)^{-3} \delta\int p^2\gamma(p)dp+(1-\delta)^{-3/2}\cF_0(\gamma_0). 
\end{equation}
Since $\cF_0(\gamma)\leq C\rho_{\rm{fc}}^2 a(\rho_{\rm{fc}}^{1/3}a)^{3/2}$ by \eqref{eq:3energy} and $\mathcal{F}(\gamma_0)\leq0$, it follows that 
\begin{equation}\label{eq:psquaredgamma}
\int p^2\gamma(p)dp \leq C \rho^{5/3}.
\end{equation}
Together with Proposition \ref{lm:gvg2}, this implies that 
$$(2\pi)^{-6}\iint \gamma(p)\widehat{V}(p-q)\gamma(q)dpdq=\widehat{V}(0)\rho^2_{\gamma}+O(\rho^2a(\rho^{1/3}a)^2)$$
and 
$$(2\pi)^{-3}\int \widehat{V}(p)\gamma(p)dp=\widehat{V}(0)\rho_{\gamma}+O(\rho a(\rho^{1/3}a)^2).$$
These two bounds together with \eqref{eq:upperbound} and \eqref{eq:lowerbound} yield
\begin{equation}\label{eq:boundcriticalregion}
\begin{aligned}
\cF_0(\gamma_{\mu(\rho)})+&\rho^2\frac{\widehat{V}(0)}{2}-\frac12[\rho-\rho_{\rm fc}]_+^2\widehat{V}(0)\geq \\ &\cF_0(\gamma)+\rho_0\rho_\gamma \widehat{V}(0)+
  \frac{\widehat{V}(0)}{2}\rho_\gamma^2  -\frac{\widehat{V}(0)}{2} \rho_0^2+O(\rho^2(\rho^{1/3}a)^2),
\end{aligned}
\end{equation}
and so
$$\rho_0^2\geq \frac12[\rho-\rho_{\rm fc}]_+^2-C\rho^2(\rho^{1/3}a)^2,$$
which implies the second statement. 
We also notice that \eqref{eq:boundcriticalregion} and \eqref{eq:e0nn} (used as in \eqref{eq:rho0firststep}) imply
\begin{equation*}
CT^{-2}\rho_0^3-\widehat{V}(0)\rho_0^2-Ca \rho^2(\rho^{1/3}a)^2\leq0,
\end{equation*}
which proves the first statement.
\end{proof}

\begin{proposition}[A priori estimate on $E_5$]
\label{prop:impconvgamma}
Let $(\gamma,\alpha,\rho_0)$ be a minimizing triple with density $\rho=\rho_\gamma+\rho_0$ such that $|\rho-\rho_{\rm{fc}}|<C\rho(\rho^{1/3} a)$. Also assume $T<D\rho^{2/3}$. Then
\begin{equation*}
(2\pi)^{-6}\iint \gamma(p)\widehat{V}(p-q)\gamma(q)dpdq=\widehat{V}(0)\rho_\gamma^2+{\frac {\zeta  \left( 3/2 \right) \zeta  \left( 5/2 \right) }{128
{\pi }^{3}}}\Delta\widehat{V}(0)T^4
+o(T^4a^3).
\end{equation*}
\end{proposition}
\begin{proof}
First notice that
\begin{equation}
\label{propertyofV}
{\frac {\zeta  \left( 3/2 \right) \zeta  \left( 5/2 \right) }{128
{\pi }^{3}}}\Delta\widehat{V}(0)T^4=\frac{\Delta\widehat{V}(0)}{3(2\pi)^{6}}\int p^2(e^{p^2/T} -1)^{-1}dp
\int (e^{p^2/T} -1)^{-1}dp, 
\end{equation}
so that according to \eqref{eq:2ndorderpq} it is enough to show that 
\begin{equation}
\label{tobebdd}
\frac13\Delta\widehat{V}(0)\rho_\gamma\int
  p^2\gamma(p)dp=\frac{1}{3(2\pi)^{3}}\Delta\widehat{V}(0)\int\gamma_0(p)dp\int p^2\gamma_0(p)dp+
o(T^4a^3).
\end{equation}
We have
\begin{equation*}
\begin{aligned}
\left|\rho_\gamma\int p^2\gamma(p)dp -\rho_{\rm{fc}} \right.&\left. \int p^2\gamma_0(p)dp\right|\leq \\
& |\rho_\gamma-\rho_{\rm{fc}}|\int p^2\gamma(p)dp+\rho_{\rm{fc}}\left|\int p^2(\gamma(p)-\gamma_0 (p))dp\right|.
\end{aligned}
\end{equation*}
The first statement in Lemma \ref{lemub}, combined with the assumptions, implies
\begin{equation*}
|\rho_\gamma-\rho_{\rm{fc}}|\leq \rho_0 +|\rho -\rho_{\rm{fc}}|\leq C \rho_{\rm{fc}}(\rho_{\rm{fc}}^{1/3}a).
\end{equation*}
This and \eqref{eq:psquaredgamma} allow us to bound the first contribution to the difference in \eqref{tobebdd}:
$$\Delta\widehat{V}(0)|\rho_\gamma-\rho_{\rm{fc}}|\int p^2\gamma(p)dp\leq C \rho_{\rm{fc}}^3a^4=o(T^4a^3). $$
To bound the other contribution, we use \eqref{eq:scalingfree}. We do the same for $\gamma_0$, but with $-\delta$. Putting these two bounds together yields
\begin{equation*}
(2\pi)^{-3}\delta\int p^2(\gamma(p)-\gamma_0(p))dp\leq \cF_0(\gamma)+\cF_0(\gamma_0)-
((1-\delta)^{-3/2}+(1+\delta)^{-3/2})\cF_0(\gamma_0).
\end{equation*}
Writing \eqref{eq:e0n} and \eqref{eq:3energy} in succession gives
$$\cF_0(\gamma_0)+CT^{-2}[\rho_{\rm{fc}}-\rho]_+^{3} \geq \cF_0(\gamma_{\mu(\rho)})\geq\cF_0(\gamma)-C\rho_{\rm{fc}}^2 a(\rho_{\rm{fc}}^{1/3}a)^{3/2},$$
which implies
$$CT^{5/2}\rho_{\rm{fc}}a^3  +C\rho_{\rm{fc}}^2 a(\rho_{\rm{fc}}^{1/3}a)^{3/2}\geq0.$$
Thus
\begin{equation*}
\begin{aligned}
\frac{\delta}{(2\pi)^3}\int p^2(\gamma(p)-\gamma_0(p))dp& \leq -((1-\delta)^{-\frac32}+(1+\delta)^{-\frac32}-2)\cF_0(\gamma_0)\\ &\quad+CT^{5/2}\rho_{\rm fc}a^3+C\rho_{\rm{fc}}^2 a(\rho_{\rm{fc}}^{1/3}a)^{3/2} \\
& \leq C\delta^2 T^{5/2}+CT^{5/2}\rho_{\rm fc}a^3+CT^{5/2}(\rho_{\rm{fc}}^{1/3}a)^{5/2}.
\end{aligned}
\end{equation*}
By choosing $|\delta|=(\rho_{\rm{fc}}^{1/3}a)^{5/4}$, we finally obtain
$$\int p^2(\gamma(p)-\gamma_0(p))dp\leq CT^{5/2}(\rho_{\rm{fc}}^{1/3}a)^{5/4},$$
which implies
$$\Delta\widehat{V}(0) \rho_{\rm{fc}}\left|\int p^2(\gamma(p)-\gamma_0 (p))dp\right|\leq C a^3 T^4 (\rho_{\rm{fc}}^{1/3}a)^{5/4} =o(T^4a^3 ).$$
This completes the proof.
\end{proof}

\subsection{Estimate on critical densities}
\label{furtherapriori}
In this section, we provide a lower bound on densities where a phase transition can occur. Together with the upper bound from Lemmas \ref{lemma189} and \ref{lemub}, we obtain the following a priori estimate.
\begin{proposition}[Estimate on critical densities]\label{prop:criticalregion}
There exists a constant $C_0$ such that for any minimizing triple:
\begin{enumerate}
\item $\rho_0\neq 0$ if $\rho>\rho_{\rm{fc}}+C_0\rho_{\rm{fc}}(\rho_{\rm{fc}}^{1/3} a)$;
\item $\rho_0= 0$ if $\rho<\rho_{\rm{fc}}-C_0\rho_{\rm{fc}}(\rho_{\rm{fc}}^{1/3} a)$.
\end{enumerate}
\end{proposition}

\begin{proof}[Proof of the second statement. (The first follows from Lemmas \ref{lemma189} and \ref{lemub}.)]
\textit{Step 1.}
We will first consider temperatures $T\leq D\rho^{2/3}$, so that we can use the a priori estimates proved in the previous section, and comment on higher temperatures in the final step. 

We are interested in the canonical minimization problem \eqref{cmin}, but our strategy will be to use the grand canonical formulation of the problem.
This is not straightforward since the canonical energy is not necessarily convex in $\rho$ (it will indeed turn out not to be as we prove in Subsection \ref{sec:crittemp}).

As a first step, we simply assume the correspondence between canonical and grand canonical is obvious. That is, given $\rho$, there is a $\mu$ such that the canonical minimizing triple $(\gamma,\alpha,\rho_0)$ with $\rho_0+\rho_\gamma=\rho$ is a minimizer of the grand canonical functional \eqref{freeenergyfunctional} with that $\mu$ (which will not be the case in general.) In \cite{NapReuSol1-15}, it was shown that $\gamma$ satisfies the Euler--Lagrange equation
\begin{equation*}
p^2-\mu+\rho\widehat{V}(0)+\rho_0\widehat{V}(p)+(2\pi)^{-3}\widehat{V}\ast\gamma(p)-T\frac{\gamma+\frac12}{\beta}\ln \frac{\beta+\frac12}{\beta-\frac12}=0.
\end{equation*}
Since $\beta=\sqrt{(\gamma+\frac12)^2-\alpha^2}$, it follows that
\begin{equation*}
p^2-\mu+\rho\widehat{V}(0)+\rho_0\widehat{V}(p)+(2\pi)^{-3}\widehat{V}\ast\gamma(p)-T\ln \frac{\gamma+1}{\gamma}\geq 0,
\end{equation*}
which implies
\begin{equation}\label{eq:ELgammabound}
\gamma(p)\geq\left[ \exp\left(\frac{p^2-\mu+\rho\widehat{V}(0)+\rho_0\widehat{V}(p)+(2\pi)^{-3}\widehat{V}\ast\gamma(p)}{T}\right)-1\right]^{-1}.
\end{equation}
The same argument as in \eqref{eq:lowerbound} implies that 
\[
\cF(\gamma,\alpha,\rho_0)\geq \cF(\gamma,0,0)+\rho_0\left((2\pi)^{-3}\int\widehat{V}(p)\gamma(p)dp+\rho_\gamma\widehat{V}(0)-\mu\right).
\]
Thus, if the minimizer has $\rho_0> 0$, then we need to have 
$$(2\pi)^{-3}\int\widehat{V}(p)\gamma(p)dp+\rho_\gamma\widehat{V}(0)-\mu\leq 0.$$
Using this in \eqref{eq:ELgammabound}, we obtain
\begin{equation}
\label{somest1209}
\gamma(p)\geq \left[ \exp\left(\frac{p^2+ C a\rho(\rho^{1/3} a)}{T}\right)-1\right]^{-1},
\end{equation}
where we also used the first statement in Lemma \ref{lemub} and Proposition \ref{lm:gvg1}. Given the claim we are trying to prove, we can assume $\rho\leq\rho_{\rm fc}$, so that, using a change of variables and the fact that $\rho_0>0$, we have 
\[
\rho_\gamma \geq T^{3/2}\int \left[ \exp\left(p^2+ C(\rho^{1/3} a)^{2}\right)-1\right]^{-1}dp \geq \rho_{\rm{fc}}(1-C(\rho^{1/3} a)),
\]
where we used \eqref{eq:densitychange}. We conclude that  there exists a constant $C_1$ such that $\rho_0= 0$ for any minimizing triple with $\rho<\rho_{\rm{fc}}-C_1\rho_{\rm{fc}}(\rho_{\rm{fc}}^{1/3} a)$ satisfying the extra assumption that there is a $\mu$ that will give the same minimizer of the grand canonical problem. This will, however, not be the case in general because the canonical energy may not be convex in $\rho$.

\textit{Step 2.}  Given a $\rho$, there are $\rho_\pm$ such that $\rho_-\leq\rho\leq \rho_+$ and such that the
convex hull of $F^{\rm can}$ is linear on the interval
$[\rho_-,\rho_+]$. 
To see this, we first use that $\rho_0=0$ for small $\rho$, as established in \cite{NapReuSol1-15}. Together with the fact that the canonical
functional with $\rho_0=0$ is strictly convex, this implies that the canonical energy is convex for small $\rho$. The simple lower bound
\[
F^{\rm can}(T,\rho)\geq -CT^{5/2}-\frac12\rho_0\int\widehat{V}+\frac12\widehat{V}(0)\rho^2
\]
then confirms the existence of $\rho_-$ and $\rho_+$.  

The assumption made in the previous step will hold  for $\rho_{\pm}$, i.e.\ $\rho_+$ and $\rho_-$ correspond to a minimum for the
grand canonical functional for some (shared) $\mu$ that is the slope of $F^{\rm can}$ on $[\rho_-,\rho_+]$, and the conclusion from step 1 above holds for these densities. Since $\rho_-\leq \rho$, this implies
that if we choose $C_0\geq C_1$ then $\rho_{0-}=0$ for the total
density $\rho_-$.

If the density $\rho_+$ also satisfies a corresponding upper bound, then $\rho_{0+}=0$ as well. In that case, as the canonical
functional with $\rho_0=0$ is strictly convex, we conclude that in the
interval $[\rho_-,\rho_+]$ we must have $\rho_0=0$ and hence
$\rho_-=\rho_+=\rho$. 

Let $\mu$ be the slope of the convex hull of $F^{\rm can}$ on $[\rho_-,\rho_+]$ (where it is linear). By the $\rho_0$-Euler--Lagrange equation for the grand canonical functional \eqref{freeenergyfunctional}, it follows that
\[
\mu\leq \rho_-\widehat{V}(0) +(2\pi)^{-3}\int\widehat{V}\gamma_- \leq
2\rho_-\widehat{V}(0) \leq 2\rho_{\rm fc}(1-C_0\rho_{\rm
  fc}^{1/3}a)\widehat{V}(0).
\]
The aim is to prove that $\rho_+<\rho_{\rm{fc}}-C_1\rho_{\rm{fc}}(\rho_{\rm{fc}}^{1/3} a)$ by proving an upper bound on any density
minimizing the grand canonical functional with $\mu$ satisfying the bound above. 
As the minimizing density increases with $\mu$, we can assume that 
\begin{equation}\label{eq:mu}
\mu=2\rho_{\rm fc}(1-C_0\rho_{\rm
  fc}^{1/3}a)\widehat{V}(0).
\end{equation}
Recall that $C_0$ is a constant that we will choose large enough to get the proof to work. Our choice for $C_0$ will be universal, so that we can make the a priori assumption that 
$C_0\rho_{\rm fc}^{1/3}a\leq1$.

\textit{Step 3.} Let $\mu$ be as in \eqref{eq:mu}. We will now first show the a-priori bound $\rho\leq C\rho_{\rm
  fc}$.  From \eqref{someequ}, the definition of $\cF^{\rm sim}$ \eqref{eq:simplifiedfunctional} and the definition of $\cF^{\rm S}$ \eqref{simplfunc}, we find that
\begin{eqnarray}
  \cF(\gamma,\alpha,\rho_0)&=& \cF^{\rm S}(\gamma,\alpha,\rho_0)-\mu\rho\nonumber\\
  &&+\frac12\widehat{V}(0)\rho_\gamma^2
  +\widehat{V}(0)\rho_0\rho_\gamma-(\rho_0+t_0)(2\pi)^{-3}\int\widehat{Vw}(p)\gamma(p)dp
  \nonumber\\&&+(2\pi)^{-3} \rho_0\int\widehat{V}(p)\gamma(p)dp-4\pi a(\rho_0+t_0)^2+8\pi a(\rho_0+t_0)\rho_0
  \nonumber\\
  &&+\frac12(2\pi)^{-6}\iint\gamma(p)\widehat{V}(p-q)\gamma(q)dpdq\label{eq:1}\\&&+
  \frac12(2\pi)^{-6}\iint(\alpha(p)-\alpha_0(p))\widehat{V}(p-q)(\alpha(p)-\alpha_0(p))dpdq.\nonumber
\end{eqnarray}
We now choose $0\geq t_0\geq -\rho_0$. If $8\pi a \rho_0\leq 4\rho_{\rm fc}\widehat{V}(0)$ 
we choose $t_0=0$. Note that in this case we already have an upper bound $\rho_0\leq C\rho_{\rm
  fc}$, and the argument below will give the desired result for $\rho_\gamma$. 
Otherwise we choose
$$
8\pi a(t_0+\rho_0)=4 \rho_{\rm fc} \widehat{V}(0)>2\mu
$$
by the assumption \eqref{eq:mu} on $\mu$. 
We now give a lower bound by ignoring the last two integrals, the
second term in the second line, and the first term in the third
line in \eqref{eq:1}. Finally we minimize $\cF^{\rm s}$ using Lemma 22 with $\delta=0$.

We first consider the last integral in the expression for the minimum of
$\cF^{\rm s}$. We know from the assumptions made at the start of Sections \ref{sec:previousresults} and \ref{mr} that
\begin{equation}\label{eq:2}
|\widehat{Vw}(p)|\leq \widehat{Vw}(0)=8\pi a,\quad \widehat{V w}(p)\geq 8\pi a-C a^3 p^2,
\end{equation}
where the first inequality follows since $Vw$ is positive. The only negative contribution to the 
last integral therefore comes from the region $|p|>C/a$. For such $p$ we have that 
\[
(\rho_0+t_0)|\widehat{Vw}(p)|/p^2\leq C\rho_{\rm fc}a^3\ll 1. 
\]
Hence the last integral can be estimated below by 
$$
-C(\rho_0+t_0)^3\int_{|p|>1/a}\frac{|\widehat{Vw}(p)|^3}{p^4}\geq -C\rho_{\rm fc}^3a^4=
-C\rho_{\rm fc}^2a (\rho_{\rm fc}a^3).
$$
This argument will again be used in the next step to bound this integral.

The first integral with $G$ can be bounded below by replacing $G$ with a lower bound. 
We use \eqref{eq:2} again:
\begin{eqnarray*}
  G=T^{-1}\sqrt{p^4+2(\rho_0+t_0)p^2\widehat{V w}(p)}\geq  T^{-1}p^2\sqrt{1-C\rho_{\rm fc}a^3}.
\end{eqnarray*}
Altogether, we arrive at a lower bound
\begin{eqnarray*}
\cF(\gamma,\alpha,\rho_0)&\geq& \cF_0(\gamma_{0})(1+C\rho_{\rm fc}a^3)
  -\mu\rho_\gamma
  -\mu\rho_0\\&&+\frac12\widehat{V}(0)\rho_\gamma^2
  -Ca\rho_{\rm fc}\rho_\gamma-Ca\rho_{\rm fc}^2+4\rho_{\rm fc}\widehat{V}(0)\rho_0-C\rho_{\rm fc}^3a^4\\&\geq&
  \cF_0(\gamma_{0})-C\rho_{\rm fc}^{8/3}a^3
  -2\rho_{\rm fc}\widehat{V}(0)\rho_\gamma+\frac12\widehat{V}(0)\rho_\gamma^2\\&&-Ca\rho_{\rm fc}\rho_\gamma
  +2\rho_{\rm fc}\widehat{V}(0)\rho_0-Ca\rho_{\rm fc}^2-C\rho_{\rm fc}^3a^4,
\end{eqnarray*}
where $\gamma_{0}$ is the minimizer of the free gas functional. 
By inserting $\gamma_{0}$  and 
$\alpha=\rho_0=0$ into $\cF$ we also get the upper bound
\[
    \inf\cF\leq\cF_0(\gamma_{0})-\mu\rho_{\rm fc}+\widehat{V}(0)\rho_{\rm fc}^2\leq
  \cF_0(\gamma_{0})-\widehat{V}(0)\rho_{\rm fc}^2+2D\rho_{\rm fc}^2\widehat{V}(0)(\rho_{\rm
    fc}^{1/3}a).
\]
Together these upper and lower bounds imply that minimizers $\rho_0,\rho_\gamma\leq C\rho_{\rm fc}$, which gives the desired a priori upper bound on $\rho$.

\textit{Step 4.} 
To finish the argument, we need to make more refined choices for both the upper and the lower bound. 
As an upper bound, we will use the minimum of the expression
$$
\cF_0(\gamma)-\mu\rho_\gamma+\widehat{V}(0)\rho_\gamma^2.
$$
The minimizer will be the free gas minimizer $\gamma_{\delta_0}$
corresponding to a positive chemical potential $\delta_0>0$,
determined such that $\rho_{\gamma_{\delta_0}}=(2\pi)^{-3}\int\gamma_{\delta_0}$
also minimizes
$$
-\mu\rho_\gamma-\delta_0\rho_\gamma+\widehat{V}(0)\rho_\gamma^2,
$$
i.e.
$$
\mu+\delta_0=2\widehat{V}(0)\rho_{\gamma_{\delta_0}}.
$$
Let us write 
$$
\delta_0=\kappa^2\rho_{\rm fc}^{4/3}a^2
$$
for some $\kappa$ that we will now determine. We know from \eqref{eq:densitychange}
that the free gas minimizer $\gamma_{\delta_0}$ will
have 
$$
\rho_{\gamma_{\delta_0}}=\rho_{\rm fc}(1-C_2\kappa(\rho_{\rm fc}^{1/3}a+o(\rho_{\rm fc}^{1/3}a)))
$$
for an appropriate constant $C_2>0$.
Hence, the equation for $\kappa$ is 
$$
-2C_0\widehat{V}(0)\rho_{\rm fc}(\rho_{\rm fc}^{1/3}a)+\kappa^2\rho_{\rm fc}^{4/3}a^2
=-2C_2\widehat{V}(0)\rho_{\rm fc}\kappa(\rho_{\rm fc}^{1/3}a+o(\rho_{\rm fc}^{1/3}a)),
$$
that is, 
\begin{equation}\label{eq:kappa}
  \kappa^2+2C_2\kappa -2 C_0=o(1),
\end{equation}
where $C_0,C_2>0$. 

We can use the a priori bounds in Proposition \ref{lm:gvg1}, and since we know that $\rho\leq
C\rho_{\rm fc}$, we can express the error terms with $\rho$ replaced by
$\rho_{\rm fc}$. We then go back to the expression \eqref{eq:1} to get
an improved lower bound.  We set $t_0=0$ and only ignore the last
double integral. We arrive at 
\begin{eqnarray*}
  \cF(\gamma,\alpha,\rho_{0})&\geq& \cF^{\rm S}(\gamma,\alpha,\rho_{0})
+ (2\widehat{V}(0)-8\pi a)\rho_{0}\rho_{\gamma}-\mu\rho\\&&
  +4\pi a \rho_{0}^2+\rho_{\gamma}^2\widehat{V}(0)-C\rho_{\rm fc}^2a(\rho_{\rm fc}^{1/3}a)^{3/2},
\end{eqnarray*}
and apply Lemma \ref{prop:simplfunctsol} with 
$$
\delta=\delta_0+(2\widehat{V}(0)-8\pi a)\rho_{0}.
$$
The expression for $G$ will then satisfy
\begin{eqnarray*}
  G&=&T^{-1}\sqrt{(p^2+\delta_0+
    (2\widehat{V}(0)-8\pi a+\widehat{Vw}(p))\rho_{0})^2-\rho_{0}^2\widehat{Vw}(p)^2}\\
  &\geq& 
  T^{-1}\sqrt{((1-C\rho_{\rm fc}a^3)p^2+\delta_0)+2\rho_{0}\widehat{V}(0))^2-(8\pi a \rho_{0})^2}\\
  &\geq&T^{-1}\sqrt{((1-C\rho_{\rm fc}a^3)p^2+\delta_0)^2+4((1-C\rho_{\rm fc}a^3)p^2+\delta_0)
    \rho_{0}\widehat{V}(0)}.
\end{eqnarray*}
If we insert into the lower bound of Lemma \ref{prop:simplfunctsol}
and bound the $G$-integral using Lemma~\ref{lm:intexp} below, we obtain
\begin{eqnarray*}
  \cF(\gamma,\alpha,\rho_{0})&\geq& \cF_0(\gamma_{\delta_0})+\delta_0\rho_{\gamma_{\delta_0}}+2\rho_{\gamma_{\delta_0}}\widehat{V}(0)\rho_{0}
  -C\rho_{\rm fc}^{2/3}(\rho_{0}a)^{3/2}\\&&-\delta_0\rho_{\gamma}
  -\mu\rho
  +4\pi a \rho_{0}^2+\rho_{\gamma}^2\widehat{V}(0)-C\rho_{\rm fc}^2a(\rho_{\rm fc}^{1/3}a)^{3/2}\\
  &\geq&\cF_0(\gamma_{\delta_0})-\mu\rho_{\gamma_{\delta_0}}+\widehat{V}(0)\rho_{\gamma_{\delta_0}}^2+
  (2\rho_{\gamma_{\delta_0}}\widehat{V}(0)-\mu)\rho_{0}
  +4\pi a \rho_{0}^2
  \\&&+\widehat{V}(0)(\rho_\gamma-\rho_{\gamma_{\delta_0}})^2-C\rho_{\rm fc}^{2/3}(\rho_{0}a)^{3/2}
  -C\rho_{\rm fc}^2a(\rho_{\rm fc}^{1/3}a)^{3/2}\\
  &=&\cF_0(\gamma_{\delta_0})-\mu\rho_{\gamma_{\delta_0}}+\widehat{V}(0)\rho_{\gamma_{\delta_0}}^2\\&&+
  2(C_0-C_2\kappa)\rho_{\rm fc}(\rho_{\rm fc}^{1/3}a)\rho_0\widehat{V}(0)+4\pi a \rho_{0}^2-C\rho_{\rm fc}^{2/3}(\rho_{0}a)^{3/2}
  \\&&+\widehat{V}(0)(\rho_\gamma-\rho_{\gamma_{\delta_0}})^2-C\rho_{\rm fc}^2a(\rho_{\rm fc}^{1/3}a)^{3/2} 
  \\ &\geq & \cF_0(\gamma_{\delta_0})-\mu\rho_{\gamma_{\delta_0}}+\widehat{V}(0)\rho_{\gamma_{\delta_0}}^2\\&&+
  2(C_0-C_2\kappa)\rho_{\rm fc}(\rho_{\rm fc}^{1/3}a)\rho_0\widehat{V}(0)+2\pi a \rho_{0}^2-C\rho_{\rm fc}^{2}(\rho_{\rm fc}^{1/3} a)^{2}
  \\&&+\widehat{V}(0)(\rho_\gamma-\rho_{\gamma_{\delta_0}})^2-C\rho_{\rm fc}^2a(\rho_{\rm fc}^{1/3}a)^{3/2}. 
\end{eqnarray*}
Thus we conclude, by choosing $C_0$ large enough (such that $C_0-C_2\kappa$ will be positive), that 
$$
\rho_\gamma\leq \rho_{\gamma_{\delta_0}}+C\rho_{\rm fc}(\rho_{\rm fc}^{1/3}a)^{3/4},\quad 
\rho_0\leq C\rho_{\rm fc}(\rho_{\rm fc}^{1/3}a)^{3/4}.
$$

We can now apply Proposition \ref{lm:gvg2} and also the bound \eqref{eq:psquaredgamma} to improve the last error term in the lines above. We consider
the terms with the Laplacian in \eqref{eq:2ndorderpq} and the second displayed estimate in Proposition \ref{lm:gvg2}
as error terms, which lead to an error of order $\rho_{\rm fc}^2a(\rho_{\rm fc}^{1/3}a)^2$.
We conclude that for $C_0$ large enough:
$$
\rho_\gamma\leq \rho_{\gamma_{\delta_0}}+C\rho_{\rm fc}(\rho_{\rm fc}^{1/3}a),\quad 
\rho_0\leq C\rho_{\rm fc}(\rho_{\rm fc}^{1/3}a).
$$
We therefore find that 
$$
\rho\leq \rho_{\rm fc}(1-C_0\kappa(\rho_{\rm fc}^{1/3}a+o(\rho_{\rm fc}^{1/3}a)))+C\rho_{\rm fc}(\rho_{\rm fc}^{1/3}a).
$$
By the expression for $\kappa$ it is therefore clear that by choosing $C_0$ large enough we obtain
that 
$$
\rho\leq \rho_{\rm fc}(1-C_1(\rho_{\rm fc}^{1/3}a))
$$
as desired. The result obtained in step 1 and the reasoning in step 2 then finish the proof for temperatures $T\leq D\rho^{2/3}$.

\textit{Step 5.} For $T>D\rho^{2/3}$, or equivalently, $\rho< D^{-3/2}T^{3/2}$, first note that the reasoning in step 1 without reference to Lemma \ref{lemub} and Proposition \ref{lm:gvg1} leads to an equivalent of \eqref{somest1209} and the conclusion that there exists a constant $C_1$ such that $\rho_0= 0$ for any minimizing triple with $\rho\leq\rho_{\rm{fc}}-C_1\rho_{\rm{fc}}(\rho_{\rm{fc}}^{1/3} a)^{1/2}$ satisfying the extra assumption that there is a $\mu$ that will give the same minimizer of the grand canonical problem. This is certainly sufficient for $\rho< D^{-3/2}T^{3/2}$, and we again try to employ the reasoning of step 2 to avoid the extra assumption. Luckily, it is immediately clear that $\rho_{0+}=0$: either $\rho_+\leq \rho_{\rm{fc}}-C_1\rho_{\rm{fc}}(\rho_{\rm{fc}}^{1/3} a)^{1/2}$, so that $\rho_{0+}=0$, or the interval $[\rho_-,\rho_+]$ contains a density $\tilde{\rho}> D^{-3/2}T^{3/2}$, in which case the steps above imply that $\tilde{\rho}_{0+}=\rho_{0+}=0$.

\end{proof}

We have used the following lemma, which is proved in Appendix \ref{app:intexp}.
\begin{lemma}
\label{lm:intexp}
For $0\leq \delta_0, b\leq 1$ there exists a constant $C>0$ such that
\[
\ba
&\left|\int\ln\left(1-e^{-\sqrt{(p^2+\delta_0)^2+2(p^2+\delta_0)b}}\right)dp\right.\\
&\left.-\int\ln(1-e^{-(p^2+\delta_0)}dp-b\int(e^{p^2+\delta_0}-1)^{-1}dp\ \right|\leq Cb^{3/2}.
\ea
\]
\end{lemma}

\subsection{Preliminary approximations}
\label{sec:prelapprox} 
The previous sections have provided all the a priori knowledge we will need. In this section, we would like to approximate the integrals in Corollary \ref{energyexpressions} in different ways. The proof of all lemmas can be found in Appendix \ref{app:intexp}. 

We will be working with the general assumption \eqref{assumpt0} on $t_0$. We will also write $\delta=d\phi^2$, where $\phi$ will be chosen to be $\sqrt{\rho_0a}$ or $Ta$ in later sections. Note that the dilute limit corresponds to $\phi\to0$, so this is what we will assume throughout the section. To keep track of the different limits, we describe $\phi^2/T\ll 1$ as `moderate temperatures', and $\phi^2/T\geq O(1)$ as `low temperatures'. Also, a statement like `$\phi a\ll1$' means $\phi a\leq C$ for some constant small enough.

We start by analysing the first contribution to the density in \eqref{eq:rhogammasimplified}.

\begin{lemma}[$\rho_\gamma^{(1)}$  approximation]
\label{lem:rhogamma2contribution}
Let $\sigma_0\geq0$ and $d_0\geq0$ be fixed constants, and let $-1\leq\theta\leq0$, $0\leq d\leq d_0$, $0\leq\sigma\leq\sigma_0$, $\phi>0$ and $0\leq\rho_0\leq\rho$. Assume $\rho_0a/\phi^2=\sigma/8\pi$ and let $\delta=d\phi^2$ and $t_0=\theta\rho_0$. For $\phi a\ll1$, we have
\[
\rho_\gamma^{(1)}=\phi^3\frac{1}{2}I_3(d,\sigma,\theta)+o\left(\phi^3\right).
\]
The error is depends only on $\sigma_0$ and $d_0$.
\end{lemma}

For the other contribution to $\rho_\gamma$, we need the following two results.

\begin{lemma}[$\rho_\gamma^{(2)}$ expansion for moderate temperatures]
\label{lem:rho_gamma1moderate}
Let $\sigma_0\geq0$ and $d_0\geq0$ be fixed constants, and let $-1\leq\theta\leq0$, $0\leq d\leq d_0$, $0\leq\sigma\leq\sigma_0$, $\phi>0$ and $0\leq\rho_0\leq\rho$. Assume $\rho_0a/\phi^2=\sigma/8\pi$ and let $\delta=d\phi^2$ and $t_0=\theta\rho_0$. For $\phi^2/T\ll1$, we have
\[
\ba 
\rho_\gamma^{(2)}&=T^{3/2}I_4(d,\sigma,\theta,\phi/\sqrt{T})+O\left(T^{5/2}a^2(\rho^{1/3}a)^{-3/8}\right)\\
&=\rho_{\rm{fc}}-\frac{1}{8\pi}\left(\frac{\phi^2}{T}\right)^{1/2}T^{3/2}\left(\sqrt{d+2(1+\theta)\sigma}+\sqrt{d}\right)\\
&\quad+o\left(T\phi\right)+O\left(T^{5/2}a^2(\rho^{1/3}a)^{-3/8}\right),
\ea
\]
The error in the first line only depends on $\sigma_0$, the one in the second line on $\sigma_0$ and $d_0$.
\end{lemma}

\begin{lemma}[$\rho_\gamma^{(2)}$ expansion for low temperatures]
\label{lem:rho_gamma1verylow}
Let $0\leq\rho_0\leq\rho$, $d\geq0$, $\sigma=8\pi$ and $t_0=\theta=0$. Let $\delta=d\rho_0 a=d\phi^2$. Then, for $\phi a=\sqrt{\rho_0a^3}\ll1$ while $\phi^2/T=\rho_0a/T\geq O(1)$, we have
\[
\rho_\gamma^{(2)}=T^{3/2}I_4(d,8\pi,0,\sqrt{\rho_0a/T})+o((\rho_0 a)^{3/2}). 
\]
The error is uniform in $d\geq0$ and $\rho_0$.
\end{lemma}

A similar preliminary analysis can be done for the energy terms in Corollary \ref{energyexpressions}. 

\begin{lemma}[$F^{(1)}$ approximation]
\label{lem:lowerorderenergyterm}
Let $\sigma_0\geq0$ and $d_0\geq0$ be fixed constants, and let $-1\leq\theta\leq0$, $0\leq d\leq d_0$, $0\leq\sigma\leq\sigma_0$, $\phi>0$ and $0\leq\rho_0\leq\rho$. Assume $\rho_0a/\phi^2=\sigma/8\pi$ and let $\delta=d\phi^2$ and $t_0=\theta\rho_0$. For $\phi a\ll1$, we have
\[
F^{(1)}=\phi^{5}\frac{1}{2} I_1(d,\sigma,\theta)+o\left(\phi^5\right).
\]
The error is depends only on $\sigma_0$ and $d_0$.
\end{lemma}
For the second term we will need the following two lemmas.

\begin{lemma}[$F^{(2)}$ expansion for moderate temperatures]
\label{lem:leadingorderenergytermmoderate}
Let $\sigma_0\geq0$ and $d_0\geq0$ be fixed constants, and let $-1\leq\theta\leq0$, $0\leq d\leq d_0$, $0\leq\sigma\leq\sigma_0$, $\phi>0$ and $0\leq\rho_0\leq\rho$. Assume $\rho_0a/\phi^2=\sigma/8\pi$ and let $\delta=d\phi^2$ and $t_0=\theta\rho_0$. For $\phi^2/T\ll1$, we have
\[
\begin{aligned}
F^{(2)}&=T^{5/2}I_2(d,\sigma,\theta,\phi/\sqrt{T})+O\left(T^{5/2}\phi^2a^2(\rho^{1/3}a)^{-1/4}\right)\\
&=T^{5/2}f_{\rm{min}}+\left(\frac{\phi^2}{T}\right)T\rho_{\rm{fc}}\left(d+(1+\theta)\sigma\right)\\
&\quad-\frac{1}{12\pi}\left(\frac{\phi^2}{T}\right)^{3/2}T^{5/2}\left((d+2(1+\theta)\sigma)^{3/2}+d^{3/2}\right)\\
&\quad+o\left(T\phi^3\right)+O\left(T^{5/2}\phi^2a^2(\rho^{1/3}a)^{-1/4}\right).
\end{aligned}
\]
The error in the first line only depends on $\sigma_0$, the one in the second line on $\sigma_0$ and $d_0$.
\end{lemma}

\begin{lemma}[$F^{(2)}$ expansion for low temperatures]
\label{lem:leadingorderenergytermverylow}
Let $0\leq\rho_0\leq\rho$, $d\geq0$, $\sigma=8\pi$ and $t_0=\theta=0$. Let $\delta=d\rho_0 a=d\phi^2$. Then, for $\phi a=\sqrt{\rho_0a^3}\ll1$ while $\phi^2/T=\rho_0a/T\geq O(1)$, we have
\[
F^{(2)}=T^{5/2}I_2(d,8\pi,0,\sqrt{\rho_0a/T})+o((\rho_0 a)^{5/2}).
\]
The error is uniform in $d\geq0$ and $\rho_0$.
\end{lemma}

We also prove two lemmas for the error terms \eqref{def:errorterms} and \eqref{def:E_5} for minimizers of the form stated in Lemma \ref{prop:simplfunctsol}.

\begin{lemma}[Error estimates for moderate temperatures]
\label{lem:errgammastar}
Let $\sigma_0\geq0$ and $d_0\geq0$ be fixed constants, and let $-1\leq\theta\leq0$, $0\leq d\leq d_0$, $0\leq\sigma\leq\sigma_0$, $\phi>0$ and $0\leq\rho_0\leq\rho$. Assume $\rho_0a/\phi^2=\sigma/8\pi$ and let $\delta=d\phi^2$ and $t_0=\theta\rho_0$. For $\phi a\ll1$, we have
\[
\left(E_2+E_3\right)(\gamma^{\rho_0,\delta},\alpha^{\rho_0,\delta},\rho_0)=O\left(Ta^3\rho\rho_0+a\rho_0\phi^3\right),
\]
and 
\[
E_4(\gamma^{\rho_0,\delta},\alpha^{\rho_0,\delta},\rho_0)=O(Ta^3\rho^2+a\rho\phi^3).
\]
The error depends only on $\sigma_0$ and $d_0$.
\end{lemma}

\begin{lemma}[Error estimates for low temperatures]
\label{lem:errgammastar2}
Let $d_0$ be a fixed constants, and let $0\leq\rho_0\leq\rho$, $0\leq d\leq d_0$, $\sigma=8\pi$ and $t_0=\theta=0$.  Let $\delta=d\rho_0 a=d\phi^2$. Then, for $\phi a=\sqrt{\rho_0a^3}\ll1$ while $\rho a/T\geq O(1)$, we have
\[
\left(E_2+E_3+E_4\right)(\gamma^{\rho_0,\delta},\alpha^{\rho_0,\delta},\rho_0)=o((\rho a)^{5/2}).
\]
The error is uniform in $d\geq0$.
\end{lemma}

We also prove a final lemma which will later be used to treat the error term $E_1$. Note that the reason we consider the function $f$ below is that $\alpha^{\rho_0,\delta}-\alpha_0=-(2\pi)^{3}t_0\delta_0-f$.

\begin{lemma}[Preparation for estimates on $E_1$]
\label{prepE1}
Let $\sigma_0\geq0$ and $d_0\geq0$ be fixed constants, and let $-1\leq\theta\leq0$, $0\leq d\leq d_0$, $0\leq\sigma\leq\sigma_0$, $\phi>0$ and $0\leq\rho_0\leq\rho$. Assume $\rho_0a/\phi^2=\sigma/8\pi$, $\phi a\ll1$, and let $\delta=d\phi^2$, $t_0=\theta\rho_0$. We define
\begin{equation}
\label{reff}
\ba
f(p)&:=(\rho_0+t_0)\left(\frac{\beta(p)}{TG(p)}- \frac{1}{2 p^2}\right)\widehat{Vw}(p).\\
&=\frac12(\rho_0+t_0)\widehat{Vw}(p)\left[\frac{1}{TG}-\frac{1}{p^2}\right]+\frac{(\rho_0+t_0)\widehat{Vw}(p)}{TG(e^G-1)}.
\ea
\end{equation}
For $\phi^2/T\ll1$, we have
\[
\int f(p)dp=T\phi(1+\theta)\sigma\frac{2\pi^2}{\sqrt{d+2(1+\theta)\sigma}+\sqrt{d}}+o(T\phi),
\]
as well as
\[
\int|f(p)|dp\leq CT\phi \hspace{1cm}\text{and}\hspace{1cm} \int\limits_{|p|>\sqrt{T}} |f(p)|dp\leq C\phi^3.
\]
For $\phi^2/T\geq O(1)$, we have
\[
\int|f(p)|dp\leq C\phi^3.
\]
The errors above depend only on $d_0$.
\end{lemma}

Note that some lemmas above assume that $d$ is bounded. In Subsections \ref{sec:crittemp} and \ref{subsectionenexp}, we will argue that this can be assumed. For Subsection \ref{subsectionenexp}, we will need the following lemma to do this.

\begin{lemma}
\label{fullintegralbehaviour}
Let $0\leq\rho_0\leq\rho$, $d\geq0$, $\sigma=8\pi$ and $t_0=\theta=0$. Let $\delta=d\rho_0 a=d\phi^2$. For $d\gg1$, we have
\[
F^{(1)}-d(\rho_0a)\rho^{(1)}_\gamma\geq C\min\{d^{1/2}(\rho_0a)^{5/2},a^{-1}(\rho_0a)^2\}.
\]
Also, $\rho^{(1)}_\gamma\leq C(\rho_0a)^{3/2}$ and $\rho^{(1)}_\gamma\to0$ as $d\to\infty$. 
\end{lemma}
 
The proof of all lemmas stated above can be found in Appendix \ref{app:intexp}.

\subsection{Proof of Theorems \ref{thm:cancrittemp} and \ref{thm:grandcancrittemp}}
\label{sec:crittemp}
According to the a priori estimate in Proposition \ref{prop:criticalregion}, it suffices to zoom in on
\begin{equation}
\label{def:critreg}
|\rho-\rho_{\rm{fc}}|<C\rho(\rho^{1/3} a)
\end{equation}
within the region \eqref{dillim} to study the critical temperature: for larger $\rho$ there is a condensate, and for smaller $\rho$ there is none. Note that $\rho a/T\ll1$ in this region, which was described as `moderate temperatures' in the previous section. We actually have more a priori information: Lemma \ref{lemub} states that $\rho_0$ is of order $\rho(\rho^{1/3} a)$, i.e.\ of order $T^2a$.

For the non-interacting gas, the critical density is of order $T^{3/2}$. Since we are considering a weakly-interacting gas (through the dilute limit), one expects to again obtain an approximate critical density of order $T^{3/2}$.  We therefore write
\begin{equation}
\label{kdef}
\rho=\rho_{\rm{fc}}+\frac{k}{8\pi}T^2a 
\end{equation}
for a dimensionless parameter $k$ (which is bounded in the region \eqref{def:critreg}). Note that $T^2a = T^{3/2}(\sqrt{T}a)\ll  T^{3/2}$, and so $T^2a$ is indeed a lower order correction to $\rho_{\rm{fc}}$. We also consider
\begin{equation}
\label{def:rho0formmoderateT}
\rho_0=\frac{\sigma}{8\pi}T^2a \hspace{2cm} \phi=Ta \hspace{2cm} \delta=dT^2a^2
\end{equation}
for some dimensionless parameters $d\geq0$ and $0\leq\sigma\leq C$. It suffices to consider bounded $\sigma$ by Lemma \ref{lemub}. We will also show that the a priori estimates allow us to assume that $d$ is bounded. This gives access to the lemmas in the previous section since $\phi a=\phi^2/T=Ta^2\ll1$ in the dilute limit. 
Finally, we write
\begin{equation}
\label{def:deltat0formmoderateT2}
t_0=\frac{\tau}{8\pi}T^2a=\left(\frac{\tau}{\sigma}\right) \rho_0=\theta\rho_0,
\end{equation}
where $\tau=\theta\sigma\in\mathbb{R}$ is dimensionless.

We are free to choose $-\rho_0\leq t_0\leq0$ depending on $\rho_0$ and $\delta$, as this was simply a parameter entering in Lemma \ref{prop:simplandfullfreeenergycomp} and the definition of $\alpha_0$ (see \eqref{def:alpha_0}). To be able to prove that the error term $E_1$ is indeed small for the $\alpha^{\rho_0,\delta}$ from Lemma \ref{prop:simplfunctsol}, we will choose $t_0$ such that the self-consistent equation 
\begin{equation}
\label{def:t0choice}
\int (\alpha^{\rho_0,\delta}-\alpha_0)=0
\end{equation}
is satisfied. The following lemma confirms that this choice implies that the error $E_1$ is small. It also shows that the equation above leads to a concrete equation for $\tau$ in terms of $\sigma\geq0$ and $d\geq0$, which implies that $-\sigma\leq\tau\leq0$, i.e.\ $-1\leq\theta\leq0$.

\begin{lemma}[Self-consistent equation for $t_0$ and estimate on $E_1$]
\label{lem:alphaalpha_0small}
Under the assumptions introduced at the start of this subsection, in particular the self-consistent equation \eqref{def:t0choice} and $\sqrt{T}a\ll1$, we have 
\begin{equation} 
\label{eq:tauselfconseq}
\tau=-\frac{2(\sigma+\tau)}{\sqrt{d+2(\sigma+\tau)}+\sqrt{d}}+o(1).
\end{equation}
This equation has a unique solution for every $d\geq0$ and $\sigma\geq0$, and it satisfies $-\sigma\leq\tau\leq0$. We also have
\[
\iint(\alpha^{\rho_0,\delta}-\alpha_0)(p)\widehat{V}(p-q)(\alpha^{\rho_0,\delta}-\alpha_0)(q)dpdq= o(T^4a^3).
\]
The errors above holds uniformly in $\sigma$ and $d$ as long as they are bounded.
\end{lemma}
\begin{proof}
\textit{Step 1.}
The self-consistent equation \eqref{def:t0choice} says $(2\pi)^{3}t_0=-\int f$, with $f$ as in Lemma \ref{prepE1}, so by using that lemma and the assumptions introduced at the start of this subsection, we conclude that \eqref{eq:tauselfconseq} holds.
To see that it always has a solution in $[-\sigma,0]$, we rewrite the equation as 
\[
\tau\left(\sqrt{d+2(\sigma+\tau)}+\sqrt{d}\right)+2(\sigma+\tau)=0,
\]
and note that the left-hand side is a continuous function which goes from $-2\sigma\sqrt{d}\leq0$ to $2\sigma\geq0$ as $\tau$ goes from $-\sigma$ to $0$.\\
\textit{Step 2.} 
We use Lemma \ref{prepE1} again to conclude that 
\[
\int |f(p)|dp\leq CT\phi=CT^2a,\hspace{1cm} \int\limits_{|p|>\sqrt{T}} |f(p)|dp\leq C\phi^3=CT^3a^3.
\]
Since $\alpha^{\rho_0,\delta}-\alpha_0=-(2\pi)^{3}t_0\delta_0-f$ and $\int\alpha^{\rho_0,\delta}-\alpha_0=0$ by assumption, we have
\[
\begin{aligned}
&\left|\int(\alpha^{\rho_0,\delta}-\alpha_0)(p)\widehat{V}(p-q)(\alpha^{\rho_0,\delta}-\alpha_0)(q)dpdq\right|\\
&=\left|\int(\alpha^{\rho_0,\delta}-\alpha_0)(p)\widehat{V}(p-q)(\alpha^{\rho_0,\delta}-\alpha_0)(q)dpdq-\widehat{V}(0)\left(\int\alpha^{\rho_0,\delta}-\alpha_0\right)^2\right|\\
&\leq2(2\pi)^3|t_0|\int|\widehat{V}(p)-\widehat{V}(0)||f(p)|dp +\int |f(p)||\widehat{V}(p-q)-\widehat{V}(0)||f(q)|dpdq\\
&\leq C|t_0|T^3a^4+CT^5a^5=o(T^4a^3),
\end{aligned}
\]
where we have used the fact that $|\widehat{V}(p)-\widehat{V}(0)|\leq Ca^3T$ for $|p|\leq\sqrt{T}$ and $|\widehat{V}(p)|\leq Ca$ for all $p$.
\end{proof}

Before we prove the main theorem, we state a final error estimate. Its proof can be found in Appendix \ref{app:intexp}.
\begin{lemma}[Estimate on $E_5$]
\label{lem:errgammastar3}
Under the assumptions introduced at the start of this subsection, in particular $\sqrt{T}a\ll1$, we have 
\[
E_5(\gamma^{\rho_0,\delta},\alpha^{\rho_0,\delta},\rho_0)= o(T^4a^3).
\]
This holds uniformly in $d$ and $\sigma$ as long as they are bounded.
\end{lemma}

We are now ready to prove the first main theorem of the paper, which gives an expression for the critical temperature.
\begin{proof}[Proof of Theorem \ref{thm:cancrittemp}]
We will work with the notation introduced at the start of this section. We again refer to Proposition \ref{prop:criticalregion}, which contains the desired conclusion outside this region, so that we can restrict to the region \eqref{def:critreg}. We also recall Lemma \ref{lemub}, which implies $\rho_0\leq CT^2a$, so that we can consider $\sigma$ to be bounded. 

The proof will proceed as follows. In step 1, we will calculate the simplified minimal energy as a function of $\rho$ and $\rho_0$. In step 2, we discuss the precise relation between the minimization problem of the simplified and canonical functionals. In step 3, we prove the theorem by minimizing the simplified energy in $0\leq\rho_0\leq\rho$.\\

\textit{Step 1a.} We would like to calculate the simplified energy for $(\gamma^{\rho_0,\delta},\alpha^{\rho_0,\delta},\rho_0)$. We assume that $t_0(\delta,\rho_0)$ is defined as in Lemma \ref{lem:alphaalpha_0small}. Note that this means that $-1\leq\theta\leq0$ in \eqref{def:deltat0formmoderateT2}, so that we can apply the lemmas from the previous subsection (although we have yet to establish boundedness of $d$ to obtain uniform errors in all cases, which we will do in step 1c.). Corollary \ref{energyexpressions} and Lemmas \ref{lem:lowerorderenergyterm} and \ref{lem:leadingorderenergytermmoderate} together with \eqref{def:rho0formmoderateT} and \eqref{def:deltat0formmoderateT2} imply that for $\delta,\rho_0\geq0$:
\begin{equation}
\label{eq:simplfreeen1}
\begin{aligned}
&\mathcal{F}^{\rm{sim}}(\gamma^{\rho_0,\delta},\alpha^{\rho_0,\delta},\rho_0)=\left(\mathcal{F}^{\rm{s}}(\gamma^{\rho_0,\delta},\alpha^{\rho_0,\delta},\rho_0)+\delta\rho_{\gamma^{\rho_0,\delta}}\right)-\delta\rho_{\gamma^{\rho_0,\delta}}\\
&\quad +\widehat{V}(0)\rho^2+(12\pi a-\widehat{V}(0))\rho_0^2-8\pi a\rho\rho_0-4\pi a t_0^2-8\pi a t_0(\rho-\rho_0)\\
&=T^{5/2}f_{\rm{min}}- T^2 a^2(\rho-\rho_{\rm{fc}})(\sigma+\tau)+\widehat{V}(0)\rho^2\\
&\quad+T^4 a^3\Big[\frac{d}{8\pi}\left(\sqrt{d+2(\sigma+\tau)}+\sqrt{d}\right)-\frac{(d+2(\sigma+\tau))^{3/2}+d^{3/2}}{12\pi}\\ 
&\hspace{3cm}+\frac{\tau\sigma}{8\pi}-\frac{\tau^2}{16\pi}+(12\pi-\nu)\frac{\sigma^2}{64\pi^2}\Big]+ o\left(T^4 a^3 \right),
\end{aligned} 
\end{equation}
where we also used that according to Lemmas \ref{lem:rhogamma2contribution} and \ref{lem:rho_gamma1moderate}:
\begin{equation}
\label{rhogammasomeeq}
\rho_{\gamma^{\rho_0,\delta}}=\rho_{\rm{fc}}-\frac{T^2 a}{8\pi}\left(\sqrt{d+2(\sigma+\tau)}+\sqrt{d}\right)+o(T^2 a).
\end{equation}
The expressions above really only depend $d$ and $\sigma$, since $\tau$ satisfies \eqref{eq:tauselfconseq}. However, we are interested in rewriting the expression fully in terms of $\sigma$ and $k$. After all, we would like to investigate the nature of $\sigma$ (which defines $\rho_0$) for given $k$ (which defines $\rho$). First note that from the equation
\[
\rho=\rho_0+\rho_{\gamma^{\rho_0,\delta}}=\frac{\sigma}{8\pi}T^2 a+\rho_{\rm{fc}}-\frac{T^2 a}{8\pi}\left(\sqrt{d+2(\sigma+\tau)}+\sqrt{d}\right)+o(T^2 a),
\]
we obtain
\begin{equation}
\label{expression234}
\sqrt{d+2(\sigma+\tau)}+\sqrt{d}=\sigma+8\pi\frac{\rho_{\rm{fc}}-\rho}{T^2 a}=\sigma-k,
\end{equation}
where $k$ is defined in \eqref{kdef}.
This yields
\[
d=\left(\frac{(\sigma-k)^2-2(\sigma+\tau)}{2(\sigma-k)}\right)^2.
\]
We can also rewrite $\tau$ in terms of $\sigma$ and $k$ by using \eqref{eq:tauselfconseq} and \eqref{expression234}:
\bq
\label{eq:tausigmaksol}
 \tau=\frac{2\sigma}{k-\sigma-2}+o(1)\qquad \text{and}\qquad \sigma+\tau=\frac{\sigma(k-\sigma)}{k-\sigma-2}+o(1).
\eq
We plug these expressions into \eqref{eq:simplfreeen1} to obtain
\bq
\label{eq:simplfreeen3}
\ba
&\mathcal{F}^{\rm{sim}}(\gamma^{\rho,\rho_0},\alpha^{\rho,\rho_0},\rho_0)=T^{5/2}f_{\rm{min}}+\widehat{V}(0)\rho^2 
+T^4a^3\left[\frac{1}{8\pi}\left(\frac{(\sigma-k)^3}{12}\right.\right. \\& \left. \left. \qquad -\sigma^2\Big(\frac{1}{2} +\frac{1}{2+\sigma-k}\Big)\right)   -(\nu-8\pi)\frac{\sigma^2}{(8\pi)^2}\right]+o(T^4 a^3),
\ea 
\eq
where we now write $\gamma^{\rho_0,\rho}$ for the $\gamma^{\rho_0,\delta}$ that satisfies $\rho_{\gamma^{\rho_0,\delta}}+\rho_0=\rho$.
This can only be done for certain $\sigma$ and $k$: it was only for $\delta\geq0$ that we were able to obtain minimizers of this form. 

\textit{Step 1b.} We now determine for which $\sigma$ and $k$ \eqref{eq:simplfreeen3} holds.
Using \eqref{expression234} and the equation for $\tau$ \eqref{eq:tauselfconseq}, we know that, given a $\rho_0=\sigma \frac{T^2 a}{8\pi}$, minimizing the functional for some $d\geq 0$ leads to a minimizer with 
\[
\rho=\rho_{\rm{fc}}+\frac{T^2 a}{8\pi}\left(1+\sigma-\sqrt{d}-\sqrt{1+2\sigma+d+2\sqrt{d}}\right)+o(T^2 a),
\]
The above expression is maximal for $d=0$. Its value at this point is significant: fixing some $\rho_0$, we know that this is the maximal $\rho$ for which we will be able to find a minimizer to the simplified functional. This maximal $\rho$ is
\[
\rho_{\rm{max}}(\sigma)=\rho_{\rm{fc}}+\frac{T^2 a}{8\pi}k_{\rm{max}}(\sigma)+o(T^2 a),
\]
where we defined 
\[
k_{\rm{max}}(\sigma)=1+\sigma-\sqrt{1+2\sigma}. 
\]
Fixing some $k$, and considering all $\sigma\geq 0$, we can find out that \eqref{eq:simplfreeen3} holds whenever
\begin{equation}
\label{interval}
\sigma\in I(k):=\left\{
	\begin{array}{ll}
		[0,\infty)  & \mbox{if } k\leq 0 \\
		\left[k+\sqrt{2k},\infty\right) & \mbox{if } k>0
	\end{array}
\right..
\end{equation}
Summarizing, it is for these $\sigma$ and $k$ that there exists a $(\gamma^{\rho_0,\rho},\alpha^{\rho_0,\rho})$.

\textit{Step 1c.} We will be interested in using \eqref{eq:simplfreeen3} as a lower bound for the energy, where the error is uniform in $\sigma$ and $k$. We would now like to show that $d$ is bounded, so that we obtain uniform errors in \eqref{eq:simplfreeen1}, \eqref{rhogammasomeeq} and consequently \eqref{eq:simplfreeen3}.

As noted at the start of the proof, it suffices to consider $\rho_0\leq CT^2a$. Combined with \eqref{def:critreg}, this tells us that $\rho_\gamma\geq\rho-C_0T^2a$ for some constant $C_0$. We claim that it suffices to restrict to $d\leq d_0$, which is chosen such that
\[
\frac{2\sqrt{d_0}}{8\pi}\geq 2C_0.
\]
To see this, consider $d>d_0$. Because $\rho_{\gamma^{\rho_0,\delta}}$ is decreasing in $\delta$ by the structure of the minimization problem in Lemma \ref{prop:simplfunctsol}, we know that 
\[
\ba
\rho_{\gamma^{\rho_0,d}}\leq\rho_{\gamma^{\rho_0,d_0}}&=\rho_{\rm{fc}}-\frac{T^2 a}{8\pi}\left(\sqrt{d_0+2(\sigma+\tau)}+\sqrt{d_0}\right)+o(T^2 a)\\
&\leq \rho_{\rm fc}-2C_0T^2a+o(T^2 a),
\ea
\]
where the error only depends on $d_0$ since we have a priori restricted to bounded $\sigma$. This violates the a priori restriction, confirming that we can restrict to $d\leq d_0$. We have obtained the important conclusion that we can think of the error in \eqref{eq:simplfreeen3} as uniform. \\

\textit{Step 2.}
Our strategy will be to connect \eqref{eq:simplfreeen3} to $\cF^{\rm can}$ using Corollary \ref{lem:differencecriticalregion}. For convenience, we will first assume $k\leq0$, so that all $0\leq\sigma\in I(k)$. 

On the one hand, any potential minimizer $(\gamma,\alpha,\rho_0)$ with $\rho_\gamma+\rho_0=\rho$ will have to satisfy the a priori estimates in Propositions \ref{lm:gvg2} and \ref{prop:impconvgamma}. This means that
\begin{equation}
\label{ineqF1}
\ba
&\cF^{\rm can}(\gamma,\alpha,\rho_0)\\
&\geq\cF^{\rm sim}(\gamma,\alpha,\rho_0)+\frac{\zeta(3/2)\zeta(5/2)}{256\pi^3}\Delta\widehat{V}(0)T^4-(E_2+E_3+E_5)(\gamma,\alpha,\rho_0)\\
&\geq\cF^{\rm sim}(\gamma^{\rho_0,\rho},\alpha^{\rho_0,\rho},\rho_0)+\frac{\zeta(3/2)\zeta(5/2)}{256\pi^3}\Delta\widehat{V}(0)T^4-o(T^4a^3).
\ea
\end{equation}
On the other hand, we have for any $\rho_0$:
\begin{equation}
\label{ineqF2}
\ba
&\inf_{(\gamma,\alpha),\ \rho_0=\rho-\rho_\gamma}\cF^{\rm can}(\gamma,\alpha,\rho_0)\leq \cF^{\rm can}(\gamma^{\rho,\rho_0},\alpha^{\rho,\rho_0},\rho_0)\\
&\leq\cF^{\rm sim}(\gamma^{\rho_0,\rho},\alpha^{\rho_0,\rho},\rho_0)+\frac{\zeta(3/2)\zeta(5/2)}{256\pi^3}\Delta\widehat{V}(0)T^4\\
&\hspace{4cm}+(E_1+E_2+E_3+E_5)(\gamma^{\rho_0,\rho},\alpha^{\rho_0,\rho},\rho_0)\\
&\leq\cF^{\rm sim}(\gamma^{\rho_0,\rho},\alpha^{\rho_0,\rho},\rho_0)+\frac{\zeta(3/2)\zeta(5/2)}{256\pi^3}\Delta\widehat{V}(0)T^4+o(T^4a^3),
\ea
\end{equation}
where we have used Lemmas \ref{lem:errgammastar} and \ref{lem:alphaalpha_0small}. The errors are uniform since we assume $d$ and $\sigma$ to be bounded.

We conclude that the energy of any potential minimizer matches \eqref{eq:simplfreeen3} (up to the constant term and a small error). However, for any $\rho_0$ the expression \eqref{eq:simplfreeen3} also provides an upper bound. Therefore, if we find that the minimizing $\sigma$ of \eqref{eq:simplfreeen3} is non-zero, then the same should hold for the real minimizer. If the approximate minimizer is zero, we can only conclude that the real minimizer is approximately zero because of the small error. We will therefore need an extra step in this case\\

\textit{Step 3a.} 
We now analyse \eqref{eq:simplfreeen3} for given $k\leq0$ and $\nu$ and find out whether its minimum $\sigma_{\rm min}$ is zero or not.

An analysis of \eqref{eq:simplfreeen3} shows that there always is a single $k\leq0$ where the character of the minimizer of changes (for given $\nu$)\footnote{Because \eqref{eq:simplfreeen3} depends on $\nu$ in an easy way, and is independent from $\nu$ for $\sigma=0$, we can see that for every $k\leq0$ there is a $\nu_0(k)\in[8\pi,\infty)$ such that $\sigma_{\rm min}>0$ for $\nu>\nu_0(k)$. Moreover, $\nu_0(k)$ is continuous, monotone decreasing, and equal to $8\pi$ for $k=0$. To reach the desired conclusion, we have to combine this with the following: for every $\nu\geq8\pi$, there exists a $k$ negative enough such that $\sigma_{\rm min}=0$. This can be seen by noting that the derivative in $\sigma$ it is positive for all $\sigma$ when $k$ is negative enough.}, which implies that a function $h_1(\nu)$ exists. We can also see that the critical $k$ decreases with $\nu$. For the limit $\nu\to8\pi$, we numerically verify that the minimizing $\sigma_{\rm{min}}$ approximately satisfies
\begin{equation}
\label{criticalsigm}
\sigma_{\rm{min}}=\left\{
	\begin{array}{ll}
		0 & \mbox{if } k< -1.28 \\
		>0 & \mbox{if } k> -1.28
	\end{array}
\right..
\end{equation}
This is illustrated by Figure \ref{plots} below, which shows \eqref{eq:simplfreeen3} for three values of $k$.
\begin{figure}
\includegraphics[scale=0.42]{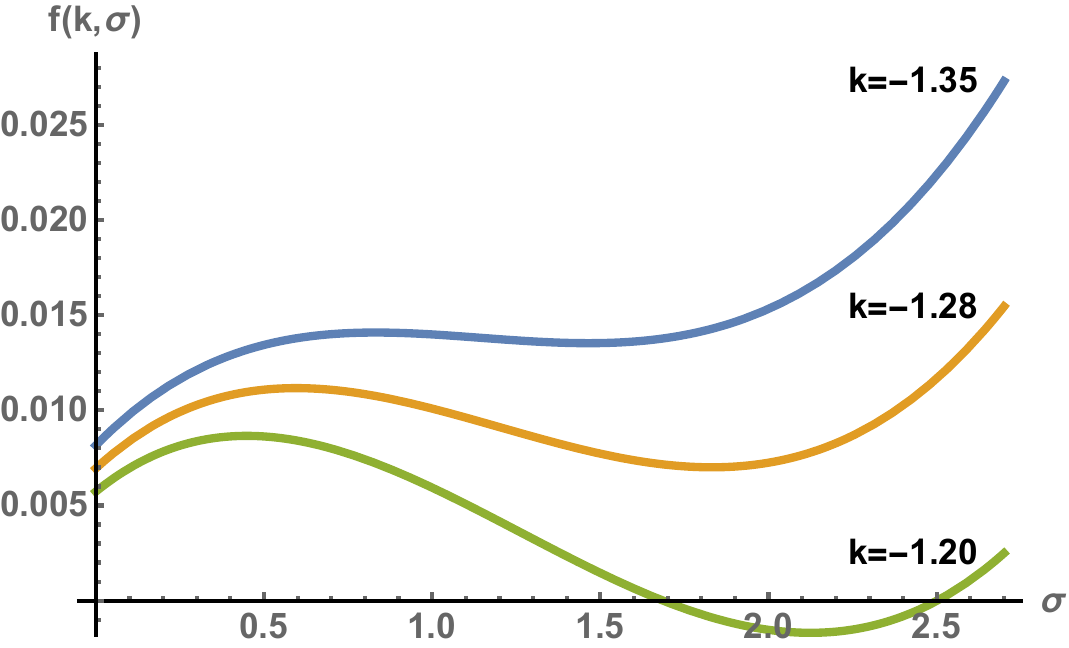}
\caption{Plots of the part of the free energy that depends on $k$ and $\sigma$ (i.e.\ between the square brackets in \eqref{eq:simplfreeen3}, denoted by $f(k,\sigma)$ in the plot) for three values of $k$. For $k=-1.35$, $\sigma=\rho_0=0$ gives the lowest energy: no BEC. For $k=-1.20$, the minimum occurs at some $\rho_0>0$: BEC. The critical value is $k_c=-1.28$, where both $\sigma=0$ and $\sigma=1.83$ are minimizers.}
\label{plots}
\end{figure}

Using the definition of $k$ \eqref{kdef}, we conclude that the point where the nature of the minimizer changes is 
\[
\ba
\rho_{\rm c}=&\rho_{\rm fc}\left(1-\frac{1.28}{8\pi}\left(\frac{\zeta(3/2)}{8\pi^{3/2}}\right)^{-4/3}\rho_{\rm fc}^{1/3}a+o(\rho_{\rm fc}^{1/3}a)\right)\\
&\hspace{6cm}=\rho_{\rm fc}\left(1-2.24\rho_{\rm fc}^{1/3}a+o(\rho_{\rm fc}^{1/3}a)\right).
\ea
\]
We can also turn this into a criterion for the critical temperature. Given $\rho$ we know that the critical temperature $T_{\rm c}$
satisfies the equation above where $\rho_{\rm fc}=n_{\rm fc}T_{\rm c}^{3/2}$,
where we calculated the constant $n_{\rm fc}$ in \eqref{constantfcdensity} (although it plays no role here). The free critical temperature would satisfy $\rho_{\rm c}=n_{\rm fc}T_{\rm fc}^{3/2}$. Hence, we have
\[
n_{\rm fc}T_{\rm fc}^{3/2}=n_{\rm fc}T_{\rm c}^{3/2}\left(1-2.24(\rho^{1/3}a)+o(\rho^{1/3}a)\right),
\]
since we can write $\rho$ instead of $\rho_{\rm c}$ to leading order. In conclusion,
\[
\ba
T_{\rm c}&=T_{\rm fc}\left(1-2.24(\rho^{1/3}a)+o(\rho^{1/3}a)\right)^{-2/3}\\
&=T_{\rm fc}\left(1+1.49(\rho^{1/3}a)+o(\rho^{1/3}a)\right).
\ea
\]

\textit{Step 3b.} For those values of $\rho$ where the minimizer of the approximate functional has $\rho_0=0$, we can only conclude that the exact minimizing $\rho_0$ is approximately zero. Because our energy approximation is accurate up to orders $T^4a^3$, we can only conclude $\rho_0=o(T^2a)$. We will need an extra argument to show that the energy increases for smaller $\rho_0$, which would then imply that the exact minimizer really is $\rho_0=0$. 

Fixing $\rho$, first define
\[
F_\rho(\rho_0)=\inf_{\int\gamma=\rho-\rho_0}\cF^{\rm can}(\gamma,\alpha,\rho_0).
\]
Note that it suffices to show there exists a $c_0>0$ such that 
\begin{equation}
\label{lb1}
F_\rho(\rho_0)\geq F_\rho(0)+\frac12c_0\rho_0T^2a^2(1-o(1))-2\rho_0^2\widehat{V}(0).
\end{equation}

To prove this lower bound, we first minimize the terms in $\alpha$, and use Proposition \ref{lm:gvg1}:
\begin{equation}
\label{somestep}
\begin{aligned}
\mathcal{F}^{\rm can}(\gamma,\alpha,\rho_0)&\geq \mathcal{F}^{\rm can}(\gamma,0,0)+2\widehat{V}(0)\rho_0\rho\\
&\quad-2\widehat{V}(0)\rho^2_0-c\rho_0T^2a^2(\sqrt{T}a)^{1/2}.
\end{aligned}
\end{equation}

To prove \eqref{lb1} from \eqref{somestep}, we need to study
\[
f_\rho(\rho_0)=\inf_{\int\gamma=\rho-\rho_0} \mathcal{F}^{\rm can}(\gamma,0,0)+2\rho_0\rho\widehat{V}(0),
\]
which is convex in $\rho_0$. 

We now use \eqref{ineqF1} and \eqref{ineqF2} to approximate the functional by $\mathcal{F}^{\rm sim}$ and go back again, denoting the constant term as $C^{\rm sim}T^4a^3$ and keeping in mind that the minimizer approximately has $\rho_0=0$ so that \eqref{ineqF1} does hold. We also apply \eqref{eq:e0nn2}, noting that $\rho_{\rm fc}-\rho\geq 1.28T^2a$. For $\epsilon>0$, we find
\[
\begin{aligned}
f_\rho(-\epsilon T^2a)&\leq\inf_{\int\gamma=\rho+\epsilon T^2a}\mathcal{F}^{\rm sim}(\gamma,0,0)-2\epsilon T^2a\rho\widehat{V}(0)+C^{\rm sim}T^4a^3+o(T^4a^3)\\
&\leq \inf_{\int\gamma=\rho+\epsilon T^2a}\mathcal{F}_0(\gamma)+\widehat{V}(0)\rho^2+c\epsilon^2T^4a^3+C^{\rm sim}T^4a^3+o(T^4a^3)\\
&\leq \inf_{\int\gamma=\rho}\mathcal{F}_0(\gamma)-(c_0\epsilon-c\epsilon^2)T^4a^3+\widehat{V}(0)\rho^2+C^{\rm sim}T^4a^3+o(T^4a^3)\\
&\leq f_\rho(0)-(c_0\epsilon-c\epsilon^2)T^4a^3+o(T^4a^3).
\end{aligned}
\]
We therefore conclude that there exists an $\epsilon_0$ small enough such that 
\[
f_\rho(-\epsilon_0 T^2a)\leq  f_\rho(0)+\frac12c_0(-\epsilon_0 T^2a)T^2a^2.
\]
Convexity of $f_\rho$ now implies that for $\rho_0\geq0$
\[
f_\rho(\rho_0)\geq f_\rho(0)+\frac12c_0\rho_0T^2a^2.
\]
This, as well as taking the infimum over $\gamma$ with $\int\gamma=\rho-\rho_0$ in \eqref{somestep}, now gives the desired lower bound \eqref{lb1}.

\textit{Step 3c.}
The theorem is still not quite proved, as we still have to show that the minimizing $\sigma$ is strictly positive for $k>0$, which corresponds to $\rho>\rho_{\rm fc}$. For $\sigma=0$, we cannot use the simplified energy \eqref{eq:simplfreeen3} because of the problem discussed in step 1b, but we can still use the first step in the lower bound \eqref{ineqF1}: if the minimum occurs at $\rho_0=0$, we know that
\[
\ba
&\inf_{(\gamma,\alpha),\ \rho_\gamma=\rho}\cF^{\rm can}(\gamma,\alpha,0)-\frac{\zeta(3/2)\zeta(5/2)}{256\pi^3}\Delta\widehat{V}(0)T^4\\
&\quad\geq\inf_{(\gamma,\alpha),\ \rho_\gamma=\rho}\cF^{\rm sim}(\gamma,\alpha,0)+o(T^4a^3)=T^{5/2}f_{\rm min}+\widehat{V}(0)\rho^2-o(T^4a^3),
\ea
\]
where $t_0=0$ for $\rho_0=0$ (which is consistent with \eqref{eq:tausigmaksol}).
Since \eqref{eq:simplfreeen3} holds at $\sigma=k+\sqrt{2k}\in I(k)$, we can see that it has a simplified energy of
\[
T^{5/2}f_{\rm{min}}+\widehat{V}(0)\rho^2 
+T^4a^3\left[-\frac{k \left(3 \sqrt{2} k^{3/2}+20 k+23 \sqrt{2} \sqrt{k}+18\right)}{24 \pi  \left(\sqrt{2} \sqrt{k}+2\right)}\right],
\]
which is lower than the value at $\sigma=0$. Using the upper bound \eqref{ineqF2}, we conclude that the minimizer cannot have $\rho_0=0$ when $\rho>\rho_{\rm fc}$. 
\end{proof}

\begin{proof}[Proof of Theorem \ref{thm:grandcancrittemp}]
\textit{Step 1.}
We now turn to the grand-canonical problem. That means that we should analyse the structure of minimizers of
\begin{equation}
\label{canminimiz}
\inf_{\rho\geq0}\left[\inf_{(\gamma,\alpha,\rho_0),\ \rho_\gamma+\rho_0=\rho}\mathcal{F}^{\rm{can}}(\gamma,\alpha,\rho_0)-\mu\rho\right]
\end{equation}
for given $\mu\in\mathbb{R}$. This requires that we calculate the canonical free energy for any given $\rho$, but we note that it again suffices to only calculate it for \eqref{def:critreg}, i.e.\ $|\rho-\rho_{\rm{fc}}|<C\rho(\rho^{1/3} a)$. By the a priori result from Proposition \ref{prop:criticalregion}, we know that if minimizer has a smaller $\rho$, it has $\rho_0=0$, and if it has a bigger $\rho$, it has $\rho_0>0$. Since the minimizing $\rho$ increases with $\mu$, this fits with the statement of the theorem.

In the region around the critical temperature, it seems natural to use the bounds \eqref{ineqF1} and \eqref{ineqF2} and simply minimize \eqref{eq:simplfreeen3}, but we only have these bounds for $\sigma\in I(k)$ (see \eqref{eq:tausigmaksol} and \eqref{interval}). In fact, the simplified functional has so far only been defined in this region as we have only made a choice for $t_0$ for $\delta,\rho_0\geq0$. To solve this problem, we now define
\[
\tau(k,\sigma)=1-\sqrt{1+2\sigma}
\]
for $\sigma\in [0,\infty)\backslash I(k)$, which is chosen because it is the value obtained for $\delta=0$. In the spirit of \eqref{ineqF1}, we know that any potential minimizer should satisfy
\begin{equation*}
\begin{aligned}
&\cF^{\rm can}(\gamma,\alpha,\rho_0)-\frac{\zeta(3/2)\zeta(5/2)}{256\pi^3}\Delta\widehat{V}(0)T^4\\
&\geq\cF^{\rm sim}(\gamma,\alpha,\rho_0)-(E_2+E_3+E_5)(\gamma,\alpha,\rho_0)\\
&\geq\inf_{(\gamma,\alpha)}\mathcal{F}^{\rm s}(\gamma,\alpha,\rho_0)+\widehat{V}(0)\rho^2-o(T^4a^3)\\
&\hspace{3cm}+(12\pi a-\widehat{V}(0))\rho_0^2-8\pi a\rho\rho_0-4\pi a t_0^2-8\pi at_0(\rho-\rho_0)\\
&=T^{5/2}f_{\rm min}+\widehat{V}(0)\rho^2+T^4a^3\Big[-(\sigma+\tau)\frac{k}{8\pi}-\frac1{12\pi}(2\sigma+2\tau)^{3/2}\\
&\hspace{3cm}+(12\pi-\nu)\left(\frac{\sigma}{8\pi}\right)^2-4\pi\left(\frac{\tau}{8\pi}\right)^2+\frac{\tau\sigma}{8\pi}\Big]-o(T^4a^3),
\end{aligned}
\end{equation*}
where we have used that the infimum of $\mathcal{F}^{\rm s}$ is attained at $(\gamma^{\rho_0,\delta=0},\alpha^{\rho_0,\delta=0})$, with an energy given by \eqref{eq:simplfreeen1}.
Minimizing this lower bound over $[0,\infty)\backslash I(k)$, we find that the infimum is attained at the boundary, i.e.\ at $\sigma=k+\sqrt{2k}$.
Since the lower bound matches \eqref{eq:simplfreeen3} at this point, we conclude that the minimizer of the canonical free energy has $\sigma\in I(k)$, so that it suffices to minimize \eqref{canminimiz} over $I(k)$ by the upper and lower bounds \eqref{ineqF1} and \eqref{ineqF2}.\\

\textit{Step 2.}
Making the result of the previous step explicit, we now know that for $|\rho-\rho_{\rm{fc}}|<C\rho(\rho^{1/3} a)$:
\[
\begin{aligned}
&\inf_{(\gamma,\alpha,\rho_0),\ \rho_\gamma+\rho_0=\rho}\mathcal{F}^{\rm{can}}(\gamma,\alpha,\rho_0)-\mu\rho\\
&=T^{5/2}f_{\rm{min}}+\widehat{V}(0)\rho_{\rm{fc}}^2 -\mu\rho_{\rm fc} +T^2a^2\left(2\left(\frac\nu{8\pi}\right)\rho_{\rm{fc}}-\frac{\mu}{8\pi a}\right)k\\
&\quad+T^4a^3\inf_{\sigma\in I(k)}\Big[\frac{1}{8\pi}\left(\frac{(\sigma-k)^3}{12}-\sigma^2\Big(\frac{1}{2} +\frac{1}{2+\sigma-k}\Big)\right)\\
& \hspace{7cm}-(\nu-8\pi)\frac{\sigma^2}{(8\pi)^2}+\nu\frac{k^2}{(8\pi)^2}\Big]\\
&\quad+\frac{\zeta(3/2)\zeta(5/2)}{256\pi^3}\Delta\widehat{V}(0)T^4+o\left(T^4a^3\right).
\end{aligned}
\]
To consider the case $\nu\to8\pi$, we show a plot of the function
\begin{equation}
\begin{aligned}
\label{somefunc1}
  g(k)&=\inf_{\sigma\in I(k)}\Big[\frac{1}{8\pi}\left(\frac{(\sigma-k)^3}{12}-\sigma^2\Big(\frac{1}{2} +\frac{1}{2+\sigma-k}\Big)\right)+\frac{k^2}{8\pi}\Big]+0.226k.
\end{aligned}
\end{equation}
in Figure \ref{energypl}.
Here, the value $0.226$ was chosen such that the convex hull is obtained by replacing the
curve between two minima by a constant function. 

\begin{figure}
\includegraphics[scale=0.14]{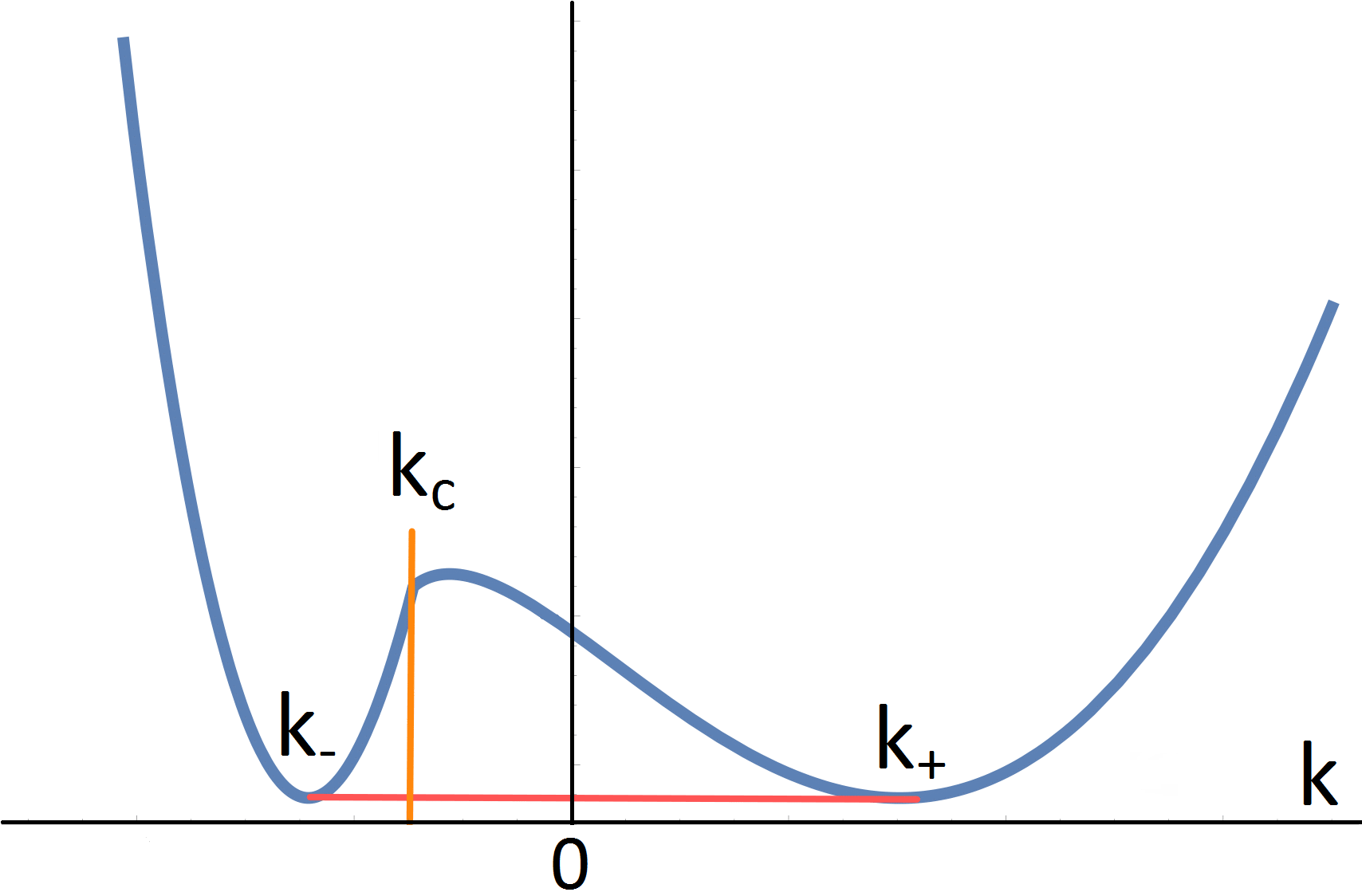}
\caption{The curve shows the energy function $g(k)$ 
in \eqref{somefunc1}. The two minima are \mbox{$k_-=-2.23$} and \mbox{$k_+=3.04$}, and the critical value (shown in orange) is $k_{\rm c}=-1.28$,  which corresponds to the value of $k$ where $\sigma$ jumps to a positive value (see \eqref{criticalsigm}). The derivative has a discontinuity at this point. The energy curve is not convex; the red line indicates the convex hull of the curve.}
\label{energypl}
\end{figure}

The two minima are 
$$
k_-=-2.23 ,\quad k_+=3.04
$$
and the value here is $g(k_\pm)=-0.27$. Hence we have a first-order phase transition where the density jumps between the critical values corresponding to $k_\pm$. This conclusion is unaltered by the fact that we can only determine the energy curve up to a small error.

Note that the minimizer changes from $\rho_0=0$ to $\rho_0>0$ at the jump since $k_-\leq-1.28\leq k_+$. We conclude that the critical chemical potential in the limit $\nu\to8\pi$ is given by
\begin{align*}
\frac{\mu_{\rm c}}{8\pi}&=2\rho_{\rm fc}a-0.226T^2a^2+o(T^2a^2)\\
&=\frac{1}{8\pi}\frac{2\zeta(3/2)}{\sqrt{\pi}}T^{3/2}a\left(1-0.226\cdot8\pi\frac{\sqrt{\pi}}{2\zeta(3/2)}\sqrt{T}a+o(\sqrt{T}a)\right).
\end{align*}
This can also be inverted to yield the critical temperature for $\mu>0$:
\begin{align*}
T_{\rm{c}}&=\left(\frac{\sqrt{\pi}}{2\zeta(3/2)}\right)^{2/3}\left(\frac{\mu}{a}\right)^{2/3}+\frac23\cdot 0.226\cdot 8\pi \left(\frac{\sqrt{\pi}}{2\zeta(3/2)}\right)^{2}\mu+o\left(\mu\right)\\
&=\left(\frac{\sqrt{\pi}}{2\zeta(3/2)}\right)^{2/3}\left(\frac{\mu}{a}\right)^{2/3}+0.44\mu+o(\mu),
\end{align*}
where the expansion is correct for $\mu\geq0$ corresponding to $\rho^{1/3}a\ll 1$. 
An analysis for general $\nu$ (in which case the leading term of $\mu_{\rm c}$ has an extra factor $\nu/8\pi$), combined with the existence of the function $h_1(\nu)$ from the previous theorem, allows the reader to deduce the existence of $h_2(\nu)$.
\end{proof}

\begin{remark}
\label{nonuniqueminimizer}
Note that the existence of two minima shows that the grand canonical functional in general will not have a unique minimizer. As for the canonical case: we have coexistence of the two minimizers (one with $\rho_0=0$ and one with $\rho_0>0$) for $\rho$ between the two values defined by $k_\pm$. This means that at least part of the gas has a condensate for any $k\in[k_-,k_+]$. Hence one could say that (part of) the system is in a condensed phase from $k_-$ onwards.
\end{remark}

\subsection{Proof of Theorems \ref{thm:canfreeenexp} and \ref{thm:expint}}
\label{subsectionenexp}
In this section, we simply set $t_0=0$. We will write $\delta=d\rho_0a=d\phi^2$, with $d\geq0$. Note that this implies that $\sigma=8\pi$ in the lemmas of Subsection \ref{sec:prelapprox}.

\begin{remark}[Properties of the integrals]
\label{remarkpropint}
We will use the following properties of the integrals \eqref{eq:integrals} with $d,s\geq0$:
\begin{itemize}
\item $I_1(d,8\pi,0)-dI_3(d,8\pi,0)$ monotonically increases to infinity in $d$.
\item $I_2(d,8\pi,0,s)-ds^2I_4(d,8\pi,0,s)$ monotonically increases to 0 in both $d$ and $s$ and it is bounded.
\item $I_2(d,8\pi,0,s)$ monotonically increases to 0 in both $d$ and $s$ and it is bounded.
\item $I_4(d,8\pi,0,s)$ monotonically decreases to zero in both $d$ and $s$ and it is bounded.
\end{itemize}
\end{remark}

\begin{proof}[Proof of Theorems \ref{thm:canfreeenexp} and \ref{thm:expint}]
Throughout the proof, we will distinguish between the regions $\rho a/T\ll1$ (`moderate temperatures') and $\rho a/T\geq O(1)$ (`low temperatures'). For simplicity, we aim to write statements with a uniform error $o(T(\rho a)^{3/2}+(\rho a)^{5/2})$, i.e.\ $o(T(\rho a)^{3/2})$ in the first region, and $o((\rho a)^{5/2})$ in the second. Note that an error of $O((\rho a)^{5/2})$ satisfies this for $\rho a/T\ll1$.\\

\textit{Step 1a.} 
As in Subsection \ref{sec:crittemp}, we consider upper and lower bounds.
First assume that $\delta\geq0$ and $\rho_0\geq0$ are such that
\begin{equation}
\label{eqnforrho_0d}
\rho=\rho_0+\rho_{\gamma^{\rho_0,\delta}}.
\end{equation}
Similar to before, this may not always have a solution for given $\rho$ and $\rho_0$.
By Lemma \ref{prop:simplandfullfreeenergycomp} and the a priori estimates in Proposition \ref{lm:gvg1}, we then know that any potential minimizer has to satisfy
\begin{equation}
\label{ineqF3}
\ba
\cF^{\rm can}(\gamma,\alpha,\rho_0)&\geq\cF^{\rm sim}(\gamma,\alpha,\rho_0)-(E_2+E_3+E_4)(\gamma,\alpha,\rho_0)\\
&\geq\cF^{\rm sim}(\gamma^{\rho_0,\delta},\alpha^{\rho_0,\delta},\rho_0)-O((\rho a)^{5/2}).
\ea
\end{equation}
Using Lemma \ref{lem:errgammastar} for $\rho a/T\ll1$, Lemma \ref{lem:errgammastar2} for $\rho a/T\geq O(1)$, and Lemma \ref{prepE1} for both, we find that\footnote{To obtain the estimate on $E_1$, we use $|\int(\alpha-\alpha_0)\widehat{V}(\alpha-\alpha_0)|\leq\widehat{V}(0)(\int|\alpha-\alpha_0|)^2$ and the fact that $|\alpha^{\rho_0,\delta}-\alpha_0|$ is equal to the $|f|$ in the statement of Lemma \ref{prepE1} since $t_0=0$.}
\begin{equation}
\label{ineqF4}
\ba
&\inf_{(\gamma,\alpha),\ \rho_0=\rho-\rho_\gamma}\cF^{\rm can}(\gamma,\alpha,\rho_0)\leq \cF^{\rm can}(\gamma^{\rho_0,\delta},\alpha^{\rho_0,\delta},\rho_0)\\
&\leq\cF^{\rm sim}(\gamma^{\rho_0,\delta},\alpha^{\rho_0,\delta},\rho_0)+(E_1+E_2+E_3+E_4)(\gamma^{\rho_0,\delta},\alpha^{\rho_0,\delta},\rho_0)\\
&\leq\cF^{\rm sim}(\gamma^{\rho_0,\delta},\alpha^{\rho_0,\delta},\rho_0)+o(T(\rho a)^{3/2}+(\rho a)^{5/2}).
\ea
\end{equation}
It is important to realize that we have yet to establish uniformity of the error in the upper bound, whereas the error in the lower bound is uniform.\\

\textit{Step 1b.}
The first line of the lower bound \eqref{ineqF3} allows us to prove the desired conclusion for $T>T_{\rm fc}\left(1+h_1(\nu)\rho^{1/3}a+o(\rho^{1/3}a)\right)$. After all, Theorem \ref{thm:cancrittemp} tells us that the minimizer has $\rho_0=0$ in this region, so that we find that
\[
F^{\rm can}(T,\rho)\geq \inf_{(\gamma,\alpha),\ \rho=\rho_\gamma}\cF^{\rm sim}(\gamma,\alpha,0)=F_0(T,\rho)+\widehat{V}(0)\rho^2-O((\rho a)^{5/2}),
\]
where $F_0(T,\rho)$ is the free energy \eqref{eq:freefreeenergy} of the non-interacting gas. We now note that
\[
\inf_{(\gamma,\alpha)}\cF^{\rm can}(\gamma,\alpha,0)=\inf_{\gamma}\cF^{\rm can}(\gamma,0,0)\leq\inf_{\gamma}\cF^{\rm sim}(\gamma,0,0),
\]
which proves the result in this region.\\

\textit{Step 2.}
Using Lemma \ref{lem:leadingorderenergytermverylow} and the first line of Lemma \ref{lem:leadingorderenergytermmoderate}, we have
\begin{equation}
\label{energyexpression0}
\begin{aligned}
\cF^{\rm sim}(\gamma^{\rho_0,\delta},\alpha^{\rho_0,\delta},\rho_0)&=F^{(1)}+T^{5/2}I_2(d,8\pi,0,\sqrt{\rho_0 a/T})\\
&\quad-d\rho_0a(\rho-\rho_0)\\
&\quad+\widehat{V}(0)\rho^2-8\pi a\rho_0\rho+\rho_0^2(12\pi a-\widehat{V}(0))\\
&\quad+o\left(T(\rho a)^{3/2}+(\rho a)^{5/2}\right),
\end{aligned}
\end{equation}
together with
\begin{equation}
\label{rho0eqn}
\ba
\rho_{\gamma^{\rho_0,\delta}}&=\rho^{(1)}_\gamma+T^{3/2}I_4(d,8\pi,0,\sqrt{\rho_0 a/T})+o\left(T(\rho a)^{1/2}+\left(\rho a\right)^{3/2}\right).
\ea
\end{equation}
In the last line, we have used  Lemma \ref{lem:rho_gamma1verylow} and the first line of Lemma \ref{lem:rho_gamma1moderate}. To use these lemmas, we have distinguished two cases: $\rho a/T\ll 1$, which implies $\rho_0a/T=\phi^2/T\ll1$; and $\rho a/T\geq O(1)$, which implies $\rho_0a/T\geq O(1)$ by the a priori estimate \eqref{apriorestrho0}. Note that the errors in the two equations above are uniform in $d$.

We know that $\rho_{\gamma^{\rho_0,\delta}}$ is decreasing in $\delta=d\rho_0a$ by the structure of the minimization problem in Lemma \ref{prop:simplfunctsol}. In fact, Lemma \ref{fullintegralbehaviour} and the fourth property in Remark \ref{remarkpropint} show that $\rho_{\gamma^{\rho_0,\delta}}$ decreases to 0 as $d\to\infty$. We therefore have that the equation \eqref{eqnforrho_0d} has a solution for every $\rho_0$ and $\rho$ such that
\begin{equation}
\label{rhoneceqn}
\rho-\rho_{\gamma^{\rho_0,\delta=0}}\leq\rho_0\leq\rho, 
\end{equation}
or, denoting the solution to \eqref{rho0eqn} for given $\rho$ and $d\geq0$ by $\rho_0(d)$, for every $\rho_0(d=0)\leq\rho_0\leq\rho$. Our assumption \eqref{eqnforrho_0d} amounts to plugging $\rho_0(d)$ into the simplified energy \eqref{energyexpression0}. In the next step, we do this for the different regions.\\

\textit{Step 3: $\rho a/T\ll 1$.}
In this step, we prove Theorem \ref{thm:expint}.

\textit{Step 3a.}
In order to be able to use more of the lemmas from Subsection \ref{sec:prelapprox}, we need to show that we can assume that $d$ is bounded.
We use \eqref{ineqF3}, \eqref{energyexpression0} and Lemma \ref{fullintegralbehaviour} to see that for $d\gg1$ and $s=\sqrt{\rho_0a/T}$:
\[
\begin{aligned}
\cF^{\rm sim}(\gamma^{\rho_0,\delta},\alpha^{\rho_0,\delta},\rho_0)&\geq 4\pi a\rho^2+T^{5/2}(I_2(d,8\pi,0,s)-ds^2I_4(d,8\pi,0,s))\\
&+2(\nu-8\pi)\rho aT^{3/2}I_4(d,8\pi,0,s)\\
&+(12\pi-\nu)a(T^{3/2}I_4(d,8\pi,0,s))^2-o(T(\rho a)^{3/2})
\end{aligned}
\]
with errors uniform in $d$. As $d$ increases, $s$ increases, and this expression gets exponentially close to $4\pi a\rho^2$ as $d\to\infty$, and thus it is higher than the value provided by the upper bound \ref{ineqF4} for $d=0$ (see \eqref{enexp00} for a calculation). We can therefore restrict to bounded $d$.

\textit{Step 3b.}
We conclude that the upper bound \eqref{ineqF4} has an error uniform in $d$. We can also apply Lemmas \ref{lem:rhogamma2contribution}, \ref{lem:rho_gamma1moderate},  \ref{lem:lowerorderenergyterm} and \ref{lem:leadingorderenergytermmoderate} to \eqref{energyexpression0} to obtain  
\begin{equation*}
\begin{aligned}
&\cF^{\rm sim}(\gamma^{\rho_0,\delta},\alpha^{\rho_0,\delta},\rho_0)=T^{5/2}f_{\rm{min}}+\left(\frac{\rho_0(d)a}{T}\right)T\rho_{\rm{fc}}(d+8\pi)\\
&\quad\quad-\frac{1}{12\pi}\left(\frac{\rho_0(d)a}{T}\right)^{3/2}T^{5/2}\left((d+16\pi)^{3/2}+d^{3/2}\right)\\
&\quad\quad-d\rho_0(d)a(\rho-\rho_0(d))+\widehat{V}(0)\rho^2-8\pi a\rho_0(d)\rho+\rho_0(d)^2(12\pi a-\widehat{V}(0))\\
&\quad\quad+o(T(\rho a)^{3/2}),
\end{aligned}
\end{equation*}
where
\begin{equation}
\begin{aligned}
\label{rho0expmodtemp}
\rho_0(d)&=\rho-\rho_{\rm{fc}}+\frac{1}{8\pi}\left(\frac{\rho_0(d)a}{T}\right)^{1/2}T^{3/2}\left(\sqrt{d+16\pi}+\sqrt{d}\right)+o(T(\rho a)^{1/2}),
\end{aligned}
\end{equation}
and the errors are uniform in $d$. We conclude that $\rho_0=\rho-\rho_{\rm fc}=:\Delta\rho$ to leading order. Rewriting the expansion in the small parameter $\Delta\rho a/T$, we obtain
\begin{equation}
\label{enexp00}
\begin{aligned}
&\cF^{\rm sim}(\gamma^{\rho_0(d),\delta},\alpha^{\rho_0(d),\delta},\rho_0(d))\\
&=T^{5/2}f_{\rm{min}}+4\pi a \rho^2+(\widehat{V}(0)-4\pi a)\rho_{\rm fc}(2\rho-\rho_{\rm fc})\\
&\quad+T(\Delta\rho a)^{3/2}\left(\frac{1}{24\pi}\right)\left[(\sqrt{d+16\pi}+\sqrt{d})(d+6(8\pi-\nu))-32\pi\sqrt{d+16\pi}\right]\\
&\quad+o\left(T(\rho a)^{3/2}\right).
\end{aligned}
\end{equation}
This can explicitly be minimized in $d\geq0$. The minimum is obtained for $d=2(\nu-8\pi)$, which leads to the expression stated in Theorem \ref{thm:expint}.

\textit{Step 3c.}
The proof is unfinished since the upper and lower bounds \eqref{ineqF3} and \eqref{ineqF4} only hold for $\rho_0$ satisfying \eqref{rhoneceqn}, i.e.\ $\rho_0(d=0)\leq\rho_0\leq\rho$. We need to deal with all other $\rho_0$ as we did in step 1 of the proof of Theorem \ref{thm:grandcancrittemp}: by revising the lower bound \eqref{ineqF3}. We know that any potential minimizer should satisfy
\begin{equation}
\label{lowerbound}
\begin{aligned}
&\cF^{\rm can}(\gamma,\alpha,\rho_0)\\
&\geq\cF^{\rm sim}(\gamma,\alpha,\rho_0)-(E_2+E_3+E_4)(\gamma,\alpha,\rho_0)\\
&\geq\inf_{(\gamma,\alpha)}\mathcal{F}^{\rm s}(\gamma,\alpha,\rho_0)+\widehat{V}(0)\rho^2-8\pi a\rho\rho_0+(12\pi a-\widehat{V}(0))\rho_0^2-O((\rho a)^{5/2})\\
&=T^{5/2}f_{\rm{min}}+8\pi a\rho_0\rho_{\rm fc}-\left(\frac{\rho_0a}{T}\right)^{3/2}T^{5/2}\frac43\sqrt{16\pi}\\
&\hspace{2cm}+\widehat{V}(0)\rho^2-8\pi a\rho_0\rho+\rho_0^2(12\pi a-\widehat{V}(0))-o(T(\rho a)^{3/2}),
\end{aligned}
\end{equation}
where we have used the energy expansion \eqref{energyexpression0} and Lemma \ref{lem:leadingorderenergytermmoderate} for the unrestricted minimizer $(\gamma^{\rho_0,\delta=0},\alpha^{\rho_0,\delta=0})$.
Using $\nu\geq8\pi$, we see that lower bound has a negative derivative for
\[
0\leq\rho_0\leq\Delta\rho+\left(\frac{\Delta\rho a}{T}\right)^{1/2}T^{3/2}\frac{1}{\sqrt{\pi}}+o(T(\rho a)^{1/2}),
\]
which is indeed bigger than
\[
\rho_0(d=0)=\Delta\rho+\left(\frac{\Delta\rho a}{T}\right)^{1/2}T^{3/2}\frac{1}{2\sqrt{\pi}}+o(T(\rho a)^{1/2}).
\] 
Since this lower bound matches our earlier lower bound \eqref{enexp00} at this point and the upper bound \eqref{ineqF4} also holds at this point, we can conclude that it suffices to consider the infimum over $\rho_0(d=0)\leq\rho_0\leq\rho$, which yielded the desired result in step 3b, and proves Theorem \ref{thm:expint}.\\

\textit{Step 4: $\rho a/T\geq O(1)$.}
In this step, we make further preparations for the proof of Theorem \ref{thm:canfreeenexp}.

\textit{Step 4a.}
To prove that we can assume that $d$ is bounded, we use \eqref{ineqF3}, \eqref{energyexpression0}, boundedness of $I_2$ and $I_4$, and Lemma \ref{fullintegralbehaviour} to see that for $d\gg 1$:
\[
\begin{aligned}
\cF^{\rm sim}(\gamma^{\rho_0,\delta},\alpha^{\rho_0,\delta},\rho_0)&\geq 4\pi a\rho^2+F^{(1)}-d(\rho_0a)\rho^{(1)}_\gamma-O((\rho a)^{5/2})\\
&\geq 4\pi a \rho^2+C\min\{(\rho_0a)^{5/2}d^{1/2},a^{-1}(\rho_0a)^2\}-O((\rho a)^{5/2}).
\end{aligned}
\]
with errors uniform in $d$. For $d\gg1$, this is of higher order than $4\pi a\rho^2+O((\rho a)^{5/2})$, which is the value provided by the upper bound \eqref{ineqF4} at $d=0$. We can therefore restrict to bounded $d$.

\textit{Step 4b.}
We again need to establish that it suffices to minimize $\eqref{energyexpression0}$ over $d\geq0$, i.e.\ to exclude $0\leq\rho_0\leq\rho_0(d=0)$ as potential minimizers. To do this, we repeat the lower bound \eqref{lowerbound}. For $\rho a/T\geq O(1)$, the infimum of $\cF^{\rm s}$ is $O((\rho a)^{5/2})$ by Lemmas \ref{lem:rhogamma2contribution} and \ref{lem:lowerorderenergyterm} and boundedness of $I_2$ and $I_4$. We obtain
\[
\cF^{\rm can}(\gamma,\alpha,\rho_0)\geq\widehat{V}(0)\rho^2+(12\pi a-\widehat{V}(0))\rho_0^2-8\pi a\rho\rho_0-O((\rho a)^{5/2}).
\]
Using $\nu\geq8\pi$, we see that this has a negative derivative throughout the region, and as such the minimum can be found at the boundary (up to a lower order error), where it matches the lower bound \eqref{ineqF3} and the upper bound \eqref{ineqF4}, and so we conclude that it suffices to consider the infimum over $\rho_0(d=0)\leq\rho_0\leq\rho$.

\textit{Step 4c.} We would now like to show that the errors in the lower bound \eqref{ineqF3} are $o((\rho a)^{5/2})$, rather than $O((\rho a)^{5/2})$. The above conclusion, Lemma \ref{lem:rhogamma2contribution} and the lower and upper bounds \eqref{ineqF3} and \eqref{ineqF4} imply that any potential minimizer has to satisfy
\[
\rho_\gamma\leq\rho-\rho_0(d=0)=O((\rho a)^{3/2}).
\]
We can now use this to improve the a priori bounds in Proposition \ref{lm:gvg1}, and hence lower the error in the lower bound \eqref{ineqF3} to $o((\rho a)^{5/2})$: we simply repeat the estimates \eqref{estm1}, \eqref{estm2} and \eqref{estimongammV}, and notice that we are able to pick a better $b$ because we know that $\rho_{\gamma}$ is small.\\

\textit{Step 5.}
Combining the steps above, we conclude
\begin{equation*}
\begin{aligned}
F^{\rm{can}}(T,\rho)=\inf_{0\leq d\leq d_0}&[\frac12(\rho_0(d) a)^{5/2}I_1(d,8\pi,0)+T^{5/2}I_2(d,8\pi,0,\sqrt{\rho_0(d) a/T})\\
&-d\rho_0(d)a(\rho-\rho_0(d))\\
&+\widehat{V}(0)\rho^2-8\pi a\rho_0(d)\rho+\rho_0(d)^2(12\pi a-\widehat{V}(0))]\\
&+o\left(\left(\rho a\right)^{5/2}+T(\rho a)^{3/2}\right),
\end{aligned}
\end{equation*}
where
\[
\rho_0(d):=\rho-\frac12(\rho_0(d) a)^{3/2}I_3(d,8\pi,0)-T^{3/2}I_4(d,8\pi,0,\sqrt{\rho_0(d) a/T}).
\]
Here, the errors in $\rho_0(d)$ have been dropped compared to \eqref{rho0eqn} since they can be absorbed in the errors in the energy expression.

To finish the proof of Theorem \ref{thm:canfreeenexp}, we just have to make a few replacements in the minimization problem above. These are

\begin{itemize}
\item replacing $(\rho_0a)^{5/2}I_1(d,8\pi,0)$ by $(\rho a)^{5/2}I_1(d,8\pi,0)$. The error made is $O((\rho a)^{5/2})$ for $\rho a/T\ll 1$, which is acceptable. For $\rho a/T\geq O(1)$, we have $\rho_0(d)=\rho$ to leading order, so we make an error of $o((\rho a)^{5/2})$.
\item replacing the similar term in $\rho_0(d)$.  This is done in a similar way. We absorb the error in the energy expansion.
\item replacing $T^{3/2}I_4(d,8\pi,0,\sqrt{\rho_0(d)a/T})$ by $T^{3/2}I_4(d,8\pi,0,\sqrt{\Delta\rho a/T})$.
 For $\rho a/T\ll 1$, we use \eqref{rho0expmodtemp} to see that this leads to an error that can be absorbed in the energy expansion. For $\rho a/T\geq O(1)$, this term is $O((\rho a)^{3/2})$ and $\rho_0(d)=\rho$ to leading order, so that the error is of lower order and the replacement is justified. 
\end{itemize}
\end{proof}

\begin{proof}[Comment about Corollary \ref{corollary112}]
To obtain the expansions for $\nu\to8\pi$, we use the first two properties in Remark \ref{remarkpropint} and the fact that we can think of the errors as uniform in $d$ to conclude that all relevant contributions to the energy are increasing in $d$. Hence, the minimum is attained at $d=0$ in the limit $\nu\to8\pi$. We also note that only $I_1$ and $I_3$ contribute, and a calculation of the integrals then yields the Lee--Huang--Yang constant.
\end{proof}

\appendix \section{Approximations to integrals}
\label{app:intexp}

\begin{proof}[Proof of Lemma \ref{lm:intexp}]
We make a change of variables to obtain
\bq
\label{someexp93}
\ba
b^{3/2}\int&\ln\left(1-e^{-b\sqrt{(p^2+\delta_0/b)^2+2(p^2+\delta_0/b)}}\right)dp\\
&-\int\ln(1-e^{-b(p^2+\delta_0/b)})dp-b\int(e^{b(p^2+\delta_0/b)}-1)^{-1}dp.
\ea
\eq
Regard $\delta_0/b$ as a fixed parameter and note that the integral has a limit as $b\to0$ by the Monotone Convergence Theorem, which is 
\[
\ba
b^{3/2}\int\left[\frac12\ln\left(1+\frac{2}{p^2+\delta_0/b}\right)-\frac{1}{p^2+\delta_0/b}\right]dp\leq Cb^{3/2}.
\ea
\]
Of course $\delta_0/b$ is not fixed, but the error term in this convergence is uniform in $\delta_0/b$ as long as that quantity is bounded (smaller than 1, say). For $\delta_0/b\geq1$, we expand the logarithm in the first line of \eqref{someexp93} as a Taylor series around $b=0$. Using the Mean Value Theorem and noting that the absolute value of the second derivative attains its maximum at 0, we conclude that the quantity of interest is bounded by
\[
b^2\int\frac{(p^2+\delta_0+1)e^{p^2+\delta_0}-1}{(p^2+\delta_0)(e^{p^2+\delta_0}-1)^2}dp\leq C\delta_0^{-1/2}b^2\leq Cb^{3/2}.
\]
\end{proof}

\begin{proof}[Proof of Lemma \ref{lem:rhogamma2contribution}]
Recall $|\widehat{Vw}(p)|\leq \widehat{Vw}(0)=8\pi a$, so that the integral converges pointwise to the desired expression as $\phi\to0$. We would like to apply the Dominated Convergence Theorem, which leads us to analyse
\[
f(t):=\frac{x+tA}{\sqrt{(x+tA)^2-t^2A^2}}-1,
\]
where $x=p^2+d$, $A=(1+\theta)\sigma$ and $t\in[-1,1]$. This function has the property that $|f(t)|\leq f(1)$ for $t\in[0,1]$, and $|f(t)|\leq f(-1)$ for $t\in[-1,1]$ as long as $x>2A$. 
We therefore dominate the function by replacing $\widehat{Vw}(\phi p)/8\pi a$ by $1$ for $|p|\leq \sqrt{3(1+\theta)\sigma}$, and by $-1$ elsewhere. This function is integrable, and so the Dominated Convergence Theorem gives the desired result. To obtain uniformity, we use continuity in the different parameters.
\end{proof}

\begin{proof}[Proof of Lemma \ref{lem:rho_gamma1moderate}]
\textit{Step 1: first line in statement.\\}
We will write $s=\phi/\sqrt{T}\ll1$. After a change of variables, we need to show that
\[
\begin{aligned}
&T^{3/2}\int\left(e^{\sqrt{(p^2+ds^2)^2+2(p^2+ds^2)(1+\theta)\sigma s^2\frac{\widehat{Vw}(\sqrt{T}p)}{8\pi a}}}-1\right)^{-1} \\
&\qquad \qquad \times \frac{p^2+ds^2+(1+\theta)\sigma s^2\frac{\widehat{Vw}(\sqrt{T}p)}{8\pi a}}{\sqrt{(p^2+ds^2)^2+2(p^2+ds^2)(1+\theta)\sigma s^2\frac{\widehat{Vw}(\sqrt{T}p)}{8\pi a}}}dp\\
&=T^{3/2}\int\left(e^{\sqrt{(p^2+ds^2)^2+2(p^2+ds^2)(1+\theta)\sigma s^2}}-1\right)^{-1} \\
&\qquad \qquad \times \frac{p^2+ds^2+(1+\theta)\sigma s^2}{\sqrt{(p^2+ds^2)^2+2(p^2+ds^2)(1+\theta)\sigma s^2}}dp+o(T^{5/2}a^2(\rho^{1/3}a)^{-3/8}).
\end{aligned}
\]
We define
\[
\ba
f(p,t)&=\left(e^{\sqrt{(p^2+ds^2)^2+2(p^2+ds^2)(1+\theta)\sigma s^2 t}}-1\right)^{-1}\\
&\qquad \qquad \times\frac{p^2+ds^2+(1+\theta)\sigma s^2 t}{\sqrt{(p^2+ds^2)^2+2(p^2+ds^2)(1+\theta)\sigma s^2 t}},
\ea
\]
and calculate its derivative in $t$:
\bq \nn
\begin{aligned}
&\partial_t f(p,t)=
\\&=-\frac14 \sinh^{-2}\left(\frac12\sqrt{(p^2+ds^2)^2+2(p^2+ds^2)(1+\theta)\sigma s^2 t}\right)\\
&\hspace{5cm}\times\frac{(1+\theta)\sigma s^2\left(p^2+ds^2+(1+\theta)\sigma s^2t\right)}{p^2+ds^2+2(1+\theta)\sigma s^2t}\\
&\quad+\left(e^{\sqrt{(p^2+ds^2)^2+2(p^2+ds^2)(1+\theta)\sigma s^2t}}-1\right)^{-1}\\
&\hspace{5cm}\times\frac{(p^2+ds^2)(1+\theta)^2\sigma^2 s^4 t}{\left((p^2+ds^2)^2+2(p^2+ds^2)(1+\theta)\sigma s^2t\right)^{3/2}}
\\&=: F_1(p,t)+F_2(p,t).
\end{aligned}
\eq
We use the Mean Value Theorem to estimate
\begin{equation}
\label{mwtapplication}
\ba 
&|f(p, \widehat{Vw}(\sqrt{T}p)/8\pi a)-f(p,1)| \\
&\hspace{3cm}\leq\left(\sup_{t\in[\widehat{Vw}(\sqrt{T}p)/8\pi a),1]}\left|\partial_t f(p,t)\right|\right)\left|\widehat{Vw}(\sqrt{T}p)/8\pi a-1\right|.
\ea
\end{equation}
Before we estimate this, we make the following two observations:
\begin{enumerate}
\item For $|p|\leq  (\rho^{1/3}a)^{-1/8}$ we have 
\[
|\widehat{Vw}(\sqrt{T}p)/8\pi a-1|\leq C a^{-1}\|\widehat{Vw}''\|_\infty Tp^2\leq CTa^2(\rho^{1/3}a)^{-1/4}.
\]
This also means that $t(p)=\widehat{Vw}(\sqrt{T}p)/8\pi a\geq1/2$ in this region.
\item For $|p|\geq (\rho^{1/3}a)^{-1/8}$, we first note that in general
\[
|\widehat{Vw}(\sqrt{T}p)/8\pi a)-1|\leq 2.
\]
We also have $|p|\geq (\rho^{1/3}a)^{-1/8} \gg 1 \gg 2\sqrt{\sigma} s$, so that 
\begin{align*}
(p^2+ds^2)^2+2(p^2+ds^2)(1+\theta)\sigma s^2 t&\geq \frac12(p^2+ds^2)^2+p^2\left(\frac12p^2-2(1+\theta)\sigma s^2\right)\\
&\quad+ds^2\left(p^2-2 (1+\theta)\sigma s^2\right)\\
&\geq\frac12(p^2+ds^2)^2\geq\frac12p^4.
\end{align*}
\end{enumerate}

Using these estimates, and the fact that $\sinh(x)^{-1}\leq 2(e^x-1)^{-1}$ for $x>0$, we estimate the contribution of $F_1$ to \eqref{mwtapplication} by
\[
\ba
\left\{
	\begin{array}{ll}
		CTa^2(\rho^{1/3}a)^{-1/4}(\frac{1}{e^{|p|\sqrt{(1+\theta)\sigma s^2}/2}-1})^2(1+\theta)\sigma s^2 & \mbox{if } |p|\leq (\rho^{1/3}a)^{-1/8} \\
		C\frac{1}{e^{p^2/(2\sqrt{2})}-1} & \mbox{if } |p|\geq (\rho^{1/3}a)^{-1/8}
	\end{array}
\right.,
\ea
\]
and the contribution from $F_2$ by
\[
\ba
\left\{
	\begin{array}{ll}
		CTa^2(\rho^{1/3}a)^{-1/4}\frac{1}{e^{|p|\sqrt{(1+\theta)\sigma s^2}}-1}&\frac{1}{|p|}\sqrt{(1+\theta)\sigma s^2} \\
&\mbox{if } |p|\leq (\rho^{1/3}a)^{-1/8} \\
		C\frac{1}{e^{p^2/\sqrt{2}}-1} \\
&\mbox{if } |p|\geq (\rho^{1/3}a)^{-1/8}
	\end{array}
\right..
\ea
\]
Integrating \eqref{mwtapplication} amounts to integrating the above contributions, which gives the desired result (this can be seen after a change of variables by noting that the outer integrals decay exponentially fast), and the error is independent of $d$.\\

\textit{Step 2: second line in statement.}
We again write $s=\phi/\sqrt{T}$ to obtain
\begin{equation}
\label{someeqn003}
\ba
T^{3/2}I_4(d,\sigma,\theta,s)-\rho_{\rm fc}&=(2\pi)^{-3}T^{3/2}s^{3}\Big[\int\left(e^{s^2\sqrt{(p^2+d)^2+2(p^2+d)(1+\theta)\sigma}}-1\right)^{-1}\\
&\times\frac{p^2+d+(1+\theta)\sigma}{\sqrt{(p^2+d)^2+2(p^2+d)(1+\theta)\sigma}}-\left(e^{s^2p^2}-1\right)^{-1}\Big]dp.
\ea
\end{equation}
If we can show that this equals
\begin{equation}
\label{someeqn004}
(2\pi)^{-3}T^{3/2}s\int\Big[\frac{p^2+d+(1+\theta)\sigma}{(p^2+d)^2+2(p^2+d)(1+\theta)\sigma}-\frac{1}{p^2}\Big]dp+o\left(T^{3/2}s\right),
\end{equation}
we would obtain the desired result by calculating the integral.

We therefore consider the difference of these two terms, and consider the regions $|p|\leq B$ and $|p|>B$ separately, where $B\gg1$ is chosen in such a way that the integrals over $|p|>B$ of \eqref{someeqn003} and \eqref{someeqn004} are $o(T^{3/2}s)$. Since the latter is a convergent integral, it is clear that this can be done. We will show the same for \eqref{someeqn003} in a moment. 

For $|p|\leq B$, we first apply the Monotone Convergence Theorem to the two terms in \eqref{someeqn003} separately. This shows convergence to the corresponding part of the integral \eqref{someeqn004}.

Employing another change of variables, and writing $b=2(d+(1+\theta)\sigma)$ and $c=d(d+2(1+\theta)\sigma)$, it remains to show that we can pick $B$ such that
\[
\ba
&\int_{|p|>Bs}\Big|\Big[(e^{p^2\sqrt{1+bs^2/p^2+cs^4/p^4}}-1)^{-1}\\
&\hspace{2cm}\times\frac{1+bs^2/2}{\sqrt{1+bs^2/p^2+cs^4/p^4}}-(e^{p^2}-1)^{-1}\Big]\Big|dp=o(s).
\ea
\]
To show this, we apply Taylor's theorem to $bs^2/p^2+cs^4/p^4\ll1$, so that the above expression is bounded by
\[
\ba
&\frac{bs^2}{2}\int\limits_{|p|>Bs}\frac{1}{e^{p^2}-1}dp\\
&\hspace{1cm}+C\int\limits_{|p|>Bs}\frac{e^{\sqrt{2}p^2}(\sqrt{2}p^2+1)-1}{(e^{p^2}-1)^2}\left(b\frac{s^2}{p^2}+c\frac{s^4}{p^4}\right)\left(1+b\frac{s^2}{2}\right)dp,
\ea
\]
where the first term comes from the zeroth-order term, and the other from the derivative.
Seeing that the main contribution from these integrals comes from $p=0$, we conclude that this is bounded by $C(b/B+c/B^3)s$, which indicates that we can indeed pick $B$ large enough to obtain $o(s)$.
\end{proof}

\begin{proof}[Proof of Lemma \ref{lem:rho_gamma1verylow}]
We would like to apply the Dominated Convergence Theorem to the limit $\phi^2=\rho_0a\to0$. We have shown how to bound the fraction in Lemma \ref{lem:rhogamma2contribution} above. The exponential can be bounded in a similar way (i.e.\ by considering $|p|\leq \sqrt{3(1+\theta)\sigma}$ and $|p|> \sqrt{3(1+\theta)\sigma}$ separately) since $\phi^2/T=\rho_0a/T\geq O(1)$ by our assumptions. Uniformity follows by continuity in the different parameters. Another change of variables gives the result stated in the lemma. We obtain uniformity of the error in $d\geq0$ since both sides of the statement are exponentially decaying in $d\gg1$.

\end{proof}

\begin{proof}[Proof of Lemma \ref{lem:lowerorderenergyterm}]
As in the proof of Lemma \ref{lem:rhogamma2contribution}, we regard $t=\widehat{Vw}(\phi p)/8\pi a$ as a parameter taking values in $[-1,1]$, and replace it by 1 for $|p|\leq \sqrt{3(1+\theta)\sigma}$. For other $p$, the function is continuous in $t\in[-1,1]$, and we can maximize it for every $p$. This way, we again obtain a dominating function which is still integrable, so that we can apply the Dominated Convergence Theorem.
\end{proof}

\begin{proof}[Proof of Lemma \ref{lem:leadingorderenergytermmoderate}]
\textit{Step 1: first line in statement.\\}
We will write $s=\phi/\sqrt{T}\ll1$. After a change of variables, our goal is to show that 
\begin{align*} 
&T^{5/2}\int\ln\left(1-e^{-\sqrt{(p^2+ds^2)^2+2(p^2+ds^2)(1+\theta)\sigma s^2\frac{\widehat{Vw}(\sqrt{T}p)}{8\pi a}}}\right)dp\\
&\hspace{1cm}=T^{5/2}\int\ln\left(1-e^{-\sqrt{(p^2+ds^2)^2+2(p^2+ds^2)(1+\theta)\sigma s^2}}\right)dp\\
&\hspace{6cm}+O\left(T^{5/2}\phi^2a^2(\rho^{1/3}a)^{-1/4}\right).
\end{align*}
To this end, we define
\[
f(p,t)=\ln\left(1-e^{-\sqrt{(p^2+ds^2)^2+2(p^2+ds^2)(1+\theta)\sigma s^2 t}}\right).
\]
This function is continuously differentiable in $t$:
\[
\ba
\partial_t f(p,t)&=\left(e^{\sqrt{(p^2+ds^2)^2+2(p^2+ds^2)(1+\theta)\sigma s^2 t}}-1\right)^{-1}\\
&\hspace{3cm}\times\frac{(p^2+ds^2)(1+\theta)\sigma s^2}{\sqrt{(p^2+ds^2)^2+2(p^2+ds^2)(1+\theta)\sigma s^2t}}.
\ea
\]
We use the Mean Value Theorem to estimate this, followed by the two estimates discussed below \eqref{mwtapplication}:
\begin{align*}
&|f(p,\widehat{Vw}(\sqrt{T}p)/8\pi a)-f(p,1)|\\
&\quad\leq\left(\sup_{t\in[\widehat{Vw}(\sqrt{T}p)/8\pi a,1]}\left|\partial_t f(p,s,t)\right|\right)\left|\widehat{Vw}(\sqrt{T}p)/8\pi a-1\right|\\
&\quad\leq\left\{
	\begin{array}{ll}
		CTa^2(\rho^{1/3}a)^{-1/4}\frac{1}{e^{p^2}-1}(1+\theta)\sigma s^2 & \mbox{if } |p|\leq (\rho^{1/3}a)^{-1/8} \\
		C \frac{1}{e^{p^2/\sqrt{2}}-1} & \mbox{if } |p|\geq (\rho^{1/3}a)^{-1/8}
	\end{array}.
\right.
\end{align*}
Integrating over $p$ gives the desired result (note that the outer integral decays exponentially fast), and the error is independent of $d$.\\

\textit{Step 2: second line in statement.\\}
We again write $s=\phi/\sqrt{T}$ to obtain
\begin{equation}
\label{someeqn001}
\ba
&T^{5/2}I_2(d,\sigma,\theta,s)-T^{5/2}f_{\rm{min}}-s^2T\rho_{\rm{fc}}(d+(1+\theta)\sigma)\\
&=(2\pi)^{-3}T^{5/2}s^{3}\int\Big[\ln\left(1-e^{-s^2\sqrt{(p^2+d)^2+2(p^2+d)(1+\theta)\sigma}}\right)\\
&\quad-\ln\left(1-e^{-s^2p^2}\right)-(e^{s^2p^2}-1)^{-1}(d+(1+\theta)\sigma)s^2\Big]dp.
\ea
\end{equation}
Since this expression divided by $T^{5/2}s^3$ is monotone in $s$, we obtain by the Monotone Convergence Theorem that
\begin{equation}
\label{someeqn002}
\ba
&(2\pi)^{-3}T^{5/2}s^{3}\int\Big[\ln\left(\frac{\sqrt{(p^2+d)^2+2(p^2+d)(1+\theta)\sigma}}{p^2}\right)-\frac{d+(1+\theta)\sigma}{p^2}\Big]dp\\
&\qquad +o(T^{5/2}s^3),
\ea
\end{equation}
which gives the desired result.

\end{proof}

\begin{proof}[Proof of Lemma \ref{lem:leadingorderenergytermverylow}]
We want apply the Dominated Convergence Theorem to the limit $\phi^2a=\rho_0a^2\to0$. As in Lemma \ref{lem:rhogamma2contribution}, we regard $t=\widehat{Vw}(\phi p)/8\pi a\in[-1,1]$ as a parameter, which we replace by 0 for $|p|\leq \sqrt{3(1+\theta)\sigma}$. For $|p|> \sqrt{3(1+\theta)\sigma}$, we replace it by $-1$ to obtain a dominating function (also using $s\geq O(1)$). Uniformity in the different parameters follows from continuity in these parameters. Another change of variables gives the desired result. We obtain uniformity of the error in $d\geq0$ since both sides of the statement are exponentially decaying in $d\gg1$. 
\end{proof}

\begin{proof}[Proof of Lemma \ref{lem:errgammastar}]
The basic estimates we will use are:
\begin{equation}
\label{basicestm}
\begin{aligned}
&\left|\rho_0\int\limits_{|p|\leq b}\gamma(p)\widehat{V}(p)dp-\widehat{V}(0)\rho_0\int\limits_{|p|\leq b}\gamma(p)dp\right|\leq Ca^3b^2\rho_0\rho_\gamma\\
&\left|\rho_0\int\limits_{|p|>b}\gamma(p)\widehat{V}(p)dp-\widehat{V}(0)\rho_0\int\limits_{|p|>b}\gamma(p)dp\right|\leq Ca\rho_0\int_{p>b}\gamma(p)dp\\
&\left|\ \iint\limits_{|p|,|q|\leq b}\gamma(p)\widehat{V}(p-q)\gamma(q)dpdq- \widehat{V}(0)\left(\ \int\limits_{|p|\leq b}\gamma(p)dp\right)^2\right|\leq Ca^3b^2\rho^2_\gamma\\
&\left|\ \iint\limits_{|p| \text{or} |q|>b}\gamma(p)\widehat{V}(p-q)\gamma(q)dpdq- \iint\limits_{|p| \text{or} |q|>b}\gamma(p)\widehat{V}(0)\gamma(q)dpdq\right|\\
&\hspace{7cm}\leq Ca\rho_\gamma\int_{|p|>b}\gamma(p)dp, 
\end{aligned}
\end{equation}
which follow from the fact that $\|\widehat{V}\|_\infty\leq 8\pi a$, $\widehat{V}'(0)=0$ and $\|\widehat{V}''\|_\infty\leq Ca^3$. We also need identical versions of the first two estimates for $\widehat{Vw}$, which hold for the same reasons.

We set $b=\sqrt{T}$. By Lemma \ref{lem:rhogamma2contribution}, we have
\begin{equation}
\label{someobserv}
\int_{|p|>\sqrt{T}} \gamma^{\rho_0,\delta}(p)dp=O(\phi^3),
\end{equation}
since the density becomes \eqref{eq:rhogammasimplified} after a change of variables and both terms are of this order (the exponent of the exponential in the second contribution is at least of order 1). This suffices to prove the statement.
\end{proof}

\begin{proof}[Proof of Lemma \ref{lem:errgammastar2}]
Using the estimates in the previous proof, the reader can check that $b=\rho^{1/3}$ suffices, since
\[
\int_{|p|>\rho^{1/3}} \gamma^{\rho_0,\delta}=o((\rho_0a)^{3/2}).
\]
This follows from an application of the Dominated Convergence Theorem to 
\[
(\rho_0a)^{3/2}\int_{|p|>\frac{\rho^{1/3}}{\sqrt{\rho_0a}}}\left(\frac{p^2+d+\frac{\widehat{Vw}(\sqrt{\rho_0a}p)}{a}}{\sqrt{(p^2+d)^2+2(p^2+d)\frac{\widehat{Vw}(\sqrt{\rho_0a}p)}{ a}}}-1\right)dp
\]
as in Lemma \ref{lem:rhogamma2contribution}, and the fact that the other contribution in \eqref{eq:rhogammasimplified} is exponentially small in this region (since $\rho^{1/3}\gg\sqrt{T}$).
\end{proof}

\begin{proof}[Proof of Lemma \ref{prepE1}]
\textit{Step 1.} We start by looking at the first term in \eqref{reff}, which does not involve $\phi^2/T$. After adding absolute values within the integral sign, we employ similar reasoning to Lemma \ref{lem:rhogamma2contribution} to conclude that it is $O(\phi^3)$ as $\phi\to0$. Similar to \eqref{someobserv}, we then have
\[
\int\limits_{|p|>\sqrt{T}} |f(p)|dp\leq C\phi^3,
\]
which was one of our goals.\\

\textit{Step 2.} We now restrict to the case $\phi^2/T\ll1$ and consider the full integral of $f$. Again using that the first term in \eqref{reff} only contributes $O(\phi^3)$, we have that
\begin{equation}
\label{falphaintegral}
\ba
\int f(p)dp&=\phi^3\int\frac{(1+\theta)\sigma\frac{\widehat{Vw}(\phi p)}{8\pi a}}{\sqrt{(p^2+d)^2+2(p^2+d)(1+\theta)\sigma\frac{\widehat{Vw}(\phi p)}{8\pi a}}} \\
&\hspace{4cm}\times\frac{1}{e^{\frac{\phi^2}{T}\sqrt{(p^2+d)^2+2(p^2+d)(1+\theta)\sigma\frac{\widehat{Vw}(\phi p)}{8\pi a}}}-1}dp\\
&\quad+O(\phi^3)\\
&=T \phi \int\frac{(1+\theta)\sigma}{(p^2+d)^2+2(p^2+d)(1+\theta)\sigma}dp+o(T\phi)\\
&=T\phi(1+\theta)\sigma\frac{2\pi^2}{\sqrt{d+2(1+\theta)\sigma}+\sqrt{d}}+o(T\phi).
\ea
\end{equation}
The step before the last requires reasoning similar to Lemma \ref{lem:lowerorderenergyterm}, where the application of the Dominated Convergence Theorem is facilitated by the fact that $(e^x-1)^{-1}\leq x^{-1}$.

An identical argument leads to the estimate that $\int|f|\leq CT\phi$.\\

\textit{Step 3.} For the case $\phi^2/T\geq O(1)$, the second line in \eqref{falphaintegral} combined with the Dominated Convergence Theorem applied as in Lemma \ref{lem:rho_gamma1verylow} leads to the desired conclusion.
\end{proof}

\begin{proof}[Proof of Lemma \ref{fullintegralbehaviour}]
We first analyse the asymptotic behaviour of $F^{(1)}$ as $d\to\infty$. Writing $A(p)=\widehat{Vw}(\phi p)/a$, we expand for $d\gg A=O(1)$:
\[
(p^2+d)\sqrt{1+\frac{2A}{p^2+d}}=p^2+d+A-\frac12\frac{A^2}{p^2+d}+o(A/d).
\]
This tells us that the asymptotic behaviour of $F^{(1)}$ is
\[
(2\pi)^{-3}d\phi^5\frac14\int A^2(p)\frac{1}{p^2(p^2+d)}dp.
\]
Similarly, we can see that the asymptotic behaviour of $-d\phi^2\rho^{(1)}_\gamma$ is
\begin{equation}
\label{rhogamm1beh}
-(2\pi)^{-3}d\phi^5\frac14\int A^2(p)\frac{1}{(p^2+d)^2}dp.
\end{equation}
By our assumptions on the derivative of the potential, there exists a $c$ such that $|\widehat{Vw}(p)|\geq4\pi a$ for $|p|\leq c/a$. Hence, for $d^{1/2}\phi a\leq C$, the two sum of the two contributions above is bounded below by
\[
Cd^{1/2}\phi^5\int\limits_{|p|\leq c(d^{1/2}\phi a)^{-1}}\frac{\widehat{Vw}^2(d^{1/2}\phi p)a^{-2}}{p^2(p^2+1)^2}dp\geq Cd^{1/2}\phi^5,
\]
whereas for $d^{1/2}\phi a\geq C$, it is bounded below by
\[
d^2\phi^5(\phi a)^3\int_{|p|\leq c}\frac{\widehat{Vw}^2(p/a)a^{-2}}{p^2(p^2+d\phi^2a^2)^2}dp\geq Ca^{-1}\phi^4.
\]

To prove the claims about $\rho^{(1)}_\gamma$ we first consider $d\gg1$ and use \eqref{rhogamm1beh} (divided by $d\phi^2$). On the remaining compact $0\leq d\leq C$, we can apply Lemma \ref{lem:rhogamma2contribution}.
\end{proof}

\begin{proof}[Proof of Lemma \ref{lem:errgammastar3}]
Let $\sqrt{T}\ll b\ll \sqrt{T}(\sqrt{T}a)^{-1/8}$. Using \eqref{someobserv}, we first notice that  
\[
\iint_{|p| \text{or} |q|>b}\gamma^{\rho_0,\delta}(p)\widehat{V}(p-q)\gamma^{\rho_0,\delta}(q)dpdq\leq Ca\rho\phi^3=o(T^4a^3).
\]
The same holds for the similar contribution to $E_5$ involving $\widehat{V}(0)$. Using \eqref{propertyofV} and \eqref{gammazero}, we see that the final contribution to the outer region is also $o(T^4a^3)$ since
\[
\int_{|p|>b} \gamma_0=o(T^{3/2}),\hspace{2cm} \int_{|p|>b} p^2\gamma_0=o(T^{5/2}).
\]
We again use \eqref{propertyofV} and estimate the contribution from the inner region by
\[
\begin{aligned}
&C\Big|\ \iint\limits_{|p|,|q|\leq b}\gamma^{\rho_0,\delta}(p)\left(\widehat{V}(p-q)-\widehat{V}(0)-\frac{\Delta\widehat{V}(0)|p-q|^2}{6}\right)\gamma^{\rho_0,\delta}(q)dpdq\Big|\\
&\quad+C\Delta\widehat{V}(0)\rho_{\gamma^{\rho_0,\delta}}\ \int\limits_{|p|\leq b} p^2|\gamma^{\rho_0,\delta}-\gamma_0|(p)dp\\
&\quad+C\Delta\widehat{V}(0)\left(\ \int\limits_{|p|\leq b} |\gamma^{\rho_0,\delta}-\gamma_0|(p)dp\right)\left(\ \int\limits_{|p|\leq b} p^2\gamma_0(p)dp\right).
\end{aligned}
\]
Lemmas \ref{lem:rhogamma2contribution} and \ref{lem:rho_gamma1moderate} (where also the proof of Lemma \ref{lem:rho_gamma1moderate} is important to deal with the absolute value for the middle term) together with the properties of $\gamma_0$ and the properties of the potential pointed out below \eqref{basicestm} imply that this is indeed $o(T^4a^3)$.
\end{proof}

\bibliographystyle{siam}

\begin{thebibliography}{10}

\bibitem{Arnold}
{\sc P. Arnold and G. Moore}, {\em BEC transition temperature of a dilute homogeneous imperfect Bose gas}, Phys. Rev. Lett., 87 (2001), p.~120401.

\bibitem{Andersen-04}
{\sc J.O. Andersen}, {\em Theory of the weakly interacting {B}ose gas}, Rev.
  Mod. Phys., 76 (2004), p.~599.

\bibitem{Cornell-95}
{\sc M.H. Anderson, J.R. Ensher, M.R. Matthews, C.E. Wieman and E.A.
  Cornell}, {\em Observation of {B}ose--{E}instein condensation in a dilute
  atomic vapor}, Science, 269 (1995), p.~198--201.


\bibitem{Baymetal-01}
{\sc G.~Baym, J.-P. Blaizot, M.~Holzmann, F.~Lalo\"e and D.~Vautherin}, {\em
  Bose--{E}instein transition in a dilute interacting gas}, Eur. Phys. J. B, 24
  (2001), p.~107--124.

\bibitem{BijSto-96}
{\sc M.~Bijlsma and H.T.C. Stoof}, {\em Renormalization group theory of the
  three-dimensional dilute {B}ose gas}, Phys. Rev. A, 54 (1996), p.~5085.

\bibitem{Bogoliubov-47b}
{\sc N.N. Bogoliubov}, {\em On the theory of superfluidity}, J. Phys. (USSR),
  11 (1947), p.~23.

\bibitem{CriSol-76}
{\sc R.H. Critchley and A.~Solomon}, {\em A {Variational} {A}pproach to
  {S}uperfluidity}, J. Stat. Phys., 14 (1976), p.~381--393.

\bibitem{Ketterle-95}
{\sc K.B. Davis, M.O. Mewes, M.R. Andrews, N.J. van Druten, D.S. Durfee,
  D.M. Kurn and W.Ketterle}, {\em {B}ose-{E}instein {C}ondensation in a
  {G}as of {S}odium {A}toms}, Phys. Rev. Lett., 75 (1995), p.~3969--3973.


\bibitem{Ensh}
{\sc J.R. Ensher et al.}, {\em Bose-Einstein condensation in a dilute gas: Measurement of energy and ground-state occupation}, Phys. Rev. Lett., 77 (1996), pp.~4984.

\bibitem{ErdSchYau-08}
{\sc L.~Erd{\"{o}}s, B.~Schlein and H.-T. Yau}, {\em {Ground-state energy of a
  low-density Bose gas: A second-order upper bound}}, Phys. Rev. A, 78 (2008),
  p.~053627.

\bibitem{Feynman-53}
{\sc R.P. Feynman}, {\em Atomic {T}heory of the $\lambda$ {T}ransition in
  {H}elium}, Phys. Rev., 91 (1953), p.~1291--1301.

\bibitem{Feynman-53.2}
{\sc R.P. Feynman}, {\em Atomic {T}heory of {L}uquid {H}elium {N}ear {A}bsolute {Z}ero}, Phys. Rev., 91 (1953), p.~1301--1308.

\bibitem{Gaunt}
{\sc A.L. Gaunt et al.}, {\em Bose-Einstein condensation of atoms in a uniform potential}, Phys. Rev. Lett., 110 (2013), p.~200406.

\bibitem{Gerb}
{\sc F. Gerbier et al.}, {\em Critical temperature of a trapped, weakly interacting Bose gas}, Phys. Rev. Lett. 92 (2004), p.~030405.

\bibitem{GiuSei-09}
{\sc A.~Giuliani and R.~Seiringer}, {\em The ground state energy of the weakly
  interacting {B}ose gas at high density}, J. Stat. Phys., 135 (2009),
  p.~915--934.

\bibitem{GKW}
{\sc A.E. Glassgold, A.N. Kaufman and K.M. Watson}, {\em Statistical Mechanics for the Nonideal Bose Gas}, Phys. Rev., 120 (1960), p.~660.

\bibitem{Hua1}
{\sc K. Huang}, {\em Studies in Statistical Mechanics Vol. II}, J. deBoer and G. Uhlenbeck, Eds., North-Holland, 1964.

\bibitem{Hua2}
{\sc K. Huang}, {\em Transition temperature of a uniform imperfect Bose gas}, Phys. Rev. Lett., 83 (1999), p.~3770.

\bibitem{HuaYan-57}
{\sc K.~Huang and C.N. Yang}, {\em Quantum-{M}echanical {M}any-{B}ody {P}roblem with {H}ard-{S}phere {I}nteraction}, Phys.
  Rev., 105 (1957), p.~767--775.

\bibitem{Kash}
{\sc V.A. Kashurnikov, N.V. Prokof'ev and B.V. Svistunov}, {\em Critical temperature shift in weakly interacting Bose gas}, Phys. Rev. Lett., 87 (2001), p.~120402.

\bibitem{LeeYan-58}
{\sc T.~Lee and C.N. Yang}, {\em Low-{T}emperature {B}ehavior of a {D}ilute
  {B}ose {S}ystem of {H}ard {S}pheres i. {E}quilibrium {P}roperties}, Phys.
  Rev., 112 (1958), p.~1419--1429.

\bibitem{LeeHuaYan-57}
{\sc T.D. Lee, K.~Huang and C.N. Yang}, {\em Eigenvalues and
  {E}igenfunctions of a {B}ose {S}ystem of {H}ard {S}pheres and its
  {L}ow-{T}emperature {P}roperties}, Phys. Rev., 106 (1957), p.~1135--1145.

%
%

\bibitem{LieSeiSolYng-05}
{\sc E.H. Lieb, R.~Seiringer, J.P. Solovej and J.Yngvason}, {\em The
  mathematics of the {B}ose gas and its condensation}, Oberwolfach {S}eminars,
  Birkh{\"a}user, 2005.


\bibitem{LieSeiYng-05}
{\sc E.H. Lieb, R.~Seiringer and J.~Yngvason}, {\em {Justification of
  $c$-Number Substitutions in Bosonic Hamiltonians}}, Phys. Rev. Lett., 94
  (2005), p.~080401.


\bibitem{NapReuSol1-15}
{\sc M.~Napi\'orkowski, R.~Reuvers and J.P. Solovej}, {\em Bogoliubov free
  energy functional {I}. {E}xistence of minimizers and phase diagram}, Arch. Ration. Mech. Anal. (in press), ArXiv:1511.05935 (2015).

\bibitem{NapReuSol-17}
{\sc M.~Napi\'orkowski, R.~Reuvers and J.P. Solovej}, {\em Calculation of the Critical Temperature of a Dilute Bose Gas in the Bogoliubov Approximation}, arXiv:1706.01822 (2017).

\bibitem{NhoLan-04}
{\sc K.~Nho and D.P. Landau}, {\em Bose--{E}instein {C}ondensation
  {T}emperature of a {H}omogeneous {W}eakly {I}nteracting {B}ose {G}as: {PIMC}
  study}, Phys. Rev. A, 70 (2004), p.~053614.

\bibitem{Sei-08}
{\sc R.~Seiringer}, {\em Free {E}nergy of a {D}ilute {B}ose {G}as: {L}ower
  {B}ound}, Commun. Math. Phys., 279 (2008), p.~595--636.

\bibitem{SeiUel-09}
{\sc R.~Seiringer and D.~Ueltschi}, {\em Rigorous upper bound on the critical
  temperature of dilute bose gases}, Phys. Rev. B, 80 (2009), p.~014502.

\bibitem{Smith}
{\sc R.P. Smith}, {\em Effects of Interactions on Bose-Einstein Condensation}, Universal Themes of Bose-Einstein Condensation, eds. N.P. Proukasis, D.W. Snoke and P.B. Littlewood, Cambridge University Press, 2017.

\bibitem{Sm2}
{\sc R.P. Smith et al.}, {\em Effects of interactions on the critical temperature of a trapped Bose gas}, Phys. Rev. Lett., 106 (2011), p.~250403.

\bibitem{Solovej-06}
{\sc J.P. Solovej}, {\em {Upper bounds to the ground state energies of the
  one- and two-component charged Bose gases}}, Commun. Math. Phys., 266 (2006),
  p.~797--818.

\bibitem{Toyoda-82}
{\sc T.~Toyoda}, {\em A microscopic theory of the lambda transition}, Ann.
  Phys., 141 (1982), pp.~154--178.

\bibitem{YauYin-09}
{\sc H.-T. Yau and J.~Yin}, {\em The second order upper bound for the ground
  energy of a {B}ose gas}, J. Stat. Phys., 136 (2009), p.~453--503.

\bibitem{Yin-10}
{\sc J.~Yin}, {\em Free {E}nergies of {D}ilute {B}ose {G}ases: {U}pper
  {B}ound}, J. Stat. Phys., 141 (2010), p.~683--726.

\end{thebibliography}

\end{document}